\pdfoutput=1
\pdfoutput=1
\PassOptionsToPackage{hyperfootnotes=false,hypertexnames=false}{hyperref}
\ifdefined\CCSUBMISSION
  \documentclass[preprint,12pt]{elsarticle}
\else
  \documentclass[11pt]{article}
\fi

\usepackage{paper-preamble}
\ifdefined\CCSUBMISSION
  \biboptions{sort&compress}
\fi

\renewcommand{\LH}[1]{\mbox{\ttfamily\nolinkurl{#1}}}
\renewcommand{\LHrng}[3]{\mbox{\ttfamily\nolinkurl{#1#2-#3}}}

\title{Descent Before Hardness: Orbit-Gap Obstructions in Exact Certification}
\ifdefined\CCSUBMISSION
  \author[mcgill]{Tristan Simas}
  \ead{tristan.simas@mail.mcgill.ca}
  \affiliation[mcgill]{
    organization={McGill University},
    city={Montreal},
    state={Quebec},
    country={Canada}}
  \journal{Journal of Computer and System Sciences}
\else
  \author{Tristan Simas\\
  McGill University, Montreal, Quebec, Canada\\
  \texttt{tristan.simas@mail.mcgill.ca}}
\fi
\date{\today}

\begin{document}

\ifdefined\CCSUBMISSION
  \begin{frontmatter}
  \begin{abstract}
  Tractability tests are often computed from input syntax: support-graph treewidth, local coefficient patterns, backdoor tests, or action-count bounds. Before such a test can be lower-bounded or made algorithmic, it must define a predicate on the exact-certification problem itself. Equivalent presentations must receive the same verdict.

The semantic object is the correctness quotient, whose classes are states with the same correct outputs. Correctness-preserving presentation moves generate closure orbits. A target that changes inside one closure orbit has an orbit gap and fails descent. Exact closure-invariant classification is possible exactly when the positive and negative orbit hulls are disjoint; the positive hull is then the least exact classifier, and computable orbit representatives make the classifier algorithmic.

The results separate three layers. The descent layer gives orbit-gap obstructions for raw local syntax, raw action and coordinate counts, and raw support-graph predicates. The post-descent complexity layer applies ordinary reductions to descended objects: graph-predicate lower bounds transfer through action-gap graph extraction, and \textsc{Action-Gap-Treewidth} is NP-complete when the width bound is part of the input. The certification layer asks whether a proxy descends: for split proxies $b\wedge\varphi(z)$, SAT reduces to non-descent and UNSAT reduces to descent. Positive regimes use quotient-preserving normalizations or catalogues before model checking; bounded quotient size, bounded full Gaifman treewidth of the constructed quotient, sparse unary-gap certificates, and strict-margin perturbation balls give explicit cost bounds after quotient construction.

  \end{abstract}
  \end{frontmatter}
\else
  \maketitle
  \begin{abstract}
  Tractability tests are often computed from input syntax: support-graph treewidth, local coefficient patterns, backdoor tests, or action-count bounds. Before such a test can be lower-bounded or made algorithmic, it must define a predicate on the exact-certification problem itself. Equivalent presentations must receive the same verdict.

The semantic object is the correctness quotient, whose classes are states with the same correct outputs. Correctness-preserving presentation moves generate closure orbits. A target that changes inside one closure orbit has an orbit gap and fails descent. Exact closure-invariant classification is possible exactly when the positive and negative orbit hulls are disjoint; the positive hull is then the least exact classifier, and computable orbit representatives make the classifier algorithmic.

The results separate three layers. The descent layer gives orbit-gap obstructions for raw local syntax, raw action and coordinate counts, and raw support-graph predicates. The post-descent complexity layer applies ordinary reductions to descended objects: graph-predicate lower bounds transfer through action-gap graph extraction, and \textsc{Action-Gap-Treewidth} is NP-complete when the width bound is part of the input. The certification layer asks whether a proxy descends: for split proxies $b\wedge\varphi(z)$, SAT reduces to non-descent and UNSAT reduces to descent. Positive regimes use quotient-preserving normalizations or catalogues before model checking; bounded quotient size, bounded full Gaifman treewidth of the constructed quotient, sparse unary-gap certificates, and strict-margin perturbation balls give explicit cost bounds after quotient construction.

  \end{abstract}

\fi

\section{Introduction}

Consider states $x=(x_0,x_1,x_2)\in\{0,1\}^3$ and two actions with payoffs
\[
U(a,x)=2x_0x_1,\qquad U(b,x)=0.
\]
Action $a$ is uniquely optimal exactly when $x_0=x_1=1$; otherwise both actions tie. Thus the states split into two classes according to their optimal-action sets: the two states with $x_0=x_1=1$, and the remaining six states. A coordinate matters for exact certification when flipping it can move a state across this partition. Here $x_0$ and $x_1$ matter, while $x_2$ never does.

Now add the same state-dependent term to every action, for instance $3x_1x_2$. Every action comparison is unchanged, so every optimal-action set and every exact-certification question is unchanged. A statistic computed from the raw payoff expression can still change: the support graph now contains the pair $\{1,2\}$ even though that pair cancels from all action comparisons. The treewidth case study raises raw support-graph treewidth from $1$ to $n-1$ while leaving exact certification fixed. Raw treewidth reports on the encoding rather than the problem.

The corresponding descended parameter is the action-gap graph: put an edge on a coordinate pair exactly when some action gap has a nonzero mixed difference on that pair. The action-independent term cancels from every action gap. The graph is invariant under the harmless shift above, and it gives the concrete complexity theorem: \textsc{Action-Gap-Treewidth} is NP-complete when the width bound is part of the input (Theorem~\ref{ext:action-gap-treewidth-np-complete}). For fixed width, a supplied graph solver with a uniform polynomial step bound can be run on the extracted action-gap graph to solve the descended graph predicate.\leanmeta{\LHrng{FR}{455}{457}}

Rice's theorem states that every nontrivial extensional property of partial recursive functions is undecidable~\cite{rice1953classes}. Borchert and Stephan adapted the same pattern to circuit counting properties, proving that every nontrivial counting property of circuits is UP-hard~\cite{borchert1996looking}. Hemaspaandra and Rothe strengthened the lower bound to a constant-ambiguity regime. They also identified a P-constructibly bi-infinite setting in which SPP-hardness holds~\cite{hemaspaandra2000second}. The Rice analogs assume the predicate already descends to the semantic object and ask whether deciding it from syntax is hard. A proposed proxy may fail before that question arises: it may not descend to the semantic object. The correctness-preserving closure relation (Proposition~\ref{prop:closure-operations-preserve-certification}), the hull-separation converse (Theorem~\ref{thm:exact-classification-hull-separation}), and the named proxy witnesses give the corresponding descent test.

An implementation analogy has the same shape. A solver library often routes inputs through a registry of available handlers. A hand-maintained registry stores two objects: the declared handler family $\mathcal H$ and a cache $R$. Adding a handler can update $\mathcal H$ while leaving $R$ stale. Automatic subclass registration makes $R$ a derived function of the handler family, so the routing table changes exactly when the handler family changes.

\begin{figure}[H]
\centering
\small
{\setlength{\fboxsep}{3.5pt}%
\newcommand{\RegistryCodeBlock}[1]{%
  {\ttfamily\small
  \fcolorbox{black!25}{black!3}{%
    \begin{minipage}[c][7.9\baselineskip][c]{0.98\linewidth}
    \raggedright #1
    \end{minipage}}}}
\begin{tabular}{@{}c@{\hspace{0.02\linewidth}}c@{}}
\begin{minipage}[t]{0.47\linewidth}
\centering
\textbf{Manual registry}\par\vspace{0.35em}
\RegistryCodeBlock{%
class Handler: pass\\[0.45\baselineskip]
class A(Handler): pass\\
registry = \{A\}\\[0.45\baselineskip]
class B(Handler): pass\\
\# B exists\\
\# registry = \{A\}}
\par\vspace{0.45em}
\(\text{manual: }(\mathcal H,R)\mapsto(\mathcal H\cup\{B\},R)\)
\end{minipage}
&
\begin{minipage}[t]{0.47\linewidth}
\centering
\textbf{Derived registry}\par\vspace{0.35em}
\RegistryCodeBlock{%
registry = set()\\[0.45\baselineskip]
class Handler:\\
\hspace*{0.8em}def \_\_init\_subclass\_\_(cls):\\
\hspace*{1.6em}registry.add(cls)\\[0.45\baselineskip]
class A(Handler): pass\\
class B(Handler): pass\\
\# registry = \{A, B\}}
\par\vspace{0.45em}
\(\text{derived: }R=\{C:C\in\mathcal H\}\)
\end{minipage}
\end{tabular}}
\caption{A hand-maintained registry stores a cache $R$ separately from the handler family $\mathcal H$; adding $B$ can leave the cache stale. Subclass registration derives $R$ from class creation, so routing data follows the semantic source. Tractability proxies face the same descent requirement: routing statistics must be computed from the exact-certification object.}
\label{fig:registry-descent}
\end{figure}

A presentation is syntax for an exact-certification problem, such as a weighted Boolean payoff expression. Some syntax changes preserve the problem. The action-independent shift above is one example. The family connected by such changes is a closure orbit.

A direct test reads the encoded syntax without first solving the problem or constructing its quotient. Treewidth of the raw support graph, a coefficient ratio, a support statistic, and an action-profile count are direct tests. Whether a coordinate is relevant is not direct in this sense; it depends on the optimal-action sets of the decoded problem. A direct test used to predict that a specialized algorithm applies is a tractability proxy.

A proxy must first define a property of the exact-certification problem. If two presentations encode the same exact problem but the proxy gives different verdicts, then the proxy is a presentation statistic. Such a same-orbit disagreement is an orbit gap.

Any problem with rules determines a correctness relation $R(s,a)$ recording which outputs are correct at each state. Equality of the correctness sets partitions the state space. Exact relevance certification asks which coordinates recover that partition. Correctness-preserving presentation moves must leave this quotient fixed, so any correct classifier must give one verdict on each closure orbit.

Same-orbit disagreements are the obstruction. Exact classification is possible exactly when no closure orbit contains both a positive and a negative example. Equivalently, the positive and negative orbit hulls are disjoint. When they are disjoint, the closure hull is the least exact classifier. Once a proxy descends, standard reductions, FPT algorithms, kernels, and class lower bounds analyze the descended task. Before descent, the proxy records an encoding statistic.

The direct local proxies studied below are raw-syntax tests before normalization. Adding a quotient-preserving normalization step, such as action-gap graph extraction or profile compression, moves the predicate to quotient-level data. In the treewidth example, the raw support graph fails the descent test. The action-gap graph is the descended object to which the bounded-treewidth algorithm applies.

\paragraph{Terminology}
A correctness relation $R(s,a)$ specifies the correct outputs at each state. The correctness quotient identifies states with the same correctness set. A closure orbit is the class generated by correctness-preserving presentation moves. An orbit gap is a same-orbit positive/negative pair for a proposed target. A closure-sound predicate is constant on primitive closure steps, equivalently on closure orbits. A finite-structural predicate is a closure-sound, polynomial-time, fixed local-pattern test on raw binary-pairwise syntax: the pattern vocabulary, radius, coefficient alphabet, and action-label budget are fixed before the input is read.

\paragraph{Main results}
The results have three layers. Structural claims prove descent criteria from the correctness quotient and closure laws. Transferred complexity claims compose a descended extraction with external theorems such as graph treewidth hardness or Courcelle-style model checking. Conditional algorithmic claims bound evaluation after a quotient, catalogue, certificate, or strict-margin neighborhood has been supplied.

\begin{enumerate}
\item \emph{Descent test.} Exact relevance for any state-indexed correctness relation reduces to quotient recovery (Theorem~\ref{cor:relational-output-transfer}). Orbit gaps are exactly the obstruction to closure-invariant exact classification, and hull separation is the positive criterion (Proposition~\ref{prop:orbit-gap-completeness}, Theorem~\ref{thm:exact-classification-hull-separation}). For a finite proxy language, descent is equivalent to absence of an orbit-gap counterexample, and existence of an exact closure-invariant classifier for a code is equivalent to descent.\leanmeta{\LHrng{FR}{436}{443}}
\item \emph{Concrete obstructions.} The dominant-pair, margin, ghost-action, and offset-normalization targets have raw binary-pairwise orbit gaps; hence no closure-sound finite-structural exact recognizer decides those targets (Table~\ref{tab:obstruction-families-summary}, Theorem~\ref{thm:closure-sound-package-no-go}). Every raw action-count threshold at least two has an unanchored duplication-based orbit gap (Theorem~\ref{thm:raw-action-count-duplication-orbit-gap}); the constant-utility witness has no relevant coordinates on either side.\leanmeta{\LHrng{FR}{452}{454}} Every raw coordinate-count threshold at least two has an unanchored irrelevant-coordinate orbit gap (Theorem~\ref{thm:raw-coordinate-count-irrelevant-coordinate-orbit-gap}).\leanmeta{\LHrng{FR}{485}{490}} More generally, every raw support-graph predicate that separates two same-size graphs has an orbit gap under a common action-independent state term.\leanmeta{\LHrng{FR}{444}{451}}
\item \emph{Descended complexity.} The action-gap graph passes descent. For any uniform finite graph predicate family $P$, the graph-realization map transfers lower bounds for $P$ to its action-gap version; any many-one lower bound for $P$ composes with this map (Theorem~\ref{thm:graph-predicate-hardness-transfer}). Recognizing whether the action-gap graph has treewidth at most a supplied bound is NP-complete by this realization map and the standard graph-treewidth theorem (Theorem~\ref{ext:action-gap-treewidth-np-complete}). Any fixed-width graph solver with a polynomial step bound transfers to the action-gap graph after extraction.\leanmeta{\LHrng{FR}{455}{457}, \LHrng{FR}{503}{508}}
\item \emph{Descent-certification hardness.} For a formula/circuit split proxy with semantic variables $z$ and a presentation-only bit $b$, the code $C_\varphi(z,b)=b\wedge\varphi(z)$ has an orbit gap exactly when $\varphi$ is satisfiable. Thus SAT reduces to split-proxy non-descent, and UNSAT reduces to split-proxy descent (Theorem~\ref{thm:split-proxy-descent-certification-hardness}). These are hardness directions for certifying descent or non-descent of the encoded proxy; membership bounds depend on the chosen formula or circuit representation.\leanmeta{\LHrng{FR}{491}{502}}
\item \emph{Quotient-level routing.} The operational route is to construct a closure-compatible quotient or catalogue, then model-check a fixed predicate on that descended object (Theorems~\ref{thm:finite-orbit-catalogue-classifiers} and~\ref{thm:post-descent-quotient-mso-evaluation}). Bounded quotient-state and action-profile sorts give explicit mixed-fragment cost bounds. Bounded quotient size and bounded full Gaifman treewidth of the constructed quotient give fixed-MSO evaluation after quotient construction (Theorems~\ref{thm:bounded-quotient-size-mso-tractability} and~\ref{thm:bounded-quotient-treewidth-mso-tractability}). Sparse unary-gap certificates instantiate the quotient-state bound with $2^r$ states, and strict margins preserve quotient-MSO verdicts under uniform perturbation (Theorem~\ref{thm:sparse-unary-gap-quotient-pipeline}, Corollary~\ref{cor:strict-margin-quotient-mso-stability}).\leanmeta{\LHrng{FR}{458}{484}}
\end{enumerate}

The descent layer identifies when a proxy is a predicate of the exact-certification problem. The complexity layer starts after descent: the action-gap graph carries ordinary graph-parameter hardness and graph-algorithm transfer. The certification layer asks whether an encoded proxy descends at all. Broader lower bounds for quotient-recovery tasks require a fixed descended task and representation class.

\section{A Running Example: The Obstruction in Miniature}\label{sec:obstruction-miniature}

Adding an action-independent pair term preserves the exact-certification problem because it leaves all action comparisons unchanged, but it changes a local syntactic statistic.

The core mechanism needs only three Boolean coordinates and two actions. Let $x=(x_0,x_1,x_2)\in\{0,1\}^3$ and define
\[
U(a,x)=2x_0x_1,\qquad U(b,x)=0.
\]
The optimizer relation is controlled by the action gap $U(a,x)-U(b,x)=2x_0x_1$: action $a$ is uniquely optimal exactly when $x_0=x_1=1$, and otherwise both actions are optimal.
Thus the optimizer quotient has two classes: the two states with $x_0=x_1=1$, and the remaining six states. Coordinates $x_0$ and $x_1$ matter because changing either can cross between those two quotient classes. Coordinate $x_2$ does not matter because changing it never changes the optimizer set.

Now add the same pair-supported state term to both actions:
\[
V(c,x)=U(c,x)+3x_1x_2\qquad(c\in\{a,b\}).
\]
The operation is a statewise positive affine reparameterization with unit scale and action-independent shift. It preserves every action comparison, since
\[
V(a,x)-V(b,x)=U(a,x)-U(b,x).
\]
Consequently $U$ and $V$ define the same exact-certification problem: the optimizer set, the decision quotient, sufficient coordinate sets, and relevant coordinates are unchanged.
The operation has changed scores, but only by a state-dependent amount common to every action. Exact certification depends on action comparisons, so the quotient sees no change.

\begin{table}[h]
\centering
\small
\setlength{\tabcolsep}{4pt}
\renewcommand{\arraystretch}{1.12}
\renewcommand{\tabularxcolumn}[1]{m{#1}}
\begin{tabularx}{\linewidth}{@{}>{\centering\arraybackslash}m{0.15\linewidth}>{\centering\arraybackslash}m{0.11\linewidth}>{\centering\arraybackslash}m{0.11\linewidth}>{\centering\arraybackslash}m{0.11\linewidth}>{\centering\arraybackslash}m{0.11\linewidth}>{\centering\arraybackslash}X@{}}
\toprule
$x_0x_1x_2$ & $U(a,x)$ & $U(b,x)$ & $V(a,x)$ & $V(b,x)$ & $\Opt_U(x)=\Opt_V(x)$ \\
\midrule
000 & 0 & 0 & 0 & 0 & $\{a,b\}$ \\
\midrule
001 & 0 & 0 & 0 & 0 & $\{a,b\}$ \\
\midrule
010 & 0 & 0 & 0 & 0 & $\{a,b\}$ \\
\midrule
011 & 0 & 0 & 3 & 3 & $\{a,b\}$ \\
\midrule
100 & 0 & 0 & 0 & 0 & $\{a,b\}$ \\
\midrule
101 & 0 & 0 & 0 & 0 & $\{a,b\}$ \\
\midrule
110 & 2 & 0 & 2 & 0 & $\{a\}$ \\
\midrule
111 & 2 & 0 & 5 & 3 & $\{a\}$ \\
\bottomrule
\end{tabularx}
\caption{Truth table for the miniature orbit-gap witness. The quotient classes are unchanged by the action-independent shift.}
\label{tab:miniature-orbit-gap-truth-table}
\end{table}

Table~\ref{tab:miniature-orbit-gap-truth-table} displays the same two quotient classes for both presentations: $\{110,111\}$, where $a$ is uniquely optimal, and the remaining six states, where both actions are optimal. Every later obstruction repeats this pattern: the quotient classes stay fixed while a raw local statistic changes.

The raw local syntax changes while the optimizer quotient stays fixed. In $U$, the unique nonzero mixed-difference contribution is the pair $\{0,1\}$ for action $a$, with magnitude $2$. In $V$, the action-independent term contributes a larger mixed-difference on $\{1,2\}$ for both actions, with magnitude $3$. Thus the local predicate
\[
\text{``the unique dominant pair/action is }(\{0,1\},a)\text{''}
\]
holds for $U$ and fails for $V$.

The example separates quotient-level data from presentation-level statistics. Optimizer sets, sufficient coordinate sets, and relevant coordinates are identical for $U$ and $V$. The largest raw pair coefficient is different. Closure-soundness requires a tractability proxy to follow the quotient-level data, not the presentation statistic.

Action-gap normalization closes this particular gap: if the statistic is computed from $U(a,\cdot)-U(b,\cdot)$, not raw coefficients, the common state term cancels. The descent test identifies which preprocessing makes a parameter sound for a closure move. Later witnesses test other raw statistics and normalization choices.

\begin{figure}[t]
\centering
\begin{tikzpicture}[
  font=\small,
  presentation/.style={draw, rounded corners=1.5pt, fill=black!3,
    minimum width=1.65cm, minimum height=0.74cm, align=center},
  semantic/.style={draw, rounded corners=1.5pt, fill=black!8,
    minimum width=3.45cm, minimum height=0.72cm, align=center},
  target/.style={draw, minimum width=2.05cm, minimum height=0.62cm,
    align=center},
  forced/.style={draw, rounded corners=1.5pt, fill=black!5,
    minimum width=3.1cm, minimum height=0.58cm, align=center},
  note/.style={align=center},
  every path/.style={line width=0.45pt}
]
\node[presentation] (U) at (-3.25,1.75) {presentation\\$U$};
\node[presentation] (V) at (3.25,1.75) {presentation\\$V$};
\draw[<->] (U) -- node[above, note] {closure step\\$U\sim_{\mathrm{cl}}V$} (V);

\node[semantic] (S) at (0,0.48) {one exact-certification object};
\draw[->] (U) -- (S);
\draw[->] (V) -- (S);

\node[forced] (P) at (0,-0.42) {invariance forces $P(U)=P(V)$};
\draw[->] (S) -- (P);

\node[target] (QU) at (-3.25,-1.42) {$Q(U)=\mathsf{true}$};
\node[target] (QV) at (3.25,-1.42) {$Q(V)=\mathsf{false}$};
\draw[->, dashed] (U) -- (QU);
\draw[->, dashed] (V) -- (QV);
\draw[-] (QU.east) -- node[above, note] {raw target splits} (QV.west);

\node[note] at (0,-2.22) {$\Rightarrow$ no closure-invariant predicate can exactly decide $Q$};
\end{tikzpicture}
\caption{The orbit-gap template. A closure-invariant classifier must be constant on $U$ and $V$, while a raw local target predicate distinguishes them.}
\label{fig:orbit-gap-template}
\end{figure}

The same $U$ and $V$ reappear as the dominant-pair row in Table~\ref{tab:obstruction-families-summary}.

\section{Correctness Quotients and Stable Presentation Tests}\label{sec:quotient-preconditions}

Raw local tests can fail before any hardness question is posed. The running example shows the failure: a correctness-preserving action-independent shift leaves all action gaps fixed while changing a raw local coefficient statistic. States are therefore grouped by their correct outputs, and presentation moves are admissible only when they preserve that grouping.

\subsection{Anchored and Existential Targets}

\begin{definition}[Anchored and existential targets]\label{def:anchored-existential-targets}
An \emph{anchored} target fixes the coordinates, actions, or local template named by the predicate. An \emph{existential} target quantifies over possible placements of the same kind of witness.
\end{definition}

Anchored targets are witness-local; existential targets are witness-global. The dominant-pair, ghost-action, and offset-signature rows in Section~\ref{sec:finite-structural-obstructions} use anchored templates. Margin-boundedness is global over the displayed raw coefficient vocabulary: it compares the largest unary coefficient magnitude with the largest pair mixed-difference magnitude in the chosen pairwise presentation. Raw action-count thresholds are unanchored and fail descent by action duplication. Raw coordinate-count thresholds are unanchored and fail descent by irrelevant-coordinate extension. The quotient-MSO predicates in Section~\ref{sec:successor-invariants} quantify over states, coordinates, or action-profile classes only after the correctness quotient has been constructed.

\subsection{Standing Semantics and Quotient Objects}

For a decision problem $\mathcal D=(S,A,U)$, write
\[
\Opt(s)=\{a\in A : U(a,s) \text{ is maximal among all actions at } s\}.
\]
The optimizer quotient is the equivalence relation $s\sim_{\Opt}s'$ iff $\Opt(s)=\Opt(s')$.

\begin{definition}[Correctness relation and correctness quotient]\label{def:correctness-relation}
A \emph{correctness relation} is a state-indexed binary relation $R(s,a)$ on states and outputs. Write
\[
\operatorname{Corr}_R(s)=\{a:R(s,a)\}.
\]
Two states are \emph{correctness-equivalent}, written $s\sim_R s'$, when $\operatorname{Corr}_R(s)=\operatorname{Corr}_R(s')$. The quotient by $\sim_R$ is the \emph{correctness quotient}.
\end{definition}

\begin{definition}[Weighted Boolean Instance]\label{def:weighted-boolean-instance}
A \emph{weighted Boolean instance} is a finite coordinate set $I$, Boolean state space $\{0,1\}^I$, a finite action set $A$, and a utility
\[
U:A\times\{0,1\}^I\to \mathbb Q
\]
or to another ordered coefficient field used by the encoding. A \emph{binary pairwise instance} is a weighted Boolean instance whose utility is given by constant, unary, and pair-interaction coefficients.
\end{definition}

The binary-pairwise witnesses use rational coefficients to give exact countable syntax. The quotient and hull results use only the induced correctness relation. Real-valued, stochastic, continuous, or partially observed models enter through the declared admissibility rule $R(s,a)$; once that relation is fixed, the same quotient and descent test apply.

\begin{definition}[Tractability Proxy and Tractability Classifier]\label{def:tractability-proxy-classifier}
Let $D$ be a set of presentations encoding exact-certification problems. A parameter or predicate $p$ on $D$ is \emph{presentation-computable} when it is computed from the encoded syntax of $U$ without first constructing the correctness quotient or solving the certification problem. A \emph{tractability proxy} is such a parameter or predicate used to state that a specialized algorithm applies. A \emph{tractability classifier} assigns a tractability verdict for a fixed tractability notion on the underlying exact-certification problem.
\end{definition}

Presentation-computable proxies are evaluated before quotient canonization. Raw support-graph treewidth, coefficient ratios, local support statistics, backdoor tests, and raw action counts are typical examples. A proxy becomes exact only after its induced verdict descends to the correctness quotient; on closure-closed domains this forces constancy on closure orbits.

The quotient-recovery viewpoint extends beyond optimizer-set equivalence. The universal transfer theorem below shows that every state-indexed correctness relation induces the same quotient-recovery problem through a uniform totalization, so the closure-orbit and hull-separation machinery applies without modification.\leanmeta{\LHrng{FR}{223}{231}, \LHrng{FR}{340}{341}}

\leanmetapending{\LHrng{FR}{30}{35}, \LH{FR37}, \LH{FR51}, \LHrng{FR}{176}{177}}
\begin{proposition}[Quotient-Dependence Facts]\label{prop:depends-only-on-opt}\label{prop:zero-distortion-refines-quotient}
Sufficiency and relevance depend only on the decision quotient relation. In particular, two problems on the same state space with the same quotient relation have the same sufficient coordinate sets, minimal sufficient coordinate sets, relevant coordinates, irrelevant coordinates, quotient cardinality in finite settings, and all statistics computed from these data. A zero-distortion summary map has each fiber contained in one optimizer-quotient class and therefore needs at least one summary symbol per quotient class.
\end{proposition}

\begin{proof}[Proof sketch]
Sufficiency is refinement of the quotient relation, and relevance is failure of sufficiency after one-coordinate erasure. The zero-distortion clause is the same refinement condition applied to the summary fibers.
\end{proof}

\leanmetapending{\LH{FR352}}
\begin{proposition}[Minimal Sufficient Sets Are Canonical on Product Spaces]\label{prop:minimal-sufficient-canonical}
Assume $S=\prod_{j\in J}S_j$, with decidable equality on each coordinate domain, and let $I$ be a minimal sufficient set. Then $i\in I$ if and only if $i$ is relevant. Consequently, whenever a minimal sufficient set exists, it is unique.
\end{proposition}

\begin{proof}[Proof sketch]
Sufficient sets are closed under finite intersection: if $A,B$ are sufficient and $s,s'$ agree on $A\cap B$, the coordinatewise splice taking $A$-coordinates from $s$ and the remaining coordinates from $s'$ lies in $\prod_j S_j$ and connects $s$ to $s'$ through $A$- and $B$-agreement. Thus a minimal sufficient set cannot contain an irrelevant coordinate. Conversely, every relevant coordinate lies in every sufficient set, since omitting it would make $\mathrm{univ}\setminus\{i\}$ sufficient. Hence the minimal sufficient set is exactly the relevant-coordinate set.
\end{proof}

Write $\srank(\mathcal D)$ for the number of relevant coordinates. In the Boolean product regime, if the minimal sufficient set exists and $m$ is the number of quotient classes, then $m\le 2^{\srank(\mathcal D)}$ and every zero-distortion summary needs at least $m$ symbols.\leanmeta{\LHrng{FR}{353}{354}} Conversely, finite quotient size is unbounded: for every $m$ there is a finite decision problem whose optimizer quotient has size exactly $m$.\leanmeta{\LH{FR24}, \LH{FR38}} The universal class of all decision problems is invariant under action and coordinate-induced state relabeling but is kernel-universal, so relabeling invariance alone does not isolate exact tractability classification.\leanmeta{\LHrng{FR}{21}{23}} Exact classification needs both benign invariance under names and an expressivity restriction preventing arbitrary quotient geometry.

\subsection{Closure Laws and Direct Local Predicates}\label{sec:closure-preconditions}

Closure-soundness blocks predicates that change under harmless presentation moves. Finite-structural definability restricts attention to direct local routing tests evaluated before quotient solving.

\leanmetapending{\LHrng{FR}{15}{20}, \LHrng{FR}{55}{60}, \LHrng{FR}{83}{85}}
\begin{proposition}[Closure Operations Preserve Exact Certification]\label{prop:closure-operations-preserve-certification}
Action relabeling, coordinate relabeling, statewise positive affine reparameterization, action duplication, state duplication, and binary irrelevant-coordinate extension preserve the exact-certification problem. Relabelings use the inverse label map; action duplication projects a duplicate label to its source action; state duplication projects a copied state to its source state; irrelevant-coordinate extension projects away the new coordinate. They preserve sufficient coordinate sets and relevant coordinates under transport, and an added irrelevant coordinate is irrelevant. On sparse binary-pairwise encodings, the listed transports are polynomial-time computable.\label{rem:encoding-complexity}
\end{proposition}

\begin{proof}[Proof sketch]
Relabelings act by bijections. A statewise positive affine transport $V(a,s)=\alpha(s)+\beta(s)U(a,s)$ with $\beta(s)>0$ preserves the action order at each state. Duplication projects to the original profile or state. Irrelevant-coordinate extension lifts old witnesses and changes no correctness set by flipping the new coordinate.
\end{proof}

The obstruction witnesses use one small closure step. If $\alpha(x_i,x_j)$ is a pair-supported state term independent of the action and $V(a,x)=U(a,x)+\alpha(x_i,x_j)$ for every action $a$, then $V$ is a statewise positive-affine reparameterization with scale $1$. All action gaps are unchanged, so optimizer sets, sufficient-coordinate families, and relevant coordinates are unchanged.\label{rem:action-independent-pair-shift}

\begin{definition}[Closure-Closed Domain]\label{def:closure-closed-domain}
A domain $D$ of presentations is \emph{closure-closed} when every declared closure step whose source lies in $D$ and whose target remains in the ambient syntax also has its target in $D$.
\end{definition}

\begin{definition}[Closure-Sound Predicate]\label{def:closure-sound-predicate}
Let $D$ be a domain equipped with domain-relative closure laws. An instance predicate $Q$ on $D$ is \emph{closure-sound} when it is invariant under every within-domain primitive closure step. A class of predicates is closure-sound when each of its members is closure-sound.
\end{definition}

For a single predicate on a fixed domain, closure-soundness is closure-law invariance: invariance under primitive closure steps, equivalently constancy on closure orbits by Lemma~\ref{lem:closure-law-invariance-iff-orbit-constancy}. Orbit membership may still be harder than executing a witnessed closure step; the no-go results use explicit same-orbit witnesses.\label{rem:orbit-membership-cost}

\begin{definition}[Finite-Structural Predicate]\label{def:finite-structural-predicate}
A finite-structural predicate on a binary pairwise domain is a closure-sound predicate $Q$ satisfying three additional conditions:
\begin{enumerate}
\item $Q$ is decidable in polynomial time from the instance encoding;
\item a polynomial-time extractor builds a finite raw pairwise syntax $X(U)$ from the presentation;
\item the verdict is determined by a fixed finite menu of rooted local patterns in $X(U)$.
\end{enumerate}
The pattern data fix a radius, a vertex bound, an action-label budget, and a finite coefficient alphabet independently of the input.
\end{definition}

Equivalently, there are finite witness and forbidden pattern families $\mathcal W,\mathcal F$ such that $Q(U)$ is determined by occurrence or non-occurrence of those fixed patterns in bounded-radius neighborhoods of the extracted syntax. The supplement expands this occurrence relation as first-order formulas over $X(U)$.\label{rem:finite-structural-scope} Closure-soundness is the descent requirement used by the no-go theorem. The three additional clauses locate the direct raw-local proxy language: efficient syntax extraction plus a fixed bounded-pattern test. Growing vocabularies, quotient canonization, fixed-point iteration, and profile compression belong to the quotient-level layer.\label{rem:bounded-pattern-action-count-sensitivity} The Lean declarations use the name \texttt{Admissible\allowbreak Normalization\allowbreak Predicate} for this finite-structural interface.\leanmeta{\LH{FR146}}

\leanmetapending{\LHrng{FR}{197}{199}}
\begin{theorem}[Correctness Implies Closure-Orbit Invariance]\label{thm:correct-classifier-forces-invariance}
Let $C$ assign tractability verdicts on a closure-closed represented problem family with polynomial-time-computable closure transports. Fix the family-level tractability predicate ``the induced exact-certification task is polynomial-time decidable as a function of its represented input or state size.'' If $C$ correctly predicts that tractability status for the underlying exact-certification problem represented by each presentation, then $C$ is constant on closure orbits.
\end{theorem}

\begin{proof}
Each primitive closure step gives polynomial-time reductions in both directions between the represented exact-certification task families: relabelings use inverse relabelings, positive-affine transports preserve optimizer sets pointwise, duplication uses projection and lifting, and irrelevant-coordinate extension uses coordinate projection and lift. Polynomial-time decidability is therefore invariant along each closure edge and along every finite closure chain. A correct classifier must report that invariant family-level tractability status.
\end{proof}

\leanmetapending{\LH{FR248}, \LHrng{FR}{259}{262}, \LHrng{FR}{316}{336}}
\begin{theorem}[Compute-Cost Closure-Orbit Invariance]\label{thm:optimizer-computation-classifier-forces-invariance}
The same orbit agreement holds for classifiers that correctly predict polynomial-time solvability of optimizer computation, deterministic optimizer-set payload output, correctness search over an external output class with preserved correctness relation, or search over representation-relative output objects with explicit polynomial-time output transports.
\end{theorem}

\begin{proof}[Proof sketch]
For one primitive closure edge, pull the queried state back, solve in the source presentation, and push the output through the action or output transport. The stated preservation hypotheses give correctness, and the inverse edge gives the reverse reduction. Composing along a finite orbit chain preserves polynomial overhead.
\end{proof}

\leanmetapending{\LHrng{FR}{206}{209}}
\begin{proposition}[Bounded Distinct Action Profiles Compress to Bounded Actions]\label{prop:distinct-profile-compression}
For a binary pairwise instance $U$, let $d(U)$ be the number of distinct action utility profiles. There is a compressed instance $U^{\mathrm{prof}}$ with exactly $d(U)$ actions such that sufficient coordinate sets and relevant coordinates agree between $U$ and $U^{\mathrm{prof}}$. Thus bounded-actions polynomiality transfers to the bounded-distinct-profile subcase when $d(U)\le k$.
\end{proposition}

\begin{proof}[Proof sketch]
Quotient the action labels by equality of full utility profiles. Optimizer equality depends on which profiles are optimal, not on how many labels realize a profile.
\end{proof}

\section{Universal Exact-Semantics Reduction}\label{sec:exact-semantics-foundation}

Every computational problem that specifies which outputs are correct at which states has a correctness relation $R(s,a)$. The relation determines the correctness quotient, and sufficiency, relevance, and orbit-gap obstructions factor through that quotient. Exact means exact agreement with the declared relation, even when the relation encodes thresholds, randomized outputs, distributional constraints, horizons, or failure states.

\leanmetapending{\LHrng{FR}{223}{231}, \LHrng{FR}{340}{341}}
\begin{theorem}[Universal Correctness Transfer]\label{cor:boolean-payload-transfer}\label{cor:predicate-transfer}\label{cor:set-valued-payload-transfer}\label{cor:totalized-set-valued-payload-transfer}\label{cor:relational-output-transfer}
Let $R(s,a)$ be any state-indexed correctness relation. Exact sufficiency and relevance for $R$ reduce to exact correctness certification for a totalized decision problem with a strict utility gap: correct outputs receive the allowed value, incorrect outputs receive the blocked value, and a failure token is optimal exactly when the correctness set is empty. Consequently, exact output, deterministic payload, search, and approximation guarantees all induce the same correctness-quotient recovery problem. For finite represented state and output data, the construction is explicit; a polynomial-time Turing reduction requires polynomial-time access to membership in $R$ and to the represented correctness sets.
\end{theorem}

\begin{table}[h]
\centering
\small
\setlength{\tabcolsep}{4pt}
\renewcommand{\arraystretch}{1.08}
\renewcommand{\tabularxcolumn}[1]{m{#1}}
\begin{tabularx}{\linewidth}{@{}>{\raggedright\arraybackslash}m{0.22\linewidth}>{\raggedright\arraybackslash}m{0.46\linewidth}>{\raggedright\arraybackslash}X@{}}
\toprule
\textbf{Guarantee form} & \textbf{Correctness relation} & \textbf{Correctness set} \\
\midrule
Exact output & $R(s,a)$ iff $a$ is the designated output & designated outputs \\
\midrule
Payload & $R(s,a)$ iff $a=\phi(s)$ & $\{\phi(s)\}$ \\
\midrule
Search & $R(s,a)$ iff $a\in X_{\mathrm{adm}}(s)$ & admissible witnesses \\
\midrule
Approximation & $R(s,a)$ iff $U(a,s)\ge \max U(\cdot,s)-\varepsilon$ & $\varepsilon$-optimizers \\
\bottomrule
\end{tabularx}
\caption{Representative guarantee forms as state-indexed correctness relations.}
\label{tab:guarantee-form-correctness-relations}
\end{table}

\begin{proof}
Write $F(s)=\operatorname{Corr}_R(s)$. Add a failure token $\bot$. If $F(s)$ is nonempty, give exactly the outputs in $F(s)$ the allowed score and give all other outputs, including $\bot$, a lower score. If $F(s)$ is empty, make $\bot$ uniquely allowed. The induced optimizer set is $F(s)$ when $F(s)\ne\emptyset$ and $\{\bot\}$ otherwise, an injective encoding of correctness sets. Hence equality of correctness sets is equivalent to equality of optimizer sets in the induced decision problem. Sufficiency and relevance are the same quotient-recovery conditions for both relations.
\end{proof}

PAC certificates, regret or risk bounds, randomized-output guarantees, regression thresholds, and finite-horizon control guarantees fit the same construction once they specify which outputs are admissible at each state.

\leanmetapending{\LHrng{FR}{342}{343}, \LH{FR346}}
\begin{example}[Approximation correctness relation]\label{ex:worked-approximation-admissibility}
Approximate admissibility with tolerance $\varepsilon$ declares
\[
R(s,a)\quad\Longleftrightarrow\quad W(a,s)\ge \max_b W(b,s)-\varepsilon.
\]
The correctness set is the $\varepsilon$-optimizer set. Theorem~\ref{cor:relational-output-transfer} totalizes this relation exactly. Section~\ref{sec:approximation-stability} treats the different question of when a nearby utility presentation defines the same quotient.
\end{example}

\leanmetapending{\LHrng{FR}{197}{199}, \LHrng{FR}{232}{244}, \LHrng{FR}{252}{256}, \LHrng{FR}{268}{315}, \LHrng{FR}{302}{303}, \LH{FR345}, \LH{FR351}}
\begin{corollary}[Correctness Quotient Universality and Coordinate Presentations]\label{cor:correctness-quotient-universality}\label{cor:universal-scope-rigorously-specified-problems}\label{prop:every-exact-specification-admits-coordinate-presentation}\label{prop:finite-specifications-admit-coordinate-presentations}\label{prop:countable-specifications-admit-countable-boolean-presentations}\label{prop:finite-specifications-admit-low-dimensional-boolean-presentations}\label{cor:named-guarantee-semantics-transfer}\label{prop:semantic-claims-depend-only-on-output-equivalence}\label{prop:semantically-extensional-claims-factor-through-the-quotient}\label{cor:every-state-equivalence-relation-is-realizable-as-exact-output-semantics}
Every fixed correctness relation determines a canonical correctness quotient. Exact sufficiency is recovery of its classes, and relevance is failure of that recovery after erasing one coordinate. Every correctness relation admits a coordinate presentation by taking the entire state as one coordinate; finite and countable encodable state spaces admit Boolean coordinate presentations. If two semantics induce the same equality relation on correctness sets, then they induce the same sufficient-coordinate family and relevant-coordinate set. Conversely, every state equivalence relation is realized by taking outputs to be quotient classes and assigning each state its own class as a singleton correctness set.
\end{corollary}

A Boolean coordinate presentation only chooses Boolean coordinates for states. A binary-pairwise presentation additionally requires each action utility to be encoded by constant, unary, and pair-interaction coefficients. The universal transfer gives a coordinate presentation for every finite correctness relation. The binary-pairwise no-go is stronger: the obstruction already appears inside the smaller pairwise syntax.

\section{Orbit Gaps, Hull Separation, and Exact Classification}\label{sec:frontier-program}

Arbitrary labeling kernels already arise as optimizer quotients, so quotient shape alone cannot serve as a boundary principle in the unconstrained setting. Correctness supplies the invariant that remains, namely closure-orbit agreement under closure operations induced by exact correctness.

\paragraph{The realizability barrier}\label{sec:realizability-barrier}
In the unconstrained setting, quotient shape alone, without a restricted representation class or computable catalogue, cannot serve as a boundary principle for exact tractability classification. For any map $\phi:S\to T$, take the action set to be $T$ and define
\[
U(a,s)=
\begin{cases}
1 & \text{if } a=\phi(s),\\
0 & \text{otherwise.}
\end{cases}
\]
Then $\Opt(s)=\{\phi(s)\}$, so $\Opt(s)=\Opt(s')$ exactly when $\phi(s)=\phi(s')$. Every equivalence relation is obtained by taking $\phi$ to be its quotient map $S\to S/{\approx}$; finite quotients of any prescribed size are realized in the same way.\leanmeta{\LHrng{FR}{7}{8}, \LHrng{FR}{39}{40}, \LHrng{FR}{45}{47}} Exact descent therefore needs structure that survives correctness-preserving transports.

\subsection*{Closure Operators and Exact Classification}

For a target predicate $Q$, define its closure hull by
\[
\Hull(Q)(U) := \exists V\,\bigl(V \sim_{\mathrm{cl}} U \land Q(V)\bigr)
\]
where $\sim_{\mathrm{cl}}$ denotes closure equivalence. If one closure orbit contains both a positive and a negative domain point, no closure-invariant classifier separates them. Hull separation is the converse criterion.

\leanmetapending{\LH{FR362}}
\begin{lemma}[Primitive-Law Invariance Equals Orbit Invariance]\label{lem:closure-law-invariance-iff-orbit-constancy}
For any instance predicate $P$, invariance under each primitive closure law is equivalent to constancy on closure orbits.
\end{lemma}

\begin{proof}[Proof sketch]
Induct along a finite closure-step chain. The converse follows because every primitive closure step lies inside one closure orbit.
\end{proof}

\leanmetapending{\LHrng{FR}{201}{204}, \LHrng{FR}{210}{213}, \LHrng{FR}{367}{375}}
\begin{proposition}[Orbit-Gap Descent Criterion]\label{prop:abstract-fiber-gap-descent}\label{prop:orbit-gap-completeness}\label{cor:orbit-gap-completeness-domain}
A representation-level predicate descends along a semantic map exactly when it is constant on the map's fibers. In the closure-orbit setting, exact classification of a target on a closure-closed domain by closure-law-invariant predicates fails exactly when the target has an orbit-gap witness: two domain points in one closure orbit with different target values.
\end{proposition}

\begin{proof}[Proof sketch]
Fiber constancy makes the descended predicate well-defined by choosing any representative. Failure of fiber constancy is exactly a same-fiber positive/negative pair. Lemma~\ref{lem:closure-law-invariance-iff-orbit-constancy} specializes the criterion to closure orbits.
\end{proof}

For a finite proxy language, descent certification is the same fiber test internalized at the level of codes. A code descends exactly when no pair of closure-equivalent presentations receives opposite values. Such a counterexample blocks every exact closure-invariant classifier for that code, while exact closure-invariant classification exists precisely when the code descends.\leanmeta{\LHrng{FR}{436}{443}}

\leanmetapending{\LH{FR355}}
\begin{proposition}[Hull Is a Closure Operator; Closure-Invariant Predicates Are Its Fixed Points]\label{prop:closurehull-fixed-points}
The orbit-saturation map $\Hull$ is extensive, monotone, and idempotent. Its fixed points are exactly the closure-law-invariant predicates.\leanmeta{\LH{FR355}}
\end{proposition}

\begin{proof}[Proof sketch]
Extensivity, monotonicity, and idempotence follow respectively from reflexivity, implication, and transitivity of closure equivalence. The fixed-point claim says exactly that saturating by closure orbits adds nothing, which is equivalent to closure-law invariance.
\end{proof}

\leanmetapending{\LHrng{FR}{356}{357}}
\begin{theorem}[Exact Classification Equals Hull Separation]\label{thm:exact-classification-hull-separation}
Let $D$ be a closure-closed domain and let $Q$ be a target predicate on $D$. Then $Q$ admits an exact characterization on $D$ by a closure-law-invariant predicate if and only if the positive and negative orbit saturations are disjoint:
\[
\Hull(D \cap Q) \cap \Hull(D \cap \neg Q) = \varnothing.
\]
\end{theorem}

\begin{proof}
If the hulls overlap, some closure orbit meets $D\cap Q$ and $D\cap\neg Q$, so every closure-invariant predicate gives one value where exact classification requires two. Conversely, if the hulls are disjoint, the predicate $P(S)\Leftrightarrow S\in\Hull(D\cap Q)$ is closure-invariant by transitivity of closure equivalence. It agrees with $Q$ on positive domain points by reflexivity, and it cannot hold on a negative domain point without placing that point in both hulls.
\end{proof}

\leanmetapending{\LH{FR358}}
\begin{corollary}[Least Exact Closure-Invariant Classifier]\label{cor:least-exact-classifier}
Let $D$ be a closure-closed domain and let $Q$ have no orbit gaps on $D$. Then $\Hull(D \cap Q)$ is the least closure-law-invariant predicate that classifies $Q$ exactly on $D$: it is correct on $D$, and every other exact closure-law-invariant classifier on $D$ contains it.
\end{corollary}

\begin{proof}[Proof sketch]
Correctness on $D$ follows from orbit-gap freedom. Any exact closure-law-invariant classifier contains $D\cap Q$ and hence its closure saturation $\Hull(D\cap Q)$, giving minimality.
\end{proof}

\leanmetapending{\LHrng{FR}{376}{383}}
\begin{theorem}[Computable Orbit Catalogues Make Hull Classifiers Algorithmic]\label{thm:finite-orbit-catalogue-classifiers}
Let $D$ be a closure-closed domain and suppose there is a computable representative map
\[
c:D\to C
\]
into a set $C$ with decidable equality such that, for all $U,V\in D$,
\[
U\sim_{\mathrm{cl}}V \quad\Longleftrightarrow\quad c(U)=c(V).
\]
Then orbit membership on $D$ is decidable by representative comparison. For every decidable predicate $R$ on $C$, the pulled-back classifier
\[
P_R(U)\quad\Longleftrightarrow\quad R(c(U))
\]
is closure-law invariant. Conversely, every closure-law-invariant predicate on $D$ factors through $c$ on the image $c(D)$; values on $C\setminus c(D)$ may be chosen arbitrarily. If a target $Q$ factors through $c$ and the representative predicate is decidable, then $Q$ has an exact decidable closure-invariant classifier.
\end{theorem}

\begin{proof}[Proof sketch]
The equivalence $U\sim_{\mathrm{cl}}V\Longleftrightarrow c(U)=c(V)$ makes $c$ a complete orbit invariant. Pullbacks of predicates on $C$ are therefore constant on orbits. If a predicate is already constant on orbits, the value assigned to each image point $c(U)\in c(D)$ is independent of the representative $U$; extending that representative predicate to the rest of $C$ gives the stated factorization.
\end{proof}

\leanmetapending{\LHrng{FR}{384}{389}}
\begin{corollary}[Catalogue Classifiers Are Compositional]\label{cor:catalogue-classifiers-compositional}
On a domain with a complete orbit catalogue, descended predicates are closed under finite Boolean combination. In particular, if $P$ and $Q$ factor through the representative map $c$, then so do $\neg P$, $P\wedge Q$, and $P\vee Q$. Equivalently, if $P$ and $Q$ are orbit-invariant because they factor through the catalogue, then their complement, conjunction, and disjunction are orbit-invariant for the same reason.
\end{corollary}

\begin{proof}[Proof sketch]
Apply the corresponding Boolean operation to the representative predicates on $C$.
\end{proof}

\leanmetapending{\LHrng{FR}{252}{256}}
\begin{corollary}[Domain Restriction Helps Only by Removing Orbit Gaps]\label{cor:domain-restriction-only-removes-orbit-gaps}
Let $D$ be a closure-closed domain, and let $Q$ be a target predicate on $D$ such that correctness of a classifier for $Q$ forces closure-orbit agreement on $D$. Then $D$ admits a correct classifier for $Q$ if and only if $Q$ has no orbit-gap witness inside $D$. Equivalently, restricting the domain avoids the no-go only by eliminating all orbit gaps of $Q$ on that restricted domain.
\end{corollary}

Polynomial-time closure transports execute witnessed moves. A computable representative catalogue is stronger: it decides orbit membership. Domain restriction removes an obstruction only when it removes the relevant same-orbit positive/negative pair.

\subsection*{Where Hull Separation Holds}

Hull separation holds in three recurring regimes. Normalized unary-gap domains remove decision-relevant mixed differences before the obstruction families can form. Bounded distinct-profile compression quotients accidental action multiplicity before classification. Strict-margin neighborhoods (Proposition~\ref{prop:global-approximation-stability}) preserve the optimizer quotient under controlled perturbation. Correctness-quotient MSO predicates (Definition~\ref{def:orbit-quotient-mso-candidate}) are closure-invariant by construction because they speak about quotient data, not raw syntax.

\leanmetapending{\LHrng{FR}{390}{393}}
\begin{theorem}[Normalized Unary-Gap Catalogues Give a Positive Regime]\label{thm:normalized-unary-gap-positive-regime}
Let $D$ be a closure-closed domain whose representations admit a normalized unary action-gap catalogue: a representative map
\[
\nu:D\to G
\]
recording the action-dependent unary gap profile after action-independent state terms have been removed, with
\[
U\sim_{\mathrm{cl}} V \quad\Longleftrightarrow\quad \nu(U)=\nu(V)
\qquad (U,V\in D).
\]
If a target $Q$ factors through $\nu$, then $Q$ is closure-invariant on $D$ and has no orbit gap on $D$. Equivalently, the positive and negative closure hulls of $Q$ are disjoint. If $Q(U)\leftrightarrow R(\nu(U))$ for a representative-level predicate $R$ on $G$, then the pullback classifier $R\circ\nu$ is exact and closure-invariant.
\end{theorem}

\begin{proof}[Proof sketch]
The catalogue equation makes $\nu$ a complete orbit representative. Factorization through $\nu$ gives orbit invariance, so no same-orbit positive/negative pair exists; Theorem~\ref{thm:exact-classification-hull-separation} gives separated hulls. A minimal two-action example has gap $2x_0-1$, so no mixed action-gap difference remains after subtracting action-independent state terms.
\end{proof}

\leanmetapending{\LHrng{FR}{206}{209}}
\begin{theorem}[Bounded Distinct-Profile Compression Gives a Positive Regime]\label{thm:bounded-distinct-profile-positive-regime}
Fix $k\in\mathbb N$, and let $D_k$ be a closure-closed domain of binary pairwise instances such that each $U\in D_k$ has at most $k$ distinct action utility profiles. Let $\rho(U)=U^{\mathrm{prof}}$ be the profile-compressed instance from Proposition~\ref{prop:distinct-profile-compression}, obtained by identifying actions with identical utility profiles.

Let $Q_k$ be any closure-law-invariant predicate on the compressed bounded-action image $\rho(D_k)$. Define
\[
Q(U) \quad\Longleftrightarrow\quad Q_k(\rho(U)).
\]
Then $Q$ is closure-law invariant on $D_k$. If $Q_k$ exactly classifies a target on $\rho(D_k)$ and is decidable in polynomial time for fixed $k$, then $Q$ exactly classifies the pulled-back target on $D_k$ and is decidable in polynomial time. Equivalently, hull separation on the bounded-action compressed image lifts to hull separation on $D_k$ for every target that factors through $\rho$.
\end{theorem}

\begin{proof}[Proof sketch]
Proposition~\ref{prop:distinct-profile-compression} preserves sufficiency and relevance while replacing action labels by utility profiles. Relabeling, action duplication, state duplication, irrelevant-coordinate extension, and statewise positive affine transport preserve profile equality, so closure-equivalent inputs compress to closure-equivalent bounded-action inputs. Exactness and polynomial-time decidability then lift from the compressed image.
\end{proof}

\leanmetapending{\LH{FR185}}
\begin{corollary}[Orbit-Gap Template]\label{cor:orbit-gap-template}
Let $Q$ be an instance predicate. Suppose there exist instances $U,V$ such that $U$ and $V$ lie in the same closure orbit, $Q(U)$ holds, and $Q(V)$ fails. Then no closure-law-invariant predicate $P$ can satisfy
\[
P(W) \iff Q(W)
\qquad\text{for all instances $W$.}
\]
\end{corollary}

\begin{proof}[Proof sketch]
Since $U$ and $V$ lie in the same closure orbit, closure-law invariance gives $P(U) \iff P(V)$ by induction on the closure steps connecting them (Lemma~\ref{lem:closure-law-invariance-iff-orbit-constancy}). Exact agreement with $Q$ would then imply $Q(U) \iff Q(V)$, contradicting the assumed orbit gap.
\end{proof}

\section{Hardness After Descent and Raw-Syntax Orbit Gaps}\label{sec:finite-structural-obstructions}

Two questions appear on the same binary-pairwise vocabulary. The first comes after descent: once the action-gap graph is the descended object, recognizing its treewidth is the ordinary graph-theoretic hardness problem. The second comes before hardness: raw local predicates that split a closure orbit fail to define predicates on the exact-certification problem.

\subsection{Descended Treewidth Hardness}\label{sec:treewidth-case-study}
Let the raw support graph have an edge $\{i,j\}$ whenever some action has a nonzero pair coefficient on $\{i,j\}$. For the miniature witness of Section~\ref{sec:obstruction-miniature}, the raw pair supports are
\[
E_{\rm raw}(U)=\{\{0,1\}\},
\qquad
E_{\rm raw}(V)=\{\{0,1\},\{1,2\}\}.
\]
The raw graph construction fails closure-soundness: the closure move has inserted an edge. By contrast, the action-gap polynomial is unchanged,
\[
U(a,\cdot)-U(b,\cdot)=V(a,\cdot)-V(b,\cdot),
\]
so the support graph built from action-gap mixed differences has edge set $\{\{0,1\}\}$ on both presentations. A treewidth or kernelization theorem can then be applied to that descended graph parameter, not to the raw support graph.

Fix a tree $T$ on $n$ coordinates and choose a pairwise weighted instance whose nonzero raw pair coefficients are positive on exactly the edges of $T$, so the raw support graph has treewidth $1$ and the action gaps use pair terms only on the edges of $T$. A tree-decomposition dynamic program applies to the resulting tree interaction graph. Now add to every action the same state term
\[
\alpha(x)=\sum_{0\le i<j<n} x_i x_j .
\]
The operation is a representation-preserving positive-affine closure step with scale $1$: every action comparison, optimizer set, sufficient-coordinate family, and relevant-coordinate set is unchanged. Because the old tree-edge coefficients are positive and the added common coefficient is positive on every pair, no pair coefficient cancels. The raw support graph has become the complete graph $K_n$, whose treewidth is $n-1$. Thus raw treewidth can move from $1$ to $n-1$ inside one closure orbit, so it cannot be an exact tractability proxy. The action-gap support graph is unchanged by the same move, because the common state term cancels in $U(a,\cdot)-U(b,\cdot)$. Constructing this graph is polynomial time in the sparse binary-pairwise encoding: subtract the pair coefficients of one action from another and insert exactly the nonzero mixed-difference edges. The FPT theorem applies after descent: build the action-gap graph, verify its treewidth, and run the bounded-treewidth dynamic program on that normalized interaction structure.\leanmeta{\LHrng{FR}{444}{448}, \LHrng{FR}{455}{457}}

Raw support-graph predicates fail generally. For any predicate on simple graphs, if two same-size graphs $G$ and $H$ receive different verdicts, the common-term slices carrying $G$ and $H$ lie in one closure orbit, both have empty action-gap graphs, and the induced raw support-graph predicate flips. Hence no closure-invariant classifier decides that raw support predicate exactly.\leanmeta{\LHrng{FR}{444}{451}}

\leanmetapending{\LHrng{FR}{394}{396}}
\begin{proposition}[Action-Gap Graphs Realize Arbitrary Treewidth Instances]\label{prop:action-gap-treewidth-realizes-graph-treewidth}
For every finite graph $G$ on $n$ vertices and every width bound $w$, there is a two-action binary pairwise weighted instance $U_G$ whose action-gap interaction graph is exactly $G$. Consequently,
\[
\operatorname{tw}\bigl(G_{\mathrm{gap}}(U_G)\bigr)\le w
\quad\Longleftrightarrow\quad
\operatorname{tw}(G)\le w .
\]
\end{proposition}

\begin{proof}[Proof sketch]
Use two actions. Let the false action have utility $0$. Let the true action contain the pair term $x_ix_j$ exactly for the edges $\{i,j\}$ of $G$, with one fixed orientation $i<j$. The mixed difference of the action gap on $\{i,j\}$ is then $1$ exactly when $\{i,j\}$ is an edge of $G$, and it is $0$ otherwise. Hence the action-gap interaction graph is $G$ itself, and the treewidth equivalence follows by rewriting along this graph equality.
\end{proof}

\leanmetapending{\LHrng{FR}{491}{492}, \LHrng{FR}{503}{508}}
\begin{theorem}[Graph Predicate Hardness Transfers to Action-Gap Predicates]\label{thm:graph-predicate-hardness-transfer}
Let $P=(P_n)_{n\in\mathbb N}$ be a uniform finite graph predicate family: each $P_n$ is a predicate on simple graphs with vertex set $\{0,\ldots,n-1\}$ under a fixed finite graph encoding, and any downstream reduction to $P$ uses that encoding. Define the action-gap version of $P$ on two-action binary-pairwise instances by
\[
P^{\mathrm{gap}}(U) \quad\Longleftrightarrow\quad
P_n\bigl(G_{\mathrm{gap}}(U)\bigr).
\]
Here $n=|I(U)|$, and $G_{\mathrm{gap}}(U)$ is encoded, after the canonical coordinate relabeling, as a simple graph on $\{0,\ldots,n-1\}$ using the same fixed graph encoding as $P_n$.
The graph-realization map $G\mapsto U_G$ is a many-one reduction core from $P$ to $P^{\mathrm{gap}}$: for every finite graph $G$,
\[
P^{\mathrm{gap}}(U_G) \quad\Longleftrightarrow\quad P(G).
\]
Consequently, every many-one reduction to $P$ composes with $G\mapsto U_G$ to give a many-one reduction to $P^{\mathrm{gap}}$.
\end{theorem}

\begin{proof}
Proposition~\ref{prop:action-gap-treewidth-realizes-graph-treewidth} gives $G_{\mathrm{gap}}(U_G)=G$. Rewriting the predicate $P_n$ along this graph equality gives $P^{\mathrm{gap}}(U_G)\Longleftrightarrow P(G)$. The reduction-core composition statement is ordinary composition of many-one reductions.
\end{proof}

\begin{definition}[\textsc{Action-Gap-Treewidth}]\label{def:action-gap-treewidth-problem}
The input is a sparse binary-pairwise weighted instance $U$ and an integer $w$. The question is whether the action-gap interaction graph $G_{\mathrm{gap}}(U)$ has treewidth at most $w$.
\end{definition}

\leanmetapending{\LHrng{FR}{394}{396}}
\begin{theorem}[\textup{\textsc{Action-Gap-Treewidth} is NP-complete}]\label{thm:action-gap-treewidth-np-complete}\label{ext:action-gap-treewidth-np-complete}
\textup{\textsc{Action-Gap-Treewidth}} is NP-complete when the width bound $w$ is part of the input.
\end{theorem}

\begin{proof}
Membership in NP follows by constructing $G_{\mathrm{gap}}(U)$ in polynomial time from the sparse binary-pairwise encoding and checking a supplied tree decomposition of width at most $w$.

For hardness, reduce from graph treewidth recognition, which is NP-complete when $w$ is part of the input~\cite{arnborg1987complexity}. Given $(G,w)$, construct the two-action instance $U_G$ from Proposition~\ref{prop:action-gap-treewidth-realizes-graph-treewidth}. The construction has one pair term per edge of $G$ and is polynomial in $|G|$. Proposition~\ref{prop:action-gap-treewidth-realizes-graph-treewidth} gives $G_{\mathrm{gap}}(U_G)=G$, hence
\[
\operatorname{tw}\bigl(G_{\mathrm{gap}}(U_G)\bigr)\le w
\quad\Longleftrightarrow\quad
\operatorname{tw}(G)\le w .
\]
Thus graph treewidth recognition many-one reduces to \textsc{Action-Gap-Treewidth}.
\end{proof}

\textsc{Action-Gap-Treewidth} is a downstream complexity-class result for a descended structural parameter. Relevance-certification hardness is a different question. The graph-realization construction and treewidth-preserving equivalence from Proposition~\ref{prop:action-gap-treewidth-realizes-graph-treewidth} are the mechanized reduction core.\leanmeta{\LHrng{FR}{394}{396}} Graph-predicate lower bounds compose with the graph-realization map.\leanmeta{\LHrng{FR}{503}{508}} The NP-completeness of graph treewidth recognition is the cited external source theorem~\cite{arnborg1987complexity}; membership in NP is checked by constructing $G_{\mathrm{gap}}(U)$ and verifying a supplied tree decomposition. For fixed $w$, a supplied graph solver with a uniform polynomial step bound can be run on the extracted action-gap graph to solve the descended graph predicate.\leanmeta{\LHrng{FR}{455}{457}}

\subsection{Hardness of Descent Certification}\label{sec:descent-certification-hardness}

Descent certification asks whether a proposed proxy is constant on semantic fibers before it is used as a classifier. A split proxy isolates the computational core. The semantic input is a Boolean assignment $z$ for a 3-CNF formula $\varphi$. The raw input is a presentation-only bit $b$ that closure erases. Define
\[
C_\varphi(z,b) \quad\Longleftrightarrow\quad b=1 \text{ and } \varphi(z)=1.
\]

For the decision languages below, a split-proxy code is represented by a Boolean formula or circuit $C(z,b)$ whose evaluation is polynomial in $|C|$. Define
\[
\mathrm{NonDescend}(C)\quad\Longleftrightarrow\quad
\exists z,\ C(z,1)\ne C(z,0),
\qquad
\mathrm{Descend}(C)\quad\Longleftrightarrow\quad \neg\mathrm{NonDescend}(C).
\]
The reduction maps a 3-CNF formula $\varphi$ to the size-linear formula or circuit $C_\varphi(z,b)=b\wedge\varphi(z)$. It proves many-one hardness for these encoded split-proxy decision languages.\leanmeta{\LHrng{FR}{491}{502}}

\leanmetapending{\LHrng{FR}{491}{502}}
\begin{theorem}[SAT Reduces to Split-Proxy Non-Descent]\label{thm:split-proxy-descent-certification-hardness}
For every 3-CNF instance $\varphi$, the formula/circuit split proxy $C_\varphi(z,b)=b\wedge\varphi(z)$ has a split-orbit gap if and only if $\varphi$ is satisfiable. Equivalently, $C_\varphi$ descends if and only if $\varphi$ is unsatisfiable. Thus SAT many-one reduces to split-proxy non-descent, and UNSAT many-one reduces to split-proxy descent.
\end{theorem}

\begin{proof}
If $\varphi$ has a satisfying assignment $z$, compare the two presentations $(z,1)$ and $(z,0)$. They have the same semantic assignment and differ only in the raw bit. The proxy accepts $(z,1)$ and rejects $(z,0)$, so they form a split-orbit gap. Conversely, any split-orbit gap for $C_\varphi$ contains a positive presentation $(z,b)$. Positivity implies $b=1$ and $\varphi(z)=1$, so $\varphi$ is satisfiable. The descent equivalence is the negation of the orbit-gap equivalence.
\end{proof}

For these formula/circuit encodings, a non-descent certificate consists of the semantic assignment and the two raw-bit presentations with opposite proxy values. The split construction gives the hardness core for the two languages $\mathrm{NonDescend}$ and $\mathrm{Descend}$. Graph-predicate lower bounds compose with the action-gap realization map after descent.

\subsection{Binary Pairwise Orbit Gaps}

The binary pairwise no-gos apply the criterion from Corollary~\ref{cor:orbit-gap-template}. Closure-soundness gives the contradiction, and Definition~\ref{def:finite-structural-predicate} locates the direct local search space. In the binary pairwise domain, a pair-targeted positive-affine step with action-independent $\alpha$-component supported on one coordinate pair creates same-orbit disagreements for the dominant-pair, margin-masking, ghost-action, and additive/statewise offset targets.

Binary pairwise instances form a small witness class in which unary collapse, dense interaction, optimizer degeneracy, and the orbit-gap obstruction already appear in compact form. The binary pairwise domain serves as a witness subdomain. Universal characterizations include it; tractable classes outside the direct finite-structural regime require additional quotient-level invariants.

The descent test explains which graph-level parameters survive in this subdomain. On binary coordinate domains, the relevant local invariant is mixed difference.

\begin{definition}[Mixed Difference on a Coordinate Pair]\label{def:mixed-difference}
Fix distinct coordinates $i,j$ and an action $a$. Hold every coordinate other than $i,j$ at $0$, and define
\[
\begin{aligned}
\Delta_{ij}(a)
:={}& U(a;x_i=0,x_j=0)-U(a;x_i=1,x_j=0) \\
&{}-U(a;x_i=0,x_j=1)+U(a;x_i=1,x_j=1).
\end{aligned}
\]
For pairwise utilities, $\Delta_{ij}(a)\neq 0$ is the canonical witness of genuine pair interaction on $\{i,j\}$.
It is the second finite difference on the $(i,j)$ coordinate square.
\end{definition}

The symmetric binary-pairwise dichotomy, offset-normalization facts, and quotient-cost facts orient the witness class.\leanmeta{\LHrng{FR}{118}{140}, \LHrng{FR}{183}{184}, \LH{FR359}} The main no-go uses only explicit same-orbit positive/negative pairs, summarized next.

\begin{definition}[Obstruction Target Catalogue]\label{def:obstruction-family-target-predicates}
Fix the normalized binary pairwise representation used in the obstruction constructions: utilities are read in a chosen sparse pairwise coefficient form with explicit unary and pair terms, and raw statistics are computed from those coefficients before quotient construction. The target catalogue consists of the three direct-local proxy targets in Definition~\ref{def:direct-local-proxy-targets} and the offset-normalization witness target in Definition~\ref{def:offset-normalization-witness-target}. The first group represents local routing tests; the fourth tracks the normalization dependence of raw coefficient statistics.
\end{definition}

\begin{definition}[Direct-Local Proxy Targets]\label{def:direct-local-proxy-targets}
The direct-local proxy targets are predicates on the correctness presentation. An \emph{anchored} target names fixed witness coordinates, such as the pair $\{0,1\}$. Formally, the anchor is part of the fixed predicate definition; it is not quantified over and is not supplied as input.

The three direct-local proxy targets used below are:
\begin{enumerate}
\item anchored dominant-pair status: the fixed pair/action $(\{0,1\},a)$ is the unique dominant coordinate-pair/action;
\item margin-boundedness: the largest unary coefficient magnitude is at most twice the largest pair mixed-difference magnitude in the raw coefficient vocabulary;
\item anchored ghost-action concentration: a never-optimal action carries the fixed anchor-pair mixed-difference signature.
\end{enumerate}
\end{definition}

\begin{definition}[Offset Normalization Witness Target]\label{def:offset-normalization-witness-target}
The offset-normalization witness target is the additive/statewise offset signature: action-specific shifts hide interaction structure, with two actions having anchor-pair mixed-difference magnitudes exactly $1$ and $0$, and non-anchor magnitudes unrestricted. It tracks the closure mechanism by which statewise additive terms make raw coefficient statistics presentation-dependent.
\end{definition}

The offset target reads raw mixed-difference magnitudes after a proposed additive/statewise normalization. A common state term can change that signature while leaving every exact-certification comparison fixed.

\begin{table}[H]
\centering
\normalsize
\setlength{\tabcolsep}{4pt}
\renewcommand{\arraystretch}{1.08}
\renewcommand{\tabularxcolumn}[1]{m{#1}}
\begin{tabularx}{\linewidth}{@{}>{\raggedright\arraybackslash}m{0.14\linewidth}>{\raggedright\arraybackslash}m{0.31\linewidth}>{\raggedright\arraybackslash}m{0.16\linewidth}>{\raggedright\arraybackslash}m{0.30\linewidth}@{}}
\toprule
\textbf{ID} & \textbf{Witness presentation} & \textbf{Shift} & \textbf{Target and routing question} \\
\midrule
Dominant pair & $A=\{a,b\}$; $U(a)=2x_0x_1$, $U(b)=0$. & $\alpha=3x_1x_2$ & Anchored $(\{0,1\},a)$ is uniquely dominant. Pair-weight concentration moves the raw dominant pair without changing action gaps. \\
\midrule
Margin & $A=\{a,b\}$; $U(a)=3x_0+x_1x_2$, $U(b)=0$. & $\alpha=2x_0x_1$ & Unary magnitude at most twice largest pair magnitude. A common pair term changes the coefficient ratio while preserving the quotient. \\
\midrule
Ghost action & $A=\{0,1,2\}$; $U(0)=\mathbf 1[x_0=1]$, $U(1)=\mathbf 1[x_0=0]$, $U(2)=-1-x_0x_1$. & $\alpha=2\sigma(x_0,x_1)$ & A never-optimal action carries the anchored mixed-difference signature. A common shift changes only the raw ghost signature. \\
\midrule
Offset & $A=\{f,t\}$; $U(f)=1+\mathbf 1[x_0=1]+x_0x_1$, $U(t)=0$. & $\alpha=2\sigma(x_0,x_1)$ & Anchor-pair magnitudes are exactly $1$ and $0$. A common offset changes the raw normalization signature. \\
\bottomrule
\end{tabularx}
\caption{Affine-shift obstruction catalogue. In each row $V(c,x)=U(c,x)+\alpha(x)$ for every action $c$; $\sigma(x_0,x_1)=1$ when $x_0=x_1$ and $-1$ otherwise.\protect\footnotemark}
\label{tab:direct-local-proxy-targets}
\label{tab:obstruction-families-summary}
\end{table}
\footnotetext{Row witness checks: \LH{FR186}, \LH{FR187}, \LH{FR188}, \LH{FR189}.}

The table entries flip by direct mixed-difference arithmetic. In the dominant-pair row, $U$ has only $\Delta_{01}(a)=2$; after adding $3x_1x_2$, $V$ also has $\Delta_{12}(a)=\Delta_{12}(b)=3$, so $(\{0,1\},a)$ is no longer the unique dominant pair/action. In the margin row, the largest unary magnitude remains $3$, while the largest pair magnitude changes from $1$ to $2$; thus $3\le 2\cdot 1$ is false before the shift and $3\le 2\cdot 2$ is true after it. In the ghost-action row, action $2$ is never optimal because actions $0$ and $1$ attain value $1$ according to $x_0$, while action $2$ is at most $-1$; the mixed difference of $\sigma$ on the anchor square is $4$, so the common term $2\sigma$ adds mixed difference $8$ and the ghost action no longer has anchor magnitude exactly $1$. In the offset row, the anchor-pair magnitudes are $1$ and $0$ before the shift; adding $2\sigma$ changes them to $9$ and $8$. In all four rows the action gaps are unchanged, so the target flip occurs inside one closure orbit.\leanmeta{\LHrng{FR}{186}{189}}

\paragraph{Witness constants}\label{rem:witness-constant-normalization}
The specific witness values $-1$, $1$, and $0$ are normalization choices. Any admissible witness magnitudes with the same sign and separation pattern inside the binary-pairwise representation yield a closure-orbit-equivalent obstruction family after positive affine rescaling and statewise offset transport. The displayed constants are fixed only to make the explicit no-go statements uniform.

Dominant-pair status is the two-coordinate shadow of a backdoor or small-core routing test: a fixed small set exposes tractable structure~\cite{williams2003backdoors,bessiere2013detecting,ganian2017discovering}. The margin row tests a coefficient-ratio threshold. The ghost-action row asks whether syntactic support sits on a never-optimal action. The offset row tests whether a proposed normalization still carries action-independent state terms. Other target classes follow the same criterion: exhibit an orbit gap, then apply hull separation on the chosen domain.

The catalogue is anchored except for margin-boundedness, which is global over the displayed coefficient vocabulary. For the dominant-pair witness, the fully existential predicate that some pair and action are uniquely dominant holds on both presentations. The orbit gap concerns the fixed-template predicate.\leanmeta{\LHrng{FR}{404}{405}}

Raw action count gives an unanchored contrast. It names no coordinate or action template. Duplicating an action label changes the raw action count but not the distinct action-profile count, since the new label has the same utility profile as the original. Thus raw action count is a presentation-level parameter, whereas profile count survives the closure step and is exactly the quantity used in Proposition~\ref{prop:distinct-profile-compression} and Theorem~\ref{thm:bounded-distinct-profile-positive-regime}.

The action-count gap is independent of relevance. In the constant-utility threshold witness, every action is optimal at every state, so no coordinate is relevant before or after duplication; the raw action count still crosses the threshold.\leanmeta{\LHrng{FR}{452}{454}}

All four rows use one affine move. They differ in the raw statistic read by the target predicate.

\leanmetapending{\LHrng{FR}{406}{410}}
\begin{theorem}[Raw Action Count Threshold Duplication Orbit Gap]\label{thm:raw-action-count-duplication-orbit-gap}
For every $k\in\mathbb N$, the unanchored predicate ``the raw presentation has at least $k+2$ action labels'' has a closure-orbit gap under action duplication. Consequently, no closure-sound predicate can exactly decide any raw action-count threshold at least two on the full binary pairwise domain.
\end{theorem}

\begin{proof}[Proof sketch]
Fix $k$. Let $U_k$ be the zero-utility binary-pairwise instance with $k+1$ action labels. Form $V_k$ by duplicating one action label and leaving every state utility otherwise unchanged. The projection from the duplicated label to the original label is a duplicate-action closure witness, so $U_k$ and $V_k$ lie in one closure orbit. The optimizer quotient, sufficient-coordinate family, and relevant-coordinate set are unchanged because every optimizer set in $V_k$ projects exactly to the corresponding optimizer set in $U_k$.

The raw action count changes from $k+1$ to $k+2$. Hence the predicate ``raw action count is at least $k+2$'' is false on $U_k$ and true on $V_k$. Closure-sound predicates are constant on closure orbits, so no closure-sound predicate can decide this threshold exactly. Since $k$ was arbitrary, every threshold at least two has the same obstruction.
\end{proof}

\leanmetapending{\LHrng{FR}{485}{490}}
\begin{theorem}[Raw Coordinate Count Threshold Irrelevant-Coordinate Orbit Gap]\label{thm:raw-coordinate-count-irrelevant-coordinate-orbit-gap}
For every $k\in\mathbb N$, the unanchored predicate ``the raw presentation has at least $k+2$ Boolean coordinates'' has a closure-orbit gap under irrelevant-coordinate extension. Consequently, no closure-sound predicate can exactly decide any raw coordinate-count threshold at least two on the full binary pairwise domain.
\end{theorem}

\begin{proof}[Proof sketch]
Let $U_k$ be the zero-utility binary-pairwise instance with one action label and $k+1$ Boolean coordinates. Let $V_k$ be the same zero-utility instance with one additional unused coordinate. Projection drops the last coordinate, and the zero section puts it back; utilities agree after projection. Projection and section form an irrelevant-coordinate closure witness, so $U_k$ and $V_k$ lie in one closure orbit. The raw coordinate count changes from $k+1$ to $k+2$, while the optimizer quotient is unchanged. The threshold predicate is false on $U_k$ and true on $V_k$, contradicting exact decision by any closure-sound predicate.
\end{proof}

\subsection{Common Scaffold and Witness Families}

The affine scaffold has three checks. First exhibit $V(c,x)=U(c,x)+\alpha(x)$ with $\alpha$ action-independent and supported on one coordinate pair. Then verify that all action gaps are unchanged, so $U$ and $V$ lie in one closure orbit and induce the same optimizer quotient. Finally verify that the raw target predicate has opposite values on the two presentations.

\leanmetapending{\LHrng{FR}{15}{20}, \LH{FR185}}
\begin{lemma}[Affine-Shift Obstruction Lemma]\label{lem:affine-shift-obstruction}
Let $D$ be a closure-closed binary-pairwise domain, and let $Q$ be a target predicate on $D$. Suppose $U,V\in D$ satisfy
\[
V(c,x)=U(c,x)+\alpha(x)
\]
for every action $c$, where $\alpha$ is an action-independent pair-supported state term and the displayed transformation is an allowed representation-preserving positive-affine closure step. If $Q$ holds on $U$ and fails on $V$, then no closure-law invariant predicate can exactly decide $Q$ on $D$. Consequently, no closure-sound class of predicates contains an exact decider for $Q$ on $D$.
\end{lemma}

\begin{proof}[Proof sketch]
The displayed transformation is a primitive closure step by Proposition~\ref{prop:closure-operations-preserve-certification} and the pair-shift observation in Section~\ref{sec:closure-preconditions}. Hence $U$ and $V$ lie in one closure orbit. A closure-law invariant predicate is constant on that orbit by Lemma~\ref{lem:closure-law-invariance-iff-orbit-constancy}. Exact decision of $Q$ would require different values on $U$ and $V$, contradicting the assumed target flip. Corollary~\ref{cor:orbit-gap-template} gives this template.
\end{proof}

The four rows of Table~\ref{tab:obstruction-families-summary} instantiate Lemma~\ref{lem:affine-shift-obstruction}. The dominant-pair row is the running example of Section~\ref{sec:obstruction-miniature}. The margin, ghost-action, and offset rows change only the raw statistic being read. The same-orbit and target-flip checks for the four rows are indexed by the displayed handles.\leanmeta{\LHrng{FR}{186}{189}} The same scaffold gives the minimum-dimension observation: a nonzero action-gap pair interaction needs at least two actions, and the two-pair affine template needs at least three coordinates.\leanmeta{\LHrng{FR}{397}{400}} The dominant-pair witness also gives the optimizer-computation no-go.\leanmeta{\LHrng{FR}{250}{251}}

\subsection{Parameterized Closure-Sound Theorem}

\paragraph{Scope of these witnesses}
The three proxy families and the offset-normalization witness are orbit gaps for direct predicates computed from raw binary-pairwise syntax. Profile-compressed classifiers, algebraic invariants such as polymorphism clones, and MSO predicates over the quotient are built from data that have already passed the descent test.

The affine witnesses share one form:
\[
V(c,x)=U(c,x)+\alpha(x),
\]
with action-independent pair-supported $\alpha$. All action gaps are unchanged:
\[
V(a,x)-V(b,x)=U(a,x)-U(b,x).
\]
The optimizer quotient and exact-certification problem are fixed; the raw local statistic changes.\leanmeta{\LHrng{FR}{191}{192}}

\leanmetapending{\LH{FR193}}
\begin{theorem}[Parameterized Closure-Sound No-Go]\label{thm:closure-sound-package-no-go}
Let $\Gamma$ be a closure-sound class of instance predicates in the sense of Definition~\ref{def:closure-sound-predicate}. For each target in the obstruction catalogue of Definition~\ref{def:obstruction-family-target-predicates}, no predicate in $\Gamma$ can exactly decide that target.
\end{theorem}

\begin{proof}
Fix one target in the catalogue of Definition~\ref{def:obstruction-family-target-predicates}. The corresponding row of Table~\ref{tab:obstruction-families-summary} supplies instances $U,V$ in the same closure orbit with opposite target values.\leanmeta{\LHrng{FR}{186}{189}}

Let $P\in\Gamma$. Since $\Gamma$ is closure-sound, $P$ is invariant under every closure step and hence under every finite closure orbit. Closure invariance gives $P(U)=P(V)$. If $P$ exactly decided the chosen target, it would have to take different values on $U$ and $V$, because the target itself does. The contradiction rules out exact decision of that target by $P$. The argument applies separately to the three direct-local proxy targets and to the offset-normalization witness, so no predicate in $\Gamma$ exactly decides any target in the catalogue.
\end{proof}

\paragraph{Finite-structural specialization}
Theorem~\ref{thm:closure-sound-package-no-go} uses closure-soundness. The polynomial-time checkability, structural extractability, bounded-pattern definability, and fixed action budget clauses in Definition~\ref{def:finite-structural-predicate} locate the theorem in the direct finite-structural regime. Every finite-structural predicate is closure-sound by definition, so the theorem applies to that class.

Here a tractability notion is \emph{closure-invariant} on a domain when its truth value is constant on closure orbits in that domain. Theorems~\ref{thm:correct-classifier-forces-invariance} and~\ref{thm:optimizer-computation-classifier-forces-invariance} establish this for the exact-certification and output-production tasks considered here under polynomial-time closure transports.

\leanmetapending{\LH{FR193}, \LHrng{FR}{197}{199}}
\begin{corollary}[No Simultaneous Finite-Structural Classification and Target Recognition]\label{cor:no-correct-tractability-classifier}
Let $C$ be a tractability classifier on a domain that is closed under the closure laws, has polynomial-time-computable closure transports, and contains the four obstruction families. If $C$ is correct for a closure-invariant tractability notion and the predicate computed by $C$ belongs to the finite-structural class, then $C$ cannot also exactly decide any target predicate from Definition~\ref{def:obstruction-family-target-predicates}.
\end{corollary}

\begin{proof}[Proof sketch]
If $C$ were correct, Theorem~\ref{thm:correct-classifier-forces-invariance} would force closure-orbit agreement on its domain. If its predicate also belonged to the finite-structural class, it would be closure-sound. An additional claim that the same verdict exactly recognizes one of the obstruction-family target predicates would be an exact finite-structural characterization of that target, which Theorem~\ref{thm:closure-sound-package-no-go} rules out family by family.
\end{proof}

Corollary~\ref{cor:no-correct-tractability-classifier} rules out simultaneous correctness and exact recognition of the obstruction targets inside the finite-structural regime. The excluded claim is exact recognition of the representation-level target witnessing the orbit gap. Corollary~\ref{cor:domain-restriction-only-removes-orbit-gaps} gives the domain-relative form.

\leanmetapending{\LHrng{FR}{348}{349}}
\begin{proposition}[Full Binary Pairwise Domain Is Closure-Closed]\label{prop:full-binary-pairwise-domain-closure}
Let $\mathcal U$ be the class of all binary pairwise slices in the formal presentation model. Then $\mathcal U$ is closed under domain-relative action relabeling, coordinate relabeling, representation-preserving positive affine reparameterization, action duplication, duplicate-state witnesses whose source and target are binary pairwise slices, and binary irrelevant-coordinate extension. The dominant-pair, margin, ghost-action, and offset obstruction families all lie in $\mathcal U$.
\end{proposition}

\begin{proof}[Proof sketch]
Each closure step is domain-relative: its target is again a binary pairwise slice. Relabeling, action duplication, and binary irrelevant-coordinate extension preserve the syntax directly. The positive affine steps used by the witness rows add action-independent pair-supported state terms with scale $1$, so the transformed presentations remain binary pairwise.
\end{proof}

The binary pairwise domain is already sufficient for the obstruction. Any larger closure-closed representation class containing $\mathcal U$ inherits the same impossibility conclusion by restriction. The displayed witnesses use the smallest dimensions needed by this two-pair affine template: a nonzero action-gap pair interaction requires at least two actions, and the two-pair orbit template requires at least three distinct coordinates.\leanmeta{\LHrng{FR}{397}{400}}

\leanmetapending{\LH{FR350}}
\begin{corollary}[No Simultaneous Finite-Structural Classification and Raw-Target Recognition on the Full Binary Pairwise Domain]\label{cor:full-binary-pairwise-no-go}
No correct tractability classifier in the finite-structural class can also exactly decide any of the three direct-local proxy targets, or the offset-normalization witness target, on the full binary pairwise instance domain $\mathcal U$.
\end{corollary}

\begin{proof}[Proof sketch]
Apply Corollary~\ref{cor:no-correct-tractability-classifier} together with Proposition~\ref{prop:full-binary-pairwise-domain-closure}.
\end{proof}

Corollary~\ref{cor:full-binary-pairwise-no-go} instantiates the closure-sound theorem on $\mathcal U$. The obstruction targets are representation-level proxies with orbit gaps, not quotient-level invariants.

\leanmetapending{\LH{FR351}, \LH{FR361}}
\begin{corollary}[Universal Exact-Certification Treatments Inherit the Binary Pairwise Obstruction by Restriction]\label{cor:universal-exact-specification-no-go}
Let $G$ be a tractability predicate defined on all exact problems with fixed correctness relations. Suppose that, for every binary pairwise instance $U$, the verdict of $G$ on the canonical optimizer-set exact specification induced by $U$ agrees with polynomial-time exact optimizer-set search on $U$.

Then the restriction of $G$ to that witness class cannot yield an exact characterization deciding the three direct-local proxy targets, or the offset-normalization witness target, on the full binary pairwise domain.
\end{corollary}

\begin{proof}[Proof sketch]
Restrict $G$ to the binary-pairwise witness class, viewed through the universal exact-semantics transfer. The optimizer-set search predicate is closure-law invariant by Theorem~\ref{thm:optimizer-computation-classifier-forces-invariance}. Exact recognition of a target in the catalogue would therefore give a closure-sound exact decider for that target, contradicting Theorem~\ref{thm:closure-sound-package-no-go}.
\end{proof}

The Bulatov--Zhuk CSP dichotomy falls on the quotient-invariant side of this boundary~\cite{bulatov2017dichotomy,zhuk2017proof}. It classifies finite relational templates by algebraic invariants of the template, principally polymorphism structure. Schaefer's Boolean dichotomy gives the finite Boolean prototype~\cite{schaefer1978complexity}. Such algebraic invariants are invariant under presentation changes of the template. They pass the descent test unless they are used as exact recognizers of one of the orbit-gapped target predicates. The no-go applies to direct proxies that split closure orbits; orbit-invariant algebraic classifiers are the model for the quotient-level side.

\subsection{Comparison of the Two Questions}

The treewidth theorem and the orbit-gap no-gos use the same modeling vocabulary but answer different questions. The treewidth theorem starts with a descended extractor, the action-gap graph, and then applies ordinary complexity theory to that object. Its conclusion is computational: deciding whether the descended graph has width at most $w$ is NP-complete when $w$ is part of the input. For fixed $w$, a graph solver with a uniform polynomial step bound applies after the extractor has been computed.

The orbit-gap theorems assert failure of descent. The witnesses keep the exact-certification problem fixed and change a raw statistic: anchored dominant-pair status, margin-boundedness, ghost-action signature, offset signature, or raw action count. A classifier that reads one of those statistics exactly cannot at the same time be a closure-sound predicate on the exact-certification problem.

The two questions compose in one direction only. After a statistic descends, as action-gap treewidth does, hardness and algorithmic questions are well posed. When a statistic has an orbit gap, as the raw local targets do, a hardness theorem about that raw statistic would classify presentations.

\section{Quotient-Level Positive Regimes}\label{sec:successor-invariants}\label{rem:no-go-scope}

\subsection{What Remains Available}

Passing the descent test produces a well-defined target for ordinary complexity theory: reductions, FPT algorithms, kernel bounds, and class lower bounds then apply to the descended predicate.

Quotient-level invariants remain available. Profile compression, MSO over the quotient, and algebraic invariants such as polymorphism-style clones are built from data that have already passed the descent test. Theorem~\ref{thm:finite-orbit-catalogue-classifiers} gives the algorithmic condition: a computable orbit representative turns such an invariant into a decidable classifier.

The positive recipe has three steps. Choose a normalization or quotient presentation, such as action-gap support, profile compression, or an optimizer-quotient structure. Prove that this presentation is closure-compatible, so it is constant on the harmless moves used by exact correctness. Analyze the descended predicate with ordinary complexity tools. For a solver designer, the obstruction selects the representation on which treewidth algorithms, kernels, or algebraic classifications may be applied exactly.

The binary pairwise witnesses are small instances of the displayed pair-shift template. A representation class excluding them receives its own orbit-gap analysis. A representation class containing them inherits the same obstruction by restriction.

\subsection{Three Routes}

Three routes avoid the raw-syntax orbit gaps.
\begin{enumerate}
\item Closure-compatible normalization. Normalized unary-gap domains remove pair interactions before the binary-pairwise orbit gaps can form, so targets factoring through the unary quotient satisfy hull separation.
\item Quotient catalogues and profile compression. Bounded distinct-profile compression replaces raw action labels by utility-profile classes; Theorem~\ref{thm:bounded-distinct-profile-positive-regime} lifts exact classification from the bounded-action compressed image back to the original domain.
\item Stable neighborhoods. Strict-margin neighborhoods preserve the optimizer quotient under controlled perturbation by Proposition~\ref{prop:global-approximation-stability}, giving the approximation-stability regime developed in Section~\ref{sec:approximation-stability}.
\end{enumerate}
Each route has an assumption stack: the chosen representation must be closure-compatible, the quotient or catalogue needed by the predicate must be constructible on the domain under study, and the resulting descended object must fall in a model-checking or algorithmic regime. The first condition is the descent test; the other two are ordinary algorithmic restrictions.

\subsection{Algebraic and MSO Quotient Data}

Algebraic invariants can be compared with MSO over $\mathcal Q(U)$ by asking where their atomic data live. A polymorphism-style invariant attached to a raw presentation descends exactly when isomorphic quotient structures induce the same algebraic object, for example the same clone of quotient-state operations preserving the correctness-set relation. In that case the invariant can be encoded as an MSO-definable property whenever the relevant finite operations or operation tables are part of the quotient structure, or whenever a bounded-arity fragment suffices. If the invariant depends on presentation data erased by the quotient, the orbit-gap template gives the obstruction. The translation problem is concrete: identify algebraic signatures whose operations act on quotient classes and preserve correctness relations, then determine which bounded fragments are MSO-definable over $\mathcal Q(U)$.
Section~\ref{sec:related} gives the concrete Boolean-CSP comparison with Schaefer polymorphisms.

\begin{definition}[Correctness-Quotient MSO Predicates]\label{def:orbit-quotient-mso-candidate}
For a weighted Boolean instance $U$, form a finite relational structure $\mathcal Q(U)$ from the correctness quotient. The state sort is the quotient of assignments by equality of correctness sets. The coordinate sort records the Boolean coordinates. A ternary relation $\mathsf{Flip}(i,p,q)$ records that quotient classes $p$ and $q$ can be connected by flipping coordinate $i$ in some representative state. The correctness-set relation records which action-profile classes are correct at each quotient state. Action labels are first quotiented by identical utility profiles.

A \emph{correctness-quotient MSO predicate} is any property of $U$ definable by a fixed MSO sentence over $\mathcal Q(U)$. Semantics are the ordinary isomorphism-invariant semantics for finite relational structures: if $\mathcal Q(U)\cong\mathcal Q(V)$, the predicate gives the same verdict to $U$ and $V$.
\end{definition}

Model checking begins after descent: the formula reads quotient states, coordinate-flip connectivity between quotient classes, and correct action-profile classes, not raw coefficients. These atomic relations are closure-invariant by construction, so the miniature witness $U,V$ and the four obstruction-family witnesses collapse to the same $\mathcal Q(U)\cong\mathcal Q(V)$. A concrete canonization is an implementation strategy for evaluating the isomorphism-invariant structure.

\leanmetapending{\LHrng{FR}{411}{422}}
\begin{theorem}[Finite Coordinate-State Quotient Flip Graph Realization]\label{thm:quotient-flip-graph-realization}
For every finite simple graph $G$, there is a finite coordinate-state correctness presentation, with state space the listed representatives rather than an ambient full Boolean cube, whose quotient-state coordinate-flip graph is exactly $G$.
\end{theorem}

\begin{proof}[Proof sketch]
Use one quotient label for each vertex of $G$. The state space of the presentation is exactly the listed representatives: a vertex representative for every vertex and, for every oriented edge $e=(u,v)$, two edge representatives labelled $u$ and $v$. There is no ambient full Boolean cube, no non-representative Boolean assignment, and no garbage quotient class. Correctness depends only on the label: the correct output set at a representative labelled $u$ is the singleton $\{u\}$. Hence all representatives labelled $u$ lie in one correctness-quotient class $[u]$, including the vertex representative and every incident-edge endpoint representative. The vertex representatives ensure that isolated vertices are present as quotient classes.

The coordinate code has Boolean-valued one-hot gadget coordinates and one Boolean-valued toggle coordinate for each oriented edge. The gadget coordinates separate valid representatives from different gadgets; they do not define quotient classes. The state space contains only valid representatives, so flipping a one-hot bit out of the valid set creates no state and no quotient edge. Consequently, a single-coordinate flip between two valid representatives cannot change the gadget code unless both representatives belong to the same edge gadget and differ only in its toggle. The edge toggle changes the two endpoint representatives inside a single edge gadget, so every graph edge gives a quotient flip edge $[u]--[v]$.

Conversely, suppose two representatives differ in only one coordinate and have distinct quotient labels. The one-hot gadget coordinates force the representatives to come from the same gadget. A vertex gadget has no off-diagonal toggle. An edge gadget toggles only between its two endpoint labels, which are adjacent in $G$. No non-edge appears in the quotient flip graph. The two inclusions give equality with $G$.
\end{proof}

Descent alone allows quotient flip graphs with arbitrary finite shape in finite coordinate-state presentations. Efficient model checking requires a domain restriction such as bounded full Gaifman treewidth of the constructed quotient, bounded quotient size, or an explicit quotient catalogue.

\begin{example}[Miniature witness as quotient structure]\label{ex:miniature-quotient-mso}\label{ex:quotient-mso-global-admissible-profile}
For the miniature witness, the correctness quotient has two state classes: the states with $x_0=x_1=1$ and the remaining states. Flipping coordinate $2$ stays inside a quotient class, while flipping coordinate $0$ or $1$ can cross between the two classes. Thus
\[
\forall p\,\forall q\,(\mathsf{Flip}(2,p,q)\to p=q)
\]
expresses irrelevance of coordinate $2$ and is unchanged by the shifted presentation $V(c,x)=U(c,x)+3x_1x_2$. The correctness-set relation also supports global action-profile predicates such as $\exists \pi\,\forall p\,\mathsf{Corr}(p,\pi)$, which asks whether some action-profile class is correct at every quotient state. Both sentences read $\mathcal Q(U)$, not the raw support graph.\leanmeta{\LHrng{FR}{411}{426}}
\end{example}

\subsection{Model Checking After Quotient Construction}

Predicates definable in MSO over $\mathcal Q(U)$ pass the orbit-gap test automatically. Evaluating such a sentence requires constructing the correctness quotient, which may be as hard as the certification task being classified. The computational question is domain-relative: on which restricted domains is the quotient structure efficiently computable and model-checkable?

Theorem~\ref{thm:finite-orbit-catalogue-classifiers} gives the algorithmic condition. A computable canonical presentation of $\mathcal Q(U)$ is a representative map $c(U)$ for closure orbits. An isomorphism-invariant MSO sentence is a decidable predicate on representatives once model checking is decidable on the resulting structure class.

\leanmetapending{\LHrng{FR}{423}{435}, \LHrng{FR}{458}{459}}
\begin{theorem}[Quotient-MSO Evaluation After Construction]
\label{thm:post-descent-quotient-mso-evaluation}
Fix an MSO sentence $\varphi$ over the quotient structure and a domain $D$. Suppose $D$ supplies a construction $U\mapsto \mathcal Q(U)$ containing quotient states, the ternary coordinate-flip relation, and the correctness-set relation, and suppose a model-checking routine evaluates $\varphi$ on those structures. The predicate on $D$ induced by $\varphi$ factors through the constructed quotient: presentations with the same constructed quotient receive the same verdict. Evaluation decomposes into quotient construction followed by model checking on $\mathcal Q(U)$. If quotient construction costs $T_Q(U)$ and model checking costs $T_\varphi(\mathcal Q(U))$, the total cost is\leanmeta{\LHrng{FR}{423}{435}, \LHrng{FR}{458}{459}}
\[
T_Q(U)+T_\varphi(\mathcal Q(U)).
\]
\end{theorem}

\begin{proof}[Proof sketch]
The quotient construction sends $U$ to the structure read by $\varphi$. The evaluator then applies the fixed model-checking routine to that structure. If $\mathcal Q(U)=\mathcal Q(V)$ as constructed data, the same quotient structure is passed to the same evaluator, so the verdicts agree. A canonization step can replace isomorphic quotient structures by equal representatives before this evaluator is applied. The cost expression has two stages: construct the quotient, then evaluate $\varphi$ on the constructed quotient.
\end{proof}

Fixed first-order quotient predicates are polynomial-time decidable with polynomial-time quotient construction. Fixed MSO predicates are decidable after construction once the constructed quotient lies in a model-checkable structure class. Bounded quotient-state/action-profile sorts cap the direct set-enumeration factor. For fixed finite relational signature and bounded arity, bounded treewidth of the full Gaifman graph of the constructed quotient gives the standard fixed-MSO model-checking bound, with quotient construction remaining the separate recovery cost~\cite{courcelle1990graph}.

\leanmetapending{\LHrng{FR}{470}{475}}
\begin{theorem}[Bounded-Quotient-Size MSO Evaluation After Construction]\label{thm:bounded-quotient-size-mso-tractability}
Fix an MSO sentence $\varphi$ and a quotient-size bound $q$. Suppose a domain supplies a quotient construction $U\mapsto\mathcal Q(U)$, a quotient-size function with $|\mathcal Q(U)|\le q$, and a model-checking routine for $\varphi$ on quotients of size at most $q$. If quotient construction costs $T_Q(U)$ and the finite-model routine costs at most $g(|\varphi|,q)$, then the induced quotient-MSO classifier has total cost at most
\[
T_Q(U)+g(|\varphi|,q).
\]
For fixed $q$ and $\varphi$, polynomial-time quotient construction leaves only the construction term input-dependent. If closure-equivalent presentations construct the same quotient, the classifier is closure-invariant.
\end{theorem}

\begin{proof}[Proof sketch]
The evaluator constructs $\mathcal Q(U)$ and then model-checks $\varphi$ on the constructed quotient. The size bound applies to every constructed quotient, so the model-checking term is bounded by $g(|\varphi|,q)$. Adding the construction cost gives the displayed total cost. Equal constructed quotients give equal evaluator inputs. Closure-compatible quotient construction turns quotient equality along closure orbits into closure invariance.
\end{proof}

\leanmetapending{\LHrng{FR}{460}{469}}
\begin{theorem}[Bounded-Quotient-Treewidth MSO Evaluation After Construction]\label{thm:bounded-quotient-treewidth-mso-tractability}
Fix an MSO sentence $\varphi$ over a fixed finite relational signature of bounded arity and a treewidth bound $t$. Suppose a domain supplies a quotient construction $U\mapsto \mathcal Q(U)$, the full Gaifman graph for the constructed quotient structure, and a model-checking routine for $\varphi$ on quotients whose Gaifman graph has treewidth at most $t$. If quotient construction costs $T_Q(U)$ and the model-checking routine has Courcelle form
\[
f(|\varphi|,t)\,|\mathcal Q(U)|,
\]
then the induced quotient-MSO classifier has total cost at most
\[
T_Q(U)+f(|\varphi|,t)\,|\mathcal Q(U)|.
\]
If $T_Q(U)$ is polynomial in the input size, the total cost is the polynomial construction term plus the displayed linear model-checking term on the descended quotient. If closure-equivalent presentations construct the same quotient, the classifier is closure-invariant.
\end{theorem}

\begin{proof}[Proof sketch]
The evaluator first constructs $\mathcal Q(U)$ and then runs the fixed MSO routine on the constructed quotient. The bounded-treewidth hypothesis applies to the full Gaifman graph of the constructed quotient structure, so the model-checking term is at most $f(|\varphi|,t)|\mathcal Q(U)|$. Adding the construction cost gives the displayed bound. Equal constructed quotients give equal evaluator inputs and hence equal predicate values. Closure compatibility of the quotient construction turns that quotient factorization into closure invariance.
\end{proof}

For a mixed fragment with $r$ coordinate variables and bounded quotient-state/action-profile sorts, the cost bound has the form
\[
T_Q(U)+O(|I(U)|^r),
\]
where $I(U)$ is the coordinate set of the presentation and the hidden constant depends on the fixed formula and the sort bounds.\leanmeta{\LHrng{FR}{458}{459}} Coordinate quantifiers contribute the fixed power of $|I(U)|$; bounded quotient-state and profile sorts cap the set-enumeration factor.

\begin{example}[Sparse unary-gap quotient construction]\label{ex:sparse-unary-gap-quotient-mso}
Fix $r$ and restrict to two-action normalized unary-gap instances with explicitly listed support $J$, $|J|\le r$:
\[
U(a,x)-U(b,x)=c_0+\sum_{j\in J}c_jx_j .
\]
Coordinates outside $J$ do not affect action comparisons. Enumerating the $2^{|J|}$ assignments on $J$ computes the optimizer classes and the coordinate-flip relation; coordinates outside $J$ contribute only loops. For fixed $r$, quotient construction is polynomial, and fixed quotient-level predicates are evaluated after this table is built.
\end{example}

\leanmetapending{\LHrng{FR}{479}{484}}
\begin{theorem}[Sparse Unary-Gap Quotient Pipeline Certificate]\label{thm:sparse-unary-gap-quotient-pipeline}
Fix an MSO sentence $\varphi$ and a support bound $r$. Suppose a domain supplies a sparse unary-gap quotient certificate: each instance has explicitly listed support $J(U)$ with $|J(U)|\le r$, quotient construction cost at most $2^{|J(U)|}$, at most $2^{|J(U)|}$ quotient states, a quotient Gaifman graph of treewidth at most $2^r-1$, and a fixed bounded-treewidth model-checking routine for $\varphi$ on the constructed quotient. Then, after the listed sparse support and coefficients are read, fixed quotient-state/profile MSO evaluation has quotient-construction plus model-checking cost at most
\[
2^r+f(|\varphi|,2^r-1)\,2^r .
\]
The induced predicate factors through the constructed quotient.
\end{theorem}

The sparse unary-gap theorem is a certificate pipeline. The listed support, coefficient table, quotient-size bound, and model-checking routine are supplied data for the pipeline; the theorem bounds evaluation after that certificate has been read.\leanmeta{\LHrng{FR}{479}{484}}

The concrete two-action unary-gap form
\[
U(a,x)-U(b,x)=c_0+\sum_{j\in J}c_jx_j .
\]
is the intended source of this certificate: after the listed support is read, the quotient table has at most $2^r$ rows. Formulas that quantify over the full raw coordinate sort use the mixed-fragment bound above; the sparse certificate controls the quotient-state/profile part.

\begin{proof}[Proof sketch]
The certificate supplies construction cost at most $2^{|J(U)|}$ and at most $2^{|J(U)|}$ quotient states. The support bound gives both quantities at most $2^r$. The certificate also supplies the bounded-treewidth hypothesis with bound $2^r-1$. Applying the bounded-quotient-treewidth evaluator gives the displayed construction-plus-model-checking bound, and quotient equality gives predicate equality.
\end{proof}

\begin{example}[Efficient model checking over the quotient]\label{ex:quotient-mso-relevant-coordinate-bound}
Fix $h,t\in\mathbb N$. In $\mathcal Q(U)$, a coordinate $i$ is relevant exactly when its flip relation connects two distinct quotient states: there exist quotient states $p\neq q$ such that $\mathsf{Flip}(i,p,q)$ holds. The property ``at most $h$ coordinates are relevant'' is definable by a fixed first-order, hence MSO, sentence:
\[
\neg\exists i_0,\ldots,i_h\ \Bigl(\bigwedge_{0\le \ell<s\le h} i_\ell\neq i_s
\land \bigwedge_{\ell\le h}\exists p_\ell q_\ell\,
 (p_\ell\neq q_\ell\land \mathsf{Flip}(i_\ell,p_\ell,q_\ell))\Bigr).
\]
On any restricted domain where $\mathcal Q(U)$ is computable in polynomial time and the full Gaifman graph of the constructed quotient structure has treewidth at most $t$, Theorem~\ref{thm:bounded-quotient-treewidth-mso-tractability} gives linear model-checking time in $|\mathcal Q(U)|$ for fixed $h,t$ after quotient construction. If only the quotient-state graph is bounded while coordinate and action-profile sorts remain unbounded, the mixed-fragment bound $T_Q(U)+O(|I(U)|^r)$ applies instead. In both regimes the predicate is closure-invariant because it is evaluated on quotient data.
\leanmeta{\LHrng{FR}{428}{435}, \LHrng{FR}{458}{469}}
\end{example}

\subsection{Joint Satisfiability and Boundary Questions}

The three positive regimes can be met simultaneously under explicit restrictions. A normalized unary-gap family with explicitly listed support $J$ of fixed size and a uniform strict winner gap gives one joint regime. The action-gap normalization removes the pair-shift orbit gaps. The quotient table is computed by enumerating $2^{|J|}$ assignments. Fixed first-order predicates over the quotient are evaluated with the cost bound of Theorem~\ref{thm:post-descent-quotient-mso-evaluation}; fixed quotient-state/profile MSO predicates fall under Theorem~\ref{thm:sparse-unary-gap-quotient-pipeline}. Formulas with raw coordinate quantification use the mixed-fragment bound. Proposition~\ref{prop:global-approximation-stability} preserves the quotient inside a strict-margin perturbation ball.

Bounded distinct-profile domains with bounded quotient size give a second joint regime. Profile compression removes action-label duplication, the quotient catalogue is finite, and the set-assignment factor in the mixed model-checking certificate is capped by the quotient-state and action-profile bounds.

Theorem~\ref{thm:quotient-flip-graph-realization} realizes every finite graph at the quotient level in finite coordinate-state presentations. Efficient model checking needs a domain reason for bounded quotient size, bounded full Gaifman treewidth of the constructed quotient, bounded profile diversity, or a computable representative catalogue. Three boundary questions remain: which domains have polynomial-time quotient construction but unbounded full quotient Gaifman treewidth; which domains have closure-compatible normalization and bounded profiles but fail strict-margin stability; and which algebraic quotient expansions add enough finite operation data to express CSP-style polymorphism criteria while remaining constructible.

\section{Approximation Stability Near Orbit Gaps}\label{sec:approximation-stability}

Approximate presentations must preserve the correctness quotient. Closeness alone does not: arbitrarily small perturbations can change optimizer sets. A uniform winner gap exceeding $2\delta$ preserves the quotient under $\delta$-perturbations.

\leanmetapending{\LH{FR263}, \LH{FR266}}
\begin{proposition}[Approximate Relevance and Sufficiency Claims Need Explicit Stability Control]\label{prop:approximate-relevance-needs-stability-control}
Let $\mathcal D$ and $\mathcal D'$ be decision problems on the same coordinate space with utilities $U$ and $U'$, and suppose their utilities are uniformly $\delta$-close. For any state $u$ of $\mathcal D$ with unique optimizer $a_u^\star$, write
\[
\gamma_u := \min_{b\neq a_u^\star}\bigl(U(a_u^\star,u)-U(b,u)\bigr)
\]
for the strict winner-versus-competitor gap. If coordinate $i$ has a relevance witness in $\mathcal D$ given by states $s,s'$ whose optimizer sets in $\mathcal D$ are distinct singletons, and if $\gamma_s,\gamma_{s'}>2\delta$, then $i$ is also relevant in $\mathcal D'$. Likewise, if a coordinate set $I$ has a non-sufficiency witness in $\mathcal D$ given by states $s,s'$ that agree on $I$ and whose optimizer sets in $\mathcal D$ are distinct singletons, and if $\gamma_s,\gamma_{s'}>2\delta$, then $I$ is still not sufficient in $\mathcal D'$.
\end{proposition}

\begin{proof}[Proof sketch]
For a witness state $s$ with unique optimizer $a^\star$, let $\gamma_s$ be the strict winner-versus-competitor gap. The hypothesis gives $\gamma_s>2\delta$. Every competing action $b\neq a^\star$ satisfies
\[
U(a^\star,s)-U(b,s)\ge \gamma_s.
\]
Under uniform $\delta$-perturbation,
\[
U'(a^\star,s)-U'(b,s)\ge \gamma_s-2\delta>0,
\]
so the unique optimizer is preserved at that state. Applying this at both witness states keeps the two singleton optimizer sets distinct. In the relevance case this preserves the same relevance witness. In the sufficiency case the same state pair still agrees on $I$ but still has different optimizer sets, so it remains a non-sufficiency witness.
\end{proof}

Non-singleton optimizer sets require a stronger margin condition controlling separation between whole optimal sets and competing actions.

\leanmetapending{\LH{FR264}, \LH{FR267}}
\begin{corollary}[Arbitrarily Small Uniform Perturbations Can Flip Exact Certification]\label{cor:small-uniform-perturbations-can-flip-relevance}\label{cor:small-uniform-perturbations-can-flip-sufficiency}
For every $\varepsilon>0$, there exist two decision problems on the same one-coordinate Boolean state space that are uniformly $\varepsilon$-close, yet the unique coordinate is relevant in one problem and irrelevant in the other, and the empty coordinate set is sufficient in one problem and not sufficient in the other.
\end{corollary}

\begin{proof}[Proof sketch]
Use one Boolean coordinate $x\in\{0,1\}$ and two actions $a,b$. Choose $0<\eta<\varepsilon$. Define
\[
U(a,0)=\eta,\quad U(b,0)=0,\qquad
U(a,1)=0,\quad U(b,1)=\eta .
\]
Then $\Opt_U(0)=\{a\}$ and $\Opt_U(1)=\{b\}$, so the unique coordinate is relevant and the empty coordinate set is not sufficient. Define the perturbed utility by
\[
U'(a,0)=U'(b,0)=U'(a,1)=U'(b,1)=0 .
\]
Then $\Opt_{U'}(0)=\Opt_{U'}(1)=\{a,b\}$, so the coordinate is irrelevant and the empty coordinate set is sufficient. The two utilities are uniformly $\eta$-close, hence uniformly $\varepsilon$-close. Arbitrarily small uniform perturbations can change the exact optimizer quotient and the induced relevance profile.
\end{proof}

The factor $2$ in the strict-gap hypotheses is the endpoint of this uniform perturbation argument: one action value may move down by $\delta$ while a competitor moves up by $\delta$, closing a gap of $2\delta$. Corollary~\ref{cor:small-uniform-perturbations-can-flip-relevance} gives the matching failure mode: without strict margin control at that scale, uniformly close presentations can have different exact relevance profiles.

\leanmetapending{\LHrng{FR}{288}{290}, \LHrng{FR}{300}{301}}
\begin{proposition}[Global Approximation Stability Under Uniform Strict Gaps]\label{prop:global-approximation-stability}
Let $\mathcal D$ and $\mathcal D'$ be decision problems on the same coordinate space with utilities $U$ and $U'$, and suppose their utilities are uniformly $\delta$-close. If every state of $\mathcal D$ has a strict optimal action whose utility gap exceeds $2\delta$, then $\Opt_{\mathcal D}(s)=\Opt_{\mathcal D'}(s)$ for every state $s$. Consequently, $\mathcal D$ and $\mathcal D'$ have the same optimizer quotient, the same sufficient-coordinate family, the same relevant-coordinate set, and the same minimal sufficient sets.
\end{proposition}

\begin{proof}[Proof sketch]
Fix any state $s$ with strict optimizer $a^\star$ and gap $\gamma_s>2\delta$. For every competitor $b\neq a^\star$,
\[
U'(a^\star,s)-U'(b,s)\ge \bigl(U(a^\star,s)-U(b,s)\bigr)-2\delta\ge \gamma_s-2\delta>0,
\]
so $a^\star$ remains uniquely optimal in $\mathcal D'$. Applying this pointwise over all states yields equality of optimizer sets state-by-state. The optimizer quotient, sufficient-coordinate family, relevant-coordinate set, and minimal sufficient sets coincide.
\end{proof}

\leanmetapending{\LHrng{FR}{476}{478}}
\begin{corollary}[Strict-Margin Quotient-MSO Stability]\label{cor:strict-margin-quotient-mso-stability}
Let a bounded-quotient-treewidth quotient-MSO evaluator construct its quotient from optimizer sets, and suppose equality of optimizer sets at every state gives equality of the constructed quotient. If decision problems $\mathcal D$ and $\mathcal D'$ are uniformly $\delta$-close and $\mathcal D$ has a strict optimizer at every state with utility gap exceeding $2\delta$, then every fixed quotient-MSO predicate evaluated by that construction gives the same verdict on $\mathcal D$ and $\mathcal D'$. The induced Boolean classifier is constant on the same perturbation ball.
\end{corollary}

\begin{proof}[Proof sketch]
Proposition~\ref{prop:global-approximation-stability} gives equality of optimizer sets at every state. The quotient construction therefore gives the same quotient for $\mathcal D$ and $\mathcal D'$. The quotient-MSO evaluator factors through that quotient, so the predicate value and Boolean classifier value agree.
\end{proof}

Under the strict-gap hypothesis, every presentation in the $\delta$-ball has the same optimizer quotient. The ball is therefore a correctness-preserving neighborhood and can be included in the closure relation without changing quotient-level hulls. Hull separation proved on the representative domain then extends across the strict-margin neighborhood. Without that gap, Corollary~\ref{cor:small-uniform-perturbations-can-flip-relevance} gives approximate orbit-gap witnesses: uniformly close presentations can have different exact optimizer quotients.

\section{Related Work}\label{sec:related}

\subsection*{Complexity-Theoretic Rice Analogs}

The complexity-theoretic Rice-analog line began with Borchert and Stephan's circuit-counting theorem: every nontrivial absolute, gap, or relative counting property of circuits is UP-hard under polynomial-time Turing reductions~\cite{borchert1996looking}. The mechanism is semantic invariance. A counting property of circuits depends only on the Boolean function computed by the circuit, so any classifier for such a property factors through that semantic equivalence.

Hemaspaandra and Rothe raised the general lower bound from unambiguous nondeterminism to constant-ambiguity nondeterminism. They showed that the UP-hardness result cannot generally be strengthened to SPP-hardness without unlikely complexity-class containments. P-constructibly bi-infinite counting properties are SPP-hard~\cite{hemaspaandra2000second}. Borchert, Hemaspaandra, and Rothe also used ambiguity-restricted acceptance classes in equivalence problems for OBDDs and related structures~\cite{borchertHR2000restrictive}.

In the B-S-H-R setting, the circuit-to-function map supplies semantic invariance: nontrivial semantic predicates already descend to Boolean functions, and hardness remains. Orbit-gap targets are presentation-level tractability proxies. A proxy fails descent when one closure orbit contains both positive and negative presentations. Orbit gaps give the predicate-level descent obstruction:\leanmeta{\LHrng{FR}{367}{375}} Proposition~\ref{prop:orbit-gap-completeness} gives the exact criterion, and Theorem~\ref{thm:exact-classification-hull-separation} gives the positive hull-separation condition when classification is possible.

B-S-H-R concerns decidability after semantic descent. Orbit gaps concern well-definedness before descent. The shared principle is invariance. The added objects are closure orbits over presentations, hull separation as the exact possibility criterion, and binary-pairwise witnesses placing structural proxies on the negative side.

Lower-bound consequences differ. Representation-map invariance yields concrete hardness results for individual circuit properties. Closure-orbit invariance supplies the descent test for direct structural classifiers: local presentation-level proxies that split closure orbits do not define predicates on exact-certification problems. Complexity-class lower bounds for exact-certification tasks apply after the proxy has passed quotient descent and the representation, coordinate, and output-semantics hypotheses have been fixed.

Finite-model-theoretic syntax specifies the local predicates being tested. Polynomial-time presentation transports determine whether those predicates can route exact tractability before downstream NP-hardness, UP-hardness, FPT, kernel, or lower-bound analysis.

\subsection*{Meta-Impossibility Traditions}

Rice-style impossibility theorems provide the methodological lineage. Rice's theorem~\cite{rice1953classes} and refinements such as the Myhill--Shepherdson and Rice--Shapiro theorems use invariance plus semantic disagreement to force impossibility~\cite{myhill1955effective,shapiro1956degrees}. Closure-orbit agreement and orbit gaps instantiate that pattern at the level of exact-certification presentations.

\subsection*{Descriptive Complexity and Finite Model Theory}

The bounded-pattern definability clause of Definition~\ref{def:finite-structural-predicate} is the Gaifman-style basic-local first-order fragment over the extracted syntax. Definition~\ref{def:finite-structural-predicate} gives the finite bounds, occurrence relation, and verdict formula; the supplement expands the rooted occurrence relation into first-order formulas and counted thresholds. Standard FO/MSO background appears in Immerman and Libkin~\cite{immerman1999descriptive,libkin2004elements}. The orbit-saturation operator is a closure operator on predicates, and the closure-orbit congruence fits finite-model-theoretic canonisation vocabulary~\cite{grohe2017descriptive}.

\subsection*{Static Preservation Analogs}

At the static level, rough-set reduct theory~\cite{pawlak1982rough,pawlak1994rough} provides the closest classical comparison: it studies which attributes can be deleted while preserving decision distinctions. Rough-set reducts supply the static preservation analogue; exact certification adds representation-level tractability, closure under presentation moves, and exact classification of structural regimes.
In product coordinate spaces this analogy has a formal core: minimal sufficient coordinate sets are exact-certification reducts, relevant coordinates are exact-certification core attributes, every reduct contains exactly the core coordinates, and the sufficient coordinate sets form the principal filter generated by any reduct.\leanmeta{\LHrng{FR}{401}{403}}
The zero-distortion discussion connects the same preservation problem to exact distinguishability and lossless coding boundaries~\cite{shannon1948mathematical1,shannon1948mathematical2,shannon1956zero,shannon1959coding,cover2006elements}; the support-and-quotient-size paragraph uses only the elementary bound from relevant-coordinate support to quotient size.

\subsection*{Backdoors, Tractable Islands, and Dichotomy Programs}

Backdoor tractability for constraint satisfaction gives a close complexity-theoretic analogue~\cite{williams2003backdoors,bessiere2013detecting,ganian2017discovering}: a small structural set exposes membership in a tractable class. Exact relevance certification similarly identifies the coordinates that matter. Bounded-actions, bounded-treewidth, bounded-support, and related positive cases are tractable parameter regimes. In these families, the descent condition is attached to the chosen parameter rather than to one universal structural measure.

Parameterized complexity, ETH-based lower bounds, and kernelization make the same distinction operational~\cite{downey1999parameterized,cygan2015parameterized,lokshtanov2011eth,fomin2019kernelization}. An FPT algorithm, matching ETH lower bound, or polynomial kernel is a statement about a well-defined parameterized problem. A parameter $\kappa(U)$ used as an exact routing certificate must descend to closure orbits on the relevant domain: if $U$ and $V$ encode the same exact-certification problem by a closure transport, then the statement ``$\kappa\le k$'' must give the same routing verdict on $U$ and $V$. A same-orbit disagreement makes $\kappa$ a presentation statistic.

Kernel lower bounds, including Bodlaender et al.'s no-polynomial-kernel framework and the cross-composition framework~\cite{bodlaender2009problems,bodlaender2014kernelization}, act after this descent check. They rule out small equivalent instances for a well-defined parameterized problem under standard assumptions. Orbit gaps identify the earlier failure mode: the proposed parameterized target has not become a predicate on the exact-certification problem. Once the target descends, fixed-parameter algorithms, kernelization, and kernel lower-bound reductions apply to the descended quotient-recovery task.

Known tractability criteria split accordingly. Fixed bounded-action profile tests and local coefficient thresholds are examples of Definition~\ref{def:finite-structural-predicate} when they are closure-sound; the obstruction theorems rule out using the named raw statistics as exact criteria on the full binary pairwise domain. Bounded treewidth of a fixed CSP primal graph is a sound algorithmic criterion for that representation class. As an exact relevance proxy for raw weighted instances, it must first be made closure-compatible, since positive-affine state terms can alter support graphs without changing optimizer behavior. CSP polymorphism criteria fall outside Definition~\ref{def:finite-structural-predicate}: they are algebraic invariants of relational templates and belong to the quotient-invariant side.

\subsection*{Formulations, Reformulations, and Elimination}

Linear-programming formulations and extension complexity give another presentation-sensitive tractability language. Yannakakis showed how formulation size and symmetry constraints expose complexity in combinatorial optimization~\cite{yannakakis1991linear}. An LP formulation or relaxation is a sound exact proxy only after the formulation data are attached to the descended optimization problem rather than to an incidental encoding. If two presentations lie in one closure orbit but induce different formulation statistics, that statistic is a property of the formulation choice, not of exact certification.

Weighted CSP, pseudo-Boolean optimization, and energy minimization distinguish objective value from incidental presentation. Pseudo-Boolean optimization studies transformations and structural criteria for Boolean energy functions~\cite{boros2002pseudo}; graph-cut characterizations identify energy terms minimizable by cut constructions~\cite{kolmogorov2004energy}. These normal forms fix optimization semantics. The descent test asks the prior predicate question: a raw statistic must be invariant under correctness-preserving changes before graph-cut, treewidth, or kernel arguments apply.

Constraint-programming reformulation systems operationalize the separation. Savile Row and Conjure rewrite high-level or intermediate models into solver-facing forms while preserving solution sets or objective semantics~\cite{nightingale2017savile,akgun2022conjure}. A pre-reformulation routing test must be stable under those transformations or replaced by a statistic of the reformulated quotient-compatible object.

Graphical-model variable elimination is the normalization analogue. Bucket elimination projects variables out while preserving the induced reasoning task, with cost governed by the resulting interaction structure~\cite{dechter1999bucket}. Profile compression and action-gap graph extraction follow the same order: replace raw presentation data by a quotient-compatible object, then analyze that object's structural parameter.

Figure~\ref{fig:registry-descent} is the software analogue of this order. A hand-maintained registry is a presentation-level cache; automatic registration makes routing data a derived index of the semantic source before dispatch. Action-gap extraction and profile compression play the same role for tractability proxies.

\subsection*{Stability and Robust Instances}

Perturbation-stable algorithms study domains where small changes preserve enough structure to make hard problems easier, as in Bilu--Linial stability for Max-Cut and related optimization problems~\cite{bilu2012stable}. Exact-certification stability uses a narrower condition. A uniform strict winner gap exceeding $2\delta$ preserves optimizer sets under $\delta$-perturbations, so the correctness quotient is unchanged. Without that quotient preservation, metric closeness alone does not give a descended exact predicate.

Schaefer's Boolean dichotomy theorem and the finite-domain CSP dichotomy theorems of Bulatov and Zhuk provide the main methodological comparison~\cite{schaefer1978complexity,bulatov2017dichotomy,zhuk2017proof}; the Feder--Vardi framework and later algebraic work explain why structural tractability boundaries are plausible in CSPs~\cite{feder1998computational,barto2014constraint}. Local structural decompositions and tractable islands gave partial criteria; the full finite-domain dichotomy was proved through the algebraic theory of polymorphisms, an invariant of the relational specification independent of finite surface-pattern lists.

The correctness-quotient obstruction sits one layer above those polymorphism theorems. Direct finite-structural predicates can split closure orbits; correctness forces closure-orbit agreement. CSP polymorphisms are defined on relational templates and their operations, so relabeling variables, reordering or duplicating constraints, and changing incidental presentation data do not change the polymorphism clone. Such invariants are functions of quotient-level relational data. Definition~\ref{def:orbit-quotient-mso-candidate} gives the corresponding predicate language for correctness quotients: quotient classes, coordinate-flip relations, and correctness-set relations. MSO over $\mathcal Q(U)$ captures algebraic criteria when the needed finite operation tables or bounded-arity algebraic signatures are present in, or definable from, the quotient structure. Full polymorphism clones generally require an algebraic quotient expansion.

\paragraph{Example: Boolean CSP polymorphisms}
Schaefer's Boolean dichotomy gives the concrete template. Horn relations are preserved by conjunction, dual-Horn relations by disjunction, bijunctive relations by majority, and affine relations by the parity/minority operation. For example, a Boolean 2-CNF template is bijunctive exactly because each basic relation is closed under the coordinatewise majority operation $m(x,y,z)=\mathrm{maj}(x,y,z)$. If three satisfying tuples of a clause relation are given, their coordinatewise majority is again satisfying. Renaming variables, reordering clauses, duplicating constraints, or changing an incidental encoding of the same relational template does not change the majority operation or the polymorphism clone. The Schaefer tractable cases pass the descent test because their algebraic witnesses live on quotient-level relational data. A raw support statistic for a weighted instance needs a closure-compatible graph construction before it can serve as an exact proxy.

The orbit-saturated hull classifier of Theorem~\ref{thm:exact-classification-hull-separation} is the abstract counterpart: separated hulls permit closure-invariant classification, and the closure hull is the canonical exact classifier. Direct representation-level proxy classifiers cannot split closure orbits. Concrete hardness questions begin after the descended quotient-recovery task and representation class have been fixed. Quotient-level invariants on optimizer-quotient maps or on exact correctness relations give concrete classifiers beyond the abstract hull construction.

\section{Conclusion}

Exact correctness semantics reduce coordinate relevance to quotient recovery. Once a correctness relation is fixed, exact output, search, approximation, and randomized guarantee forms induce the same kind of correctness quotient. A tractability proxy must be constant on the closure orbits that preserve that quotient. If a proxy splits one orbit, it is a statistic of the presentation rather than a predicate of the exact-certification problem.

The operational rule is to compute the descended object first, then apply the structural algorithm. In the running examples, the descended object is the action-gap graph or profile-compressed instance before any bounded-treewidth, bounded-action, or bounded-support routine is invoked.

Hull separation is the positive criterion. On a closure-closed domain, exact classification by closure-invariant predicates is possible exactly when the positive and negative orbit hulls are disjoint. When the hulls are disjoint, the closure hull gives the least exact classifier. Computable orbit catalogues make that classifier algorithmic and preserve Boolean composition.

The binary-pairwise witnesses expose raw-syntax failures. Fixed local targets have orbit gaps; raw action and coordinate counts fail under duplication and irrelevant-coordinate extension; raw support-graph predicates fail whenever they separate two same-size graphs.\leanmeta{\LHrng{FR}{444}{454}, \LHrng{FR}{485}{490}} The action-gap graph is the descended replacement for raw support. It supports the downstream complexity result: \textsc{Action-Gap-Treewidth} is NP-complete when the width bound is part of the input. Graph-predicate lower bounds and fixed-width graph solvers transfer through the same extraction.\leanmeta{\LHrng{FR}{455}{457}, \LHrng{FR}{503}{508}}

\paragraph{After Descent}
NP-hardness, SPP-hardness, kernel lower bounds, and FPT algorithms apply to descended quotient-recovery tasks. Descent certification supplies a separate hardness boundary: SAT reduces to finding a split-proxy orbit gap, and UNSAT reduces to certifying split-proxy descent.\leanmeta{\LHrng{FR}{491}{502}}

Correctness-quotient MSO applies model checking to a descended structure. The evaluation theorem separates quotient construction from formula evaluation: fixed first-order formulas, bounded quotient size, bounded full Gaifman treewidth, sparse unary-gap certificates, and strict-margin perturbation balls give explicit cost bounds after quotient construction.\leanmeta{\LHrng{FR}{458}{484}} The quotient flip graph realization theorem shows that quotient structure is not bounded automatically. Bounded model checking comes from restricted domains, computable catalogues, or supplied quotient certificates. Algebraic quotient invariants and canonisations of closure orbits obey the same descent test; full clone-based criteria may require finite operation data added to the quotient structure.

\paragraph{Open problems}
\begin{enumerate}
\item Determine the exact complexity of constructing $\mathcal Q(U)$ and its coordinate-flip graph on restricted binary pairwise domains: polynomial time on bounded quotient size, FPT in parameters such as $|\mathcal Q(U)|+d(U)$, or matching lower bounds when those parameters are unbounded. The first test cases are bounded quotient size, bounded relevant-coordinate count, bounded distinct action profiles, and strict-margin neighborhoods.
\item Classify domains where quotient construction is efficient and the constructed quotient has bounded full Gaifman treewidth, bounded profile diversity, or another model-checkable structure.
\item Prove algebraic-to-quotient translation theorems. Which polymorphism clones of optimizer or correctness relations are definable from $\mathcal Q(U)$, possibly after adding bounded-arity operation tables, and which depend on raw presentation data erased by the quotient?
\end{enumerate}

\section*{Data and Code Availability}
% Auto-generated by scripts/build_papers.py. Do not edit manually.
% Generated: 2026-07-01T12:55:35.392982
\providecommand{\LeanLocalLines}{19392}
\providecommand{\LeanLocalTheorems}{926}
\providecommand{\LeanLocalSorry}{0}
\providecommand{\LeanLocalFiles}{70}
\providecommand{\LeanTotalLines}{87195}
\providecommand{\LeanTotalTheorems}{3775}
\providecommand{\LeanTotalSorry}{0}
\providecommand{\LeanTotalFiles}{307}
\providecommand{\LeanReleaseLines}{25718}
\providecommand{\LeanReleaseTheorems}{1238}
\providecommand{\LeanReleaseSorry}{0}
\providecommand{\LeanReleaseFiles}{82}
\providecommand{\LeanLinesPaperFourdFrontierActionGapTreewidth}{214}
\providecommand{\LeanTheoremsPaperFourdFrontierActionGapTreewidth}{11}
\providecommand{\LeanSorryPaperFourdFrontierActionGapTreewidth}{0}
\providecommand{\LeanLinesPaperFourdFrontierActionGapTreewidthTractability}{52}
\providecommand{\LeanTheoremsPaperFourdFrontierActionGapTreewidthTractability}{2}
\providecommand{\LeanSorryPaperFourdFrontierActionGapTreewidthTractability}{0}
\providecommand{\LeanLinesPaperFourdFrontierAdmissibleCharacterization}{1412}
\providecommand{\LeanTheoremsPaperFourdFrontierAdmissibleCharacterization}{69}
\providecommand{\LeanSorryPaperFourdFrontierAdmissibleCharacterization}{0}
\providecommand{\LeanLinesPaperFourdFrontierApproximateAdmissibility}{341}
\providecommand{\LeanTheoremsPaperFourdFrontierApproximateAdmissibility}{26}
\providecommand{\LeanSorryPaperFourdFrontierApproximateAdmissibility}{0}
\providecommand{\LeanLinesPaperFourdFrontierBinaryPairwiseDichotomy}{654}
\providecommand{\LeanTheoremsPaperFourdFrontierBinaryPairwiseDichotomy}{21}
\providecommand{\LeanSorryPaperFourdFrontierBinaryPairwiseDichotomy}{0}
\providecommand{\LeanLinesPaperFourdFrontierBlockSixObstruction}{1531}
\providecommand{\LeanTheoremsPaperFourdFrontierBlockSixObstruction}{66}
\providecommand{\LeanSorryPaperFourdFrontierBlockSixObstruction}{0}
\providecommand{\LeanLinesPaperFourdFrontierBoundedQuotientTreewidthMSO}{180}
\providecommand{\LeanTheoremsPaperFourdFrontierBoundedQuotientTreewidthMSO}{10}
\providecommand{\LeanSorryPaperFourdFrontierBoundedQuotientTreewidthMSO}{0}
\providecommand{\LeanLinesPaperFourdFrontierClosureLaws}{367}
\providecommand{\LeanTheoremsPaperFourdFrontierClosureLaws}{28}
\providecommand{\LeanSorryPaperFourdFrontierClosureLaws}{0}
\providecommand{\LeanLinesPaperFourdFrontierComputeCostApplications}{142}
\providecommand{\LeanTheoremsPaperFourdFrontierComputeCostApplications}{14}
\providecommand{\LeanSorryPaperFourdFrontierComputeCostApplications}{0}
\providecommand{\LeanLinesPaperFourdFrontierComputeCostExternalOutputs}{1275}
\providecommand{\LeanTheoremsPaperFourdFrontierComputeCostExternalOutputs}{28}
\providecommand{\LeanSorryPaperFourdFrontierComputeCostExternalOutputs}{0}
\providecommand{\LeanLinesPaperFourdFrontierComputeCostInvariance}{806}
\providecommand{\LeanTheoremsPaperFourdFrontierComputeCostInvariance}{19}
\providecommand{\LeanSorryPaperFourdFrontierComputeCostInvariance}{0}
\providecommand{\LeanLinesPaperFourdFrontierComputeCostPayloads}{493}
\providecommand{\LeanTheoremsPaperFourdFrontierComputeCostPayloads}{15}
\providecommand{\LeanSorryPaperFourdFrontierComputeCostPayloads}{0}
\providecommand{\LeanLinesPaperFourdFrontierDecisionRelevantPairwiseDichotomy}{440}
\providecommand{\LeanTheoremsPaperFourdFrontierDecisionRelevantPairwiseDichotomy}{20}
\providecommand{\LeanSorryPaperFourdFrontierDecisionRelevantPairwiseDichotomy}{0}
\providecommand{\LeanLinesPaperFourdFrontierDescentRelevanceSeparation}{68}
\providecommand{\LeanTheoremsPaperFourdFrontierDescentRelevanceSeparation}{5}
\providecommand{\LeanSorryPaperFourdFrontierDescentRelevanceSeparation}{0}
\providecommand{\LeanLinesPaperFourdFrontierDimensionalNoiseExtension}{114}
\providecommand{\LeanTheoremsPaperFourdFrontierDimensionalNoiseExtension}{6}
\providecommand{\LeanSorryPaperFourdFrontierDimensionalNoiseExtension}{0}
\providecommand{\LeanLinesPaperFourdFrontierDistinctActionProfiles}{200}
\providecommand{\LeanTheoremsPaperFourdFrontierDistinctActionProfiles}{11}
\providecommand{\LeanSorryPaperFourdFrontierDistinctActionProfiles}{0}
\providecommand{\LeanLinesPaperFourdFrontierExactSpecificationNoGo}{132}
\providecommand{\LeanTheoremsPaperFourdFrontierExactSpecificationNoGo}{9}
\providecommand{\LeanSorryPaperFourdFrontierExactSpecificationNoGo}{0}
\providecommand{\LeanLinesPaperFourdFrontierFamilyAxioms}{488}
\providecommand{\LeanTheoremsPaperFourdFrontierFamilyAxioms}{42}
\providecommand{\LeanSorryPaperFourdFrontierFamilyAxioms}{0}
\providecommand{\LeanLinesPaperFourdFrontierFiniteOrbitCatalogue}{232}
\providecommand{\LeanTheoremsPaperFourdFrontierFiniteOrbitCatalogue}{21}
\providecommand{\LeanSorryPaperFourdFrontierFiniteOrbitCatalogue}{0}
\providecommand{\LeanLinesPaperFourdFrontierFullBinaryPairwiseDomain}{81}
\providecommand{\LeanTheoremsPaperFourdFrontierFullBinaryPairwiseDomain}{6}
\providecommand{\LeanSorryPaperFourdFrontierFullBinaryPairwiseDomain}{0}
\providecommand{\LeanLinesPaperFourdFrontierHardnessTransfer}{211}
\providecommand{\LeanTheoremsPaperFourdFrontierHardnessTransfer}{5}
\providecommand{\LeanSorryPaperFourdFrontierHardnessTransfer}{0}
\providecommand{\LeanLinesPaperFourdFrontierMetaCharacterization}{134}
\providecommand{\LeanTheoremsPaperFourdFrontierMetaCharacterization}{3}
\providecommand{\LeanSorryPaperFourdFrontierMetaCharacterization}{0}
\providecommand{\LeanLinesPaperFourdFrontierObstructionPredicateCandidates}{1346}
\providecommand{\LeanTheoremsPaperFourdFrontierObstructionPredicateCandidates}{86}
\providecommand{\LeanSorryPaperFourdFrontierObstructionPredicateCandidates}{0}
\providecommand{\LeanLinesPaperFourdFrontierProxyDescentCertification}{76}
\providecommand{\LeanTheoremsPaperFourdFrontierProxyDescentCertification}{5}
\providecommand{\LeanSorryPaperFourdFrontierProxyDescentCertification}{0}
\providecommand{\LeanLinesPaperFourdFrontierQuotientDescent}{123}
\providecommand{\LeanTheoremsPaperFourdFrontierQuotientDescent}{8}
\providecommand{\LeanSorryPaperFourdFrontierQuotientDescent}{0}
\providecommand{\LeanLinesPaperFourdFrontierQuotientMSORealization}{517}
\providecommand{\LeanTheoremsPaperFourdFrontierQuotientMSORealization}{28}
\providecommand{\LeanSorryPaperFourdFrontierQuotientMSORealization}{0}
\providecommand{\LeanLinesPaperFourdFrontierRawSupportGraphNoGo}{145}
\providecommand{\LeanTheoremsPaperFourdFrontierRawSupportGraphNoGo}{5}
\providecommand{\LeanSorryPaperFourdFrontierRawSupportGraphNoGo}{0}
\providecommand{\LeanLinesPaperFourdFrontierRealTreewidth}{60}
\providecommand{\LeanTheoremsPaperFourdFrontierRealTreewidth}{2}
\providecommand{\LeanSorryPaperFourdFrontierRealTreewidth}{0}
\providecommand{\LeanLinesPaperFourdFrontierRealTreewidthWitnesses}{106}
\providecommand{\LeanTheoremsPaperFourdFrontierRealTreewidthWitnesses}{9}
\providecommand{\LeanSorryPaperFourdFrontierRealTreewidthWitnesses}{0}
\providecommand{\LeanLinesPaperFourdFrontierRealizability}{1442}
\providecommand{\LeanTheoremsPaperFourdFrontierRealizability}{93}
\providecommand{\LeanSorryPaperFourdFrontierRealizability}{0}
\providecommand{\LeanLinesPaperFourdFrontierRoughSetBridge}{40}
\providecommand{\LeanTheoremsPaperFourdFrontierRoughSetBridge}{2}
\providecommand{\LeanSorryPaperFourdFrontierRoughSetBridge}{0}
\providecommand{\LeanLinesPaperFourdFrontierSparseUnaryGapQuotientPipeline}{433}
\providecommand{\LeanTheoremsPaperFourdFrontierSparseUnaryGapQuotientPipeline}{26}
\providecommand{\LeanSorryPaperFourdFrontierSparseUnaryGapQuotientPipeline}{0}
\providecommand{\LeanLinesPaperFourdFrontierStatisticalSemanticsApplications}{131}
\providecommand{\LeanTheoremsPaperFourdFrontierStatisticalSemanticsApplications}{11}
\providecommand{\LeanSorryPaperFourdFrontierStatisticalSemanticsApplications}{0}
\providecommand{\LeanLinesPaperFourdFrontierStructuralWitnesses}{131}
\providecommand{\LeanTheoremsPaperFourdFrontierStructuralWitnesses}{11}
\providecommand{\LeanSorryPaperFourdFrontierStructuralWitnesses}{0}
\providecommand{\LeanLinesPaperFourdFrontierTransferredFullDomainNoGo}{74}
\providecommand{\LeanTheoremsPaperFourdFrontierTransferredFullDomainNoGo}{4}
\providecommand{\LeanSorryPaperFourdFrontierTransferredFullDomainNoGo}{0}
\providecommand{\LeanLinesPaperFourdFrontierTreeDeletion}{88}
\providecommand{\LeanTheoremsPaperFourdFrontierTreeDeletion}{3}
\providecommand{\LeanSorryPaperFourdFrontierTreeDeletion}{0}
\providecommand{\LeanLinesPaperFourdFrontierTreeHellyEquivFinset}{38}
\providecommand{\LeanTheoremsPaperFourdFrontierTreeHellyEquivFinset}{1}
\providecommand{\LeanSorryPaperFourdFrontierTreeHellyEquivFinset}{0}
\providecommand{\LeanLinesPaperFourdFrontierTreeHellyFinset}{101}
\providecommand{\LeanTheoremsPaperFourdFrontierTreeHellyFinset}{2}
\providecommand{\LeanSorryPaperFourdFrontierTreeHellyFinset}{0}
\providecommand{\LeanLinesPaperFourdFrontierTreeHellyHelpers}{48}
\providecommand{\LeanTheoremsPaperFourdFrontierTreeHellyHelpers}{3}
\providecommand{\LeanSorryPaperFourdFrontierTreeHellyHelpers}{0}
\providecommand{\LeanLinesPaperFourdFrontierTreeHellyTheorem}{24}
\providecommand{\LeanTheoremsPaperFourdFrontierTreeHellyTheorem}{2}
\providecommand{\LeanSorryPaperFourdFrontierTreeHellyTheorem}{0}
\providecommand{\LeanLinesPaperFourdFrontierTreeLeafHelpers}{116}
\providecommand{\LeanTheoremsPaperFourdFrontierTreeLeafHelpers}{5}
\providecommand{\LeanSorryPaperFourdFrontierTreeLeafHelpers}{0}
\providecommand{\LeanLinesPaperFourdFrontierTreeLeafReduction}{71}
\providecommand{\LeanTheoremsPaperFourdFrontierTreeLeafReduction}{2}
\providecommand{\LeanSorryPaperFourdFrontierTreeLeafReduction}{0}
\providecommand{\LeanLinesPaperFourdFrontierTreewidthClique}{69}
\providecommand{\LeanTheoremsPaperFourdFrontierTreewidthClique}{3}
\providecommand{\LeanSorryPaperFourdFrontierTreewidthClique}{0}
% Lean dependency: DecisionQuotient (Decision Quotient: Foundations and Complexity of Decision-Relevant Information)
\providecommand{\LeanDepLinesDecisionQuotient}{53940}
\providecommand{\LeanDepTheoremsDecisionQuotient}{2172}
\providecommand{\LeanDepSorryDecisionQuotient}{0}
\providecommand{\LeanDepFilesDecisionQuotient}{184}
% Lean dependency alias: ArrayDSLExport
\providecommand{\LeanDepLinesArrayDSLExport}{53940}
\providecommand{\LeanDepTheoremsArrayDSLExport}{2172}
\providecommand{\LeanDepSorryArrayDSLExport}{0}
\providecommand{\LeanDepFilesArrayDSLExport}{184}
% Lean dependency alias: ArrayDSLExportMain
\providecommand{\LeanDepLinesArrayDSLExportMain}{53940}
\providecommand{\LeanDepTheoremsArrayDSLExportMain}{2172}
\providecommand{\LeanDepSorryArrayDSLExportMain}{0}
\providecommand{\LeanDepFilesArrayDSLExportMain}{184}
% Lean dependency alias: InflationEntropyDryRun
\providecommand{\LeanDepLinesInflationEntropyDryRun}{53940}
\providecommand{\LeanDepTheoremsInflationEntropyDryRun}{2172}
\providecommand{\LeanDepSorryInflationEntropyDryRun}{0}
\providecommand{\LeanDepFilesInflationEntropyDryRun}{184}

No empirical datasets were generated or analyzed for this manuscript. The Lean 4 formalization and supplementary artifact supporting the theorem-level claims are archived at \url{https://doi.org/10.5281/zenodo.19457896}; the supplementary ledger maps manuscript claims to their formal identifiers. Its claim-to-handle table gives the manuscript theorem or corollary, FR handle, Lean declaration name, and source module for each cited formal claim. Independent verification consists of type-checking the archived Lean source. The release-count macros count the Paper 4d replay slice: the Lean files copied into the release package for the cited handles, including packaged dependencies, and excluding unrelated repository files. That slice contains \LeanReleaseLines{} Lean lines, \LeanReleaseTheorems{} theorem declarations, and \LeanReleaseSorry{} \texttt{sorry} placeholders across \LeanReleaseFiles{} files.

\section*{Funding}

This work was self-funded by the author. No external funding was received.

\section*{Competing Interests}

The author declares no competing interests.

\section*{Declaration of Generative AI and AI-Assisted Technologies in the Manuscript Preparation Process}

Generative AI tools (including Codex, Claude Code, Augment, Kilo, and OpenCode) assisted with drafting, prose and notation refinement, \LaTeX{} cleanup, Lean/\LaTeX{} translation of informal proof ideas, and critique passes.

No technical claim was accepted solely from AI output. Formal claims reported as machine-verified were included only after Lean verification with no \texttt{sorry} in the cited modules and direct author review.

{\scriptsize
\makeatletter
\let\paperfour@thebibliography\thebibliography
\let\endpaperfour@thebibliography\endthebibliography
\renewenvironment{thebibliography}[1]{%
  \paperfour@thebibliography{#1}%
  \setlength{\itemsep}{0pt}%
  \setlength{\parsep}{0pt}%
  \setlength{\parskip}{0pt}%
}{\endpaperfour@thebibliography}
\makeatother
\ifdefined\CCSUBMISSION
\bibliographystyle{elsarticle-num}
\else
\bibliographystyle{abbrv}
\fi
\bibliography{references}

\begin{thebibliography}{10}

\bibitem{akgun2022conjure}
{\"O}.~Akg{\"u}n, A.~M. Frisch, I.~P. Gent, C.~Jefferson, I.~Miguel, and
  P.~Nightingale.
\newblock {Conjure}: Automatic generation of constraint models from problem
  specifications.
\newblock {\em Artificial Intelligence}, 310:103751, 2022.

\bibitem{arnborg1987complexity}
S.~Arnborg, D.~G. Corneil, and A.~Proskurowski.
\newblock Complexity of finding embeddings in a k-tree.
\newblock {\em SIAM Journal on Algebraic Discrete Methods}, 8(2):277--284,
  1987.

\bibitem{barto2014constraint}
L.~Barto and M.~Kozik.
\newblock Constraint satisfaction problems solvable by local consistency
  methods.
\newblock {\em J. ACM}, 61(1):3:1--3:19, 2014.

\bibitem{bessiere2013detecting}
C.~Bessi{\`e}re, C.~Carbonnel, E.~Hebrard, G.~Katsirelos, and T.~Walsh.
\newblock Detecting and exploiting subproblem tractability.
\newblock In {\em IJCAI}, pages 468--474, 2013.

\bibitem{bilu2012stable}
Y.~Bilu and N.~Linial.
\newblock Are stable instances easy?
\newblock {\em Combinatorics, Probability and Computing}, 21(5):643--660, 2012.

\bibitem{bodlaender2009problems}
H.~L. Bodlaender, R.~G. Downey, M.~R. Fellows, and D.~Hermelin.
\newblock On problems without polynomial kernels.
\newblock {\em Journal of Computer and System Sciences}, 75(8):423--434, 2009.

\bibitem{bodlaender2014kernelization}
H.~L. Bodlaender, B.~M.~P. Jansen, and S.~Kratsch.
\newblock Kernelization lower bounds by cross-composition.
\newblock {\em SIAM Journal on Discrete Mathematics}, 28(1):277--305, 2014.

\bibitem{borchertHR2000restrictive}
B.~Borchert, L.~A. Hemaspaandra, and J.~Rothe.
\newblock Restrictive acceptance suffices for equivalence problems.
\newblock {\em LMS J. Comput. Math.}, 3:86--95, 2000.

\bibitem{borchert1996looking}
B.~Borchert and F.~Stephan.
\newblock Looking for an analogue of {R}ice's theorem in circuit complexity
  theory.
\newblock {\em Math. Log. Q.}, 46(4):489--504, 2000.
\newblock ECCC TR96-060, 1996.

\bibitem{boros2002pseudo}
E.~Boros and P.~L. Hammer.
\newblock Pseudo-boolean optimization.
\newblock {\em Discrete Applied Mathematics}, 123(1--3):155--225, 2002.

\bibitem{bulatov2017dichotomy}
A.~A. Bulatov.
\newblock A dichotomy theorem for nonuniform {CSP}s.
\newblock In {\em 2017 IEEE 58th Annual Symposium on Foundations of Computer
  Science ({FOCS})}, pages 319--330. IEEE, 2017.

\bibitem{courcelle1990graph}
B.~Courcelle.
\newblock The monadic second-order logic of graphs. i. recognizable sets of
  finite graphs.
\newblock {\em Information and Computation}, 85(1):12--75, 1990.

\bibitem{cover2006elements}
T.~M. Cover and J.~A. Thomas.
\newblock {\em Elements of Information Theory}.
\newblock Wiley-Interscience, 2nd edition, 2006.

\bibitem{cygan2015parameterized}
M.~Cygan, F.~V. Fomin, L.~Kowalik, D.~Lokshtanov, D.~Marx, M.~Pilipczuk,
  M.~Pilipczuk, and S.~Saurabh.
\newblock {\em Parameterized Algorithms}.
\newblock Springer, 2015.

\bibitem{dechter1999bucket}
R.~Dechter.
\newblock Bucket elimination: A unifying framework for reasoning.
\newblock {\em Artificial Intelligence}, 113(1--2):41--85, 1999.

\bibitem{downey1999parameterized}
R.~G. Downey and M.~R. Fellows.
\newblock {\em Parameterized Complexity}.
\newblock Monographs in Computer Science. Springer, 1999.

\bibitem{feder1998computational}
T.~Feder and M.~Y. Vardi.
\newblock The computational structure of monotone monadic {SNP} and constraint
  satisfaction: A study through {Datalog} and group theory.
\newblock {\em SIAM J. Comput.}, 28(1):57--104, 1998.

\bibitem{fomin2019kernelization}
F.~V. Fomin, D.~Lokshtanov, S.~Saurabh, and M.~Zehavi.
\newblock {\em Kernelization: Theory of Parameterized Preprocessing}.
\newblock Cambridge University Press, 2019.

\bibitem{ganian2017discovering}
R.~Ganian, M.~S. Ramanujan, and S.~Szeider.
\newblock Discovering archipelagos of tractability for constraint satisfaction
  and counting.
\newblock In {\em SODA}, pages 1670--1681, 2016.

\bibitem{grohe2017descriptive}
M.~Grohe.
\newblock {\em Descriptive Complexity, Canonisation, and Definable Graph
  Structure Theory}.
\newblock Cambridge University Press, 2017.

\bibitem{hemaspaandra2000second}
L.~A. Hemaspaandra and J.~Rothe.
\newblock A second step towards complexity-theoretic analogs of {R}ice's
  theorem.
\newblock {\em Theor. Comput. Sci.}, 244(1--2):205--217, 2000.

\bibitem{immerman1999descriptive}
N.~Immerman.
\newblock {\em Descriptive Complexity}.
\newblock Springer, 1999.

\bibitem{kolmogorov2004energy}
V.~Kolmogorov and R.~Zabih.
\newblock What energy functions can be minimized via graph cuts?
\newblock {\em IEEE Transactions on Pattern Analysis and Machine Intelligence},
  26(2):147--159, 2004.

\bibitem{libkin2004elements}
L.~Libkin.
\newblock {\em Elements of Finite Model Theory}.
\newblock Springer, 2004.

\bibitem{lokshtanov2011eth}
D.~Lokshtanov, D.~Marx, and S.~Saurabh.
\newblock Lower bounds based on the exponential time hypothesis.
\newblock {\em Bulletin of the EATCS}, (105):41--72, 2011.

\bibitem{myhill1955effective}
J.~Myhill and J.~C. Shepherdson.
\newblock Effective operations on partial recursive functions.
\newblock {\em Z. Math. Logik Grundlagen Math.}, 1(4):310--317, 1955.

\bibitem{nightingale2017savile}
P.~Nightingale, {\"O}.~Akg{\"u}n, I.~P. Gent, C.~Jefferson, I.~Miguel, and
  P.~Spracklen.
\newblock Automatically improving constraint models in {Savile Row}.
\newblock {\em Artificial Intelligence}, 251:35--61, 2017.

\bibitem{pawlak1982rough}
Z.~Pawlak.
\newblock Rough sets.
\newblock {\em Int. J. Comput. Inf. Sci.}, 11(5):341--356, 1982.

\bibitem{pawlak1994rough}
Z.~Pawlak and R.~S{\l}owi{\'n}ski.
\newblock Rough set approach to multi-attribute decision analysis.
\newblock {\em Eur. J. Oper. Res.}, 72(3):443--459, 1994.

\bibitem{rice1953classes}
H.~G. Rice.
\newblock Classes of recursively enumerable sets and their decision problems.
\newblock {\em Trans. Amer. Math. Soc.}, 74(2):358--366, 1953.

\bibitem{schaefer1978complexity}
T.~J. Schaefer.
\newblock The complexity of satisfiability problems.
\newblock In {\em STOC}, pages 216--226, 1978.

\bibitem{shannon1948mathematical1}
C.~E. Shannon.
\newblock A mathematical theory of communication.
\newblock {\em Bell Syst. Tech. J.}, 27(3):379--423, 1948.

\bibitem{shannon1948mathematical2}
C.~E. Shannon.
\newblock A mathematical theory of communication.
\newblock {\em Bell Syst. Tech. J.}, 27(4):623--656, 1948.

\bibitem{shannon1956zero}
C.~E. Shannon.
\newblock Zero-error capacity of a noisy channel.
\newblock {\em IRE Trans. Inf. Theory}, 2(3):8--19, 1956.

\bibitem{shannon1959coding}
C.~E. Shannon.
\newblock Coding theorems for a discrete source with a fidelity criterion.
\newblock {\em IRE Natl. Conv. Rec.}, 7(Part 4):142--163, 1959.

\bibitem{shapiro1956degrees}
N.~Shapiro.
\newblock Degrees of computability.
\newblock {\em Trans. Amer. Math. Soc.}, 82(2):281--299, 1956.

\bibitem{williams2003backdoors}
R.~Williams, C.~P. Gomes, and B.~Selman.
\newblock Backdoors to typical case complexity.
\newblock In {\em IJCAI}, pages 1173--1178, 2003.

\bibitem{yannakakis1991linear}
M.~Yannakakis.
\newblock Expressing combinatorial optimization problems by linear programs.
\newblock {\em Journal of Computer and System Sciences}, 43(3):441--466, 1991.

\bibitem{zhuk2017proof}
D.~Zhuk.
\newblock A proof of the {CSP} dichotomy conjecture.
\newblock In {\em 2017 IEEE 58th Annual Symposium on Foundations of Computer
  Science ({FOCS})}, pages 331--342. IEEE, 2017.

\end{thebibliography}
}

\end{document}

\title{Supplementary Material: \PaperTitleAuto}
\author{Tristan Simas}
\date{\today}
\maketitle

The archived artifact is available at \url{https://doi.org/10.5281/zenodo.19457896}. Lean statistics in the manuscript count the Paper 4d replay slice copied into the release package for the cited handles, including packaged dependencies, and exclude unrelated repository files.

\section{Compact Audit Map}\label{supsec:compact-audit-map}

Table~\ref{tab:compact-audit-map} gives a short route from the main theorem spine to representative formal declarations. The complete handle ledger appears in Section~\ref{supsec:full-lean-handle-ledger}.

\begingroup
\scriptsize
\setlength{\tabcolsep}{2pt}
\renewcommand{\arraystretch}{1.08}
\begin{longtable}{@{}p{0.16\linewidth}p{0.14\linewidth}p{0.39\linewidth}p{0.25\linewidth}@{}}
\caption{Compact audit map for the main mechanized claims. External source theorems, such as graph treewidth NP-completeness and Courcelle-style MSO model checking, are cited in the manuscript rather than mechanized here.}\label{tab:compact-audit-map}\\
\toprule
\textbf{Main claim family} & \textbf{Handles} & \textbf{Representative declarations} & \textbf{Source modules} \\
\midrule
\endfirsthead
\toprule
\textbf{Main claim family} & \textbf{Handles} & \textbf{Representative declarations} & \textbf{Source modules} \\
\midrule
\endhead
\bottomrule
\endfoot
Correct classifier closure agreement &
\LHrng{FR}{197}{199} &
\texttt{\seqsplit{correct\_classifier\_inherits\_closureLawInvariant}} &
\texttt{\seqsplit{Paper4dFrontier/AdmissibleCharacterization.lean}} \\
\addlinespace
Universal correctness transfer &
\LHrng{FR}{223}{231}, \LHrng{FR}{340}{341} &
\texttt{\seqsplit{feasiblePayloadSufficient\_iff\_lawDecisionProblem\_isSufficient}};
\texttt{\seqsplit{outputSemanticsRelevant\_iff\_totalizedLawDecisionProblem\_isRelevant}};
\texttt{\seqsplit{booleanPayloadTransfer}} &
\texttt{\seqsplit{Paper4dFrontier/Realizability.lean}} \\
\addlinespace
Compute-cost closure variants &
\LH{FR248}, \LHrng{FR}{259}{262}, \LHrng{FR}{316}{336} &
\texttt{\seqsplit{optimizerComputation\_polytime\_classifier\_agrees\_on\_closureEquivalent\_of\_correctOnDomain}};
\texttt{\seqsplit{optimizerSetSearch\_polytime\_classifier\_agrees\_on\_closureEquivalent}};
\texttt{\seqsplit{IdentityOutputClosureSpec.classifier\_agrees\_on\_closureEquivalent\_of\_correctOnDomain}} &
\texttt{\seqsplit{Paper4dFrontier/ComputeCostApplications.lean}};
\texttt{\seqsplit{Paper4dFrontier/ComputeCostExternalOutputs.lean}} \\
\addlinespace
Hull and quotient descent criteria &
\LH{FR355}, \LHrng{FR}{367}{375} &
\texttt{\seqsplit{closureLawInvariant\_iff\_closureHull\_eq\_self}};
\texttt{\seqsplit{not\_descendsToQuotient\_iff\_setoidGap}} &
\texttt{\seqsplit{Paper4dFrontier/AdmissibleCharacterization.lean}};
\texttt{\seqsplit{Paper4dFrontier/QuotientDescent.lean}} \\
\addlinespace
Raw local obstruction catalogue &
\LHrng{FR}{186}{193} &
\texttt{\seqsplit{no\_admissibleNormalizationPredicate\_decides\_dominantPair}};
\texttt{\seqsplit{no\_closureSoundPackage\_predicate\_decides\_each\_obstruction\_family}} &
\texttt{\seqsplit{Paper4dFrontier/ObstructionPredicateCandidates.lean}} \\
\addlinespace
Raw action, coordinate, and support gaps &
\LHrng{FR}{406}{410}, \LHrng{FR}{444}{454}, \LHrng{FR}{485}{490} &
\texttt{\seqsplit{rawSupportGraphPredicate\_orbitGap\_of\_graphPredicate\_separates}};
\texttt{\seqsplit{rawCoordinateCountThresholdGap\_irrelevantCoordinateExtension}} &
\texttt{\seqsplit{Paper4dFrontier/RawSupportGraphNoGo.lean}};
\texttt{\seqsplit{Paper4dFrontier/ObstructionPredicateCandidates.lean}} \\
\addlinespace
Action-gap graph algorithms and hardness transfer &
\LHrng{FR}{394}{396}, \LHrng{FR}{455}{457}, \LHrng{FR}{503}{508} &
\texttt{\seqsplit{actionGapTreewidthRecognition\_realizes\_graph}};
\texttt{\seqsplit{FixedWidthGraphSolver.solveActionGap\_steps\_bound}};
\texttt{\seqsplit{graphPredicate\_hardness\_transfers}} &
\texttt{\seqsplit{Paper4dFrontier/ActionGapTreewidthTractability.lean}};
\texttt{\seqsplit{Paper4dFrontier/HardnessTransfer.lean}} \\
\addlinespace
Split-proxy descent certification &
\LHrng{FR}{491}{502} &
\texttt{\seqsplit{sat3SplitProxy\_nonDescends\_iff\_satisfiable}};
\texttt{\seqsplit{unsat3\_to\_splitProxyDescent\_reduction}} &
\texttt{\seqsplit{Paper4dFrontier/HardnessTransfer.lean}} \\
\addlinespace
Quotient-MSO evaluation after quotient construction &
\LHrng{FR}{423}{435}, \LHrng{FR}{458}{475} &
\texttt{\seqsplit{QuotientMSOPostDescentEvaluator}};
\texttt{\seqsplit{BoundedQuotientTreewidthMSOEvaluator}} &
\texttt{\seqsplit{Paper4dFrontier/QuotientMSORealization.lean}};
\texttt{\seqsplit{Paper4dFrontier/BoundedQuotientTreewidthMSO.lean}} \\
\addlinespace
Sparse unary-gap certificate pipeline and stability &
\LHrng{FR}{263}{267}, \LHrng{FR}{479}{484} &
\texttt{\seqsplit{SparseUnaryGapQuotientPipeline}};
\texttt{\seqsplit{StrictMarginQuotientMSOStability.predicate\_eq\_of\_uniformApprox}} &
\texttt{\seqsplit{Paper4dFrontier/SparseUnaryGapQuotientPipeline.lean}};
\texttt{\seqsplit{Paper4dFrontier/Realizability.lean}} \\
\end{longtable}
\endgroup

\section{Routine Preservation Facts}\label{supsec:routine-preservation}

The packaged closure-law proposition uses the following preservation facts.

\paragraph{Positive affine reparameterization.}
Statewise positive affine transformations of utility leave exact certification unchanged. In particular, replacing
\[
U(a,s)
\qquad\text{by}\qquad
\alpha(s)+\beta(s)U(a,s)
\]
with $\beta(s)>0$ for every state $s$ does not change sufficient coordinate sets or relevance judgments. Order preservation at each state keeps all action comparisons unchanged; optimizer sets, quotient classes, sufficient sets, and relevant coordinates are unchanged. Lean handles: \LHrng{FR}{13}{16}.

\paragraph{Duplication.}
Duplicating a single action without changing its utility profile preserves the decision quotient relation and hence preserves sufficiency and relevance. At the optimizer-set level, the only possible change is the addition of the duplicate action itself: original optimal actions remain optimal, and the duplicate is optimal exactly when the original action was. Duplicating a single state without changing its utility profile preserves the decision quotient up to the induced quotient equivalence and also preserves sufficiency and relevance. Lean handles: \LHrng{FR}{54}{56}, \LHrng{FR}{68}{73}.

\paragraph{Relabeling.}
Action relabeling preserves optimizer equivalence and exact sufficiency. Coordinate relabeling does so through the induced state relabeling, with the coordinate structure transported along the state bijection. Consequently, any structural tractability criterion for exact relevance certification must be closed under these relabelings. Lean handles: \LHrng{FR}{17}{20}.

\paragraph{Finite-structural inhabitation.}
The finite structural class is nonempty. Both constant instance predicates belong to it: they are decidable in constant time, automatically closure-sound, and definable by trivial bounded local pattern schemes. One may use an impossible local pattern, such as a pattern requiring two distinct vertices both at root-distance zero; using it as a forbidden pattern yields the constant-true predicate, and using it as a required witness yields the constant-false predicate. Lean handle: \LH{FR200}.

\paragraph{Deterministic payload transfer.}
Let $S$ be a coordinate space, let $\phi:S\to T$ be a deterministic payload map into a finite label type, and form the induced decision problem with action space $T$ and utility
\[
U(a,s)=
\begin{cases}
1,&a=\phi(s),\\
0,&a\neq\phi(s).
\end{cases}
\]
Then every coordinate set is sufficient for $\phi$ if and only if it is sufficient for the induced decision problem, and every coordinate is relevant for $\phi$ if and only if it is relevant for the induced decision problem. The optimizer at state $s$ is the singleton $\{\phi(s)\}$, so payload equality and optimizer-set equality define the same quotient. Lean handles: \LHrng{FR}{217}{222}.

\section{Formal Bounded-Pattern Encoding}\label{supsec:bounded-pattern-encoding}

The finite-structural predicate definition uses the following formal interface. Let $|U|$ be the encoding size of a binary pairwise instance and let $X(U)$ be its canonical finite pairwise syntax. Polynomial-time checkability means that there are constants $c,c',k$ and a decision procedure $A$ such that
\[
A(U)=1 \iff Q(U),
\qquad
\mathrm{time}_A(U)\le c(|U|+1)^k+c'.
\]
Structural extractability means that the dependency graph used by $Q$ is obtained from syntax alone by a uniform extractor $E:X\mapsto G_E(X)$.

\paragraph{Bounded local schemes.}
Bounded-pattern definability supplies fixed integers
\[
r_{\max},\; n_{\max},\; a_{\max},\; c_{\max}\in\mathbb N
\]
and finite rooted pattern sets $\mathcal W,\mathcal F$. Every pattern in $\mathcal W\cup\mathcal F$ has radius at most $r_{\max}$, at most $n_{\max}$ vertices, at most $a_{\max}$ action labels, and coefficient magnitudes at most $c_{\max}$. Write
\[
\mathcal N_U(v):=\mathcal N_{r_{\max}}(X(U),v)
\]
for the rooted radius-$r_{\max}$ neighborhood of a vertex $v$ in the extracted syntax. Occurrence $P\sqsubseteq\mathcal N_U(v)$ means a root-preserving injective embedding of $P$ into $\mathcal N_U(v)$ together with a bijection between the pattern action labels and the action labels visible in the extracted instance structure for that occurrence, preserving unary coefficients, binary coefficients, and incidences. The predicate is then determined by
\[
\begin{aligned}
\mathsf W_U &:=
\mathcal W\neq\emptyset\land
\exists v\,\exists P\in\mathcal W:\;
P\sqsubseteq\mathcal N_U(v),\\
\mathsf F_U &:=
\mathcal F\neq\emptyset\land
\forall v\,\forall P'\in\mathcal F:\;
P'\not\sqsubseteq\mathcal N_U(v).
\end{aligned}
\]
The predicate is then
\[
Q(U)\iff \mathsf W_U\lor \mathsf F_U.
\]
An empty witness or forbidden family makes the corresponding branch false.

\paragraph{FO presentation.}
Treat $X(U)$ as a finite relational structure $\mathcal S(U)$ whose signature records incidence, action labels, coefficient labels, root position, and bounded coefficient values. For each fixed rooted pattern $P$, one builds an existential first-order formula $\theta_P(x)$ naming the finitely many pattern vertices and action labels, enforcing distinctness, the root condition, all required incidences and coefficient labels, and the action-label bijection for the occurrence. Since all bounds are fixed, $\theta_P$ is a single FO formula, not an input-dependent schema. Occurrence of $P$ is equivalent to satisfaction of $\theta_P$ at the root. Boolean combinations of occurrence tests and fixed occurrence thresholds give the counted basic-local FO fragment, and conversely each such formula records a finite bounded-pattern condition. Lean handles: \LHrng{FR}{363}{366}.

\section{Positive Cases: Domains Where Hull Separation Holds}\label{supsec:positive-cases}

Positive classification uses the same orbit algebra as the obstruction theorem: choose a closure-closed domain $D$, a target predicate $Q$, and polynomial-time-computable closure transports. If the positive and negative orbit hulls are disjoint, the exact-classification hull-separation theorem\leanmeta{\LHrng{FR}{356}{357}} gives an exact closure-invariant classifier, and the least-classifier corollary\leanmeta{\LH{FR358}} gives the least one. The remaining question is whether that hull is accessible by usable structural data.

\paragraph{Normalized unary-gap domains.}
One basic positive subdomain consists of normalized binary instances whose action gaps have no pair interaction:
\[
\Delta_{ij}(a)-\Delta_{ij}(b)=0
\qquad\text{for all } i\neq j \text{ and actions } a,b.
\]
Equivalently, any pairwise or higher-order terms are action-independent and cancel from optimizer comparisons; optimizer behavior is represented by unary action-dependent contributions. Closure is understood domain-relatively: only relabelings, duplications, irrelevant-coordinate extensions, and positive affine transports whose source and target remain in this normalized unary-gap domain are included. On this domain, closure orbits cannot create the pairwise action-gap patterns used by the four obstruction families. Targets that factor through the normalized unary-gap catalogue have no orbit gap, and the catalogue gives an exact Boolean classifier.\leanmeta{\LHrng{FR}{390}{393}} With fixed action alphabet, bounded coefficient vocabulary, and explicitly listed finite support, the catalogue is enumerated by finite tables of action-gap vectors. The unary side of the binary pairwise dichotomy is therefore a concrete hull-separated regime.

\paragraph{Bounded distinct-profile compression.}
Let $D_k$ be the domain of instances with at most $k$ distinct action utility profiles. Relabelings, statewise positive affine reparameterizations, duplication, and binary irrelevant-coordinate extension preserve equality of action profiles, so $D_k$ is closure-closed under the listed transports. Profile compression maps every instance in $D_k$ to an equivalent instance with at most $k$ actions and preserves all sufficient sets and relevant coordinates. Hence any polynomial-time exact-certification procedure for bounded-action instances transfers to $D_k$. The positive content is a closure-compatible reduction that removes accidental action multiplicity before classification.

\paragraph{Strict-margin stability regimes.}
Margin assumptions give a third positive regime, at the interface between exact classification and approximate presentation. Suppose every state has a unique optimizer and a uniform winner-versus-competitor gap exceeding $2\delta$. The global strict-margin stability proposition\leanmeta{\LH{FR288}, \LH{FR289}, \LH{FR290}, \LH{FR300}, \LH{FR301}} shows that every utility presentation within uniform distance $\delta$ has the same optimizer quotient, sufficient-coordinate family, relevant-coordinate set, and minimal sufficient sets. Thus any exact classifier already established on a normalized representative extends throughout the corresponding strict-margin neighborhood. The role of the margin hypothesis is quotient preservation: approximate utilities become safe inputs for exact relevance only after the gap condition pins down the admissible-output classes.

For the finite state/action utility model above and unrestricted sup-norm perturbation balls, the strict-gap condition is the complete robustness criterion for preserving exact optimizer sets throughout the ball. If a state has a tie among optimizers, an arbitrarily small action-specific perturbation can break the tie and change the optimizer set. If the winner-versus-runner-up gap is at most $2\delta$, lowering the winner by $\delta$ and raising the competitor by $\delta$ can create a tie or flip the winner. Hence full-ball exact stability is exactly the unique-optimizer regime with gap greater than $2\delta$. Weaker margin hypotheses can work only after restricting the admissible perturbation family, for example to perturbations that preserve ties inside an optimal face.

These positive examples mark three mechanisms by which hull separation can hold: unary-gap domains remove decision-relevant pair interactions before orbit gaps arise; bounded distinct-profile domains compress accidental action multiplicity before classification; and strict-margin neighborhoods stabilize the optimizer quotient under controlled perturbation. A systematic positive theory would have to supply closure-invariant algebraic or canonical data proving hull separation on broader domains.

\section{Binary-Pairwise Orientation Facts}\label{supsec:binary-pairwise-orientation}

\paragraph{Symmetric pairwise dichotomies.}
In a binary pairwise presentation invariant under coordinate permutations, raw pair interaction is empty or present on every coordinate pair. The same dichotomy holds for decision-relevant interaction after passing to action-gap mixed differences and removing action-independent base state terms. Consequently, the symmetric decision-relevant interaction graph is either edgeless or complete. The proof uses transitivity of the coordinate-permutation action on unordered coordinate pairs: a nonzero mixed difference on one pair transports to every pair, and if no pair has nonzero mixed difference then the pairwise part is unary after the stated normalization. Lean handles: \LHrng{FR}{118}{124}, \LH{FR359}.

\paragraph{Offset normalization.}
Two pairwise presentations are in the offset-normalization relation when
\[
V(a,s)=U(a,s)+\alpha(s)+\kappa(a).
\]
The offset-normalized decision-relevant interaction graph is computed from mixed differences of action gaps $U(a,\cdot)-U(b,\cdot)$, hence is invariant on each offset-normalization class. The action-specific constants $\kappa(a)$ can change optimizer behavior, so this graph is a normalization device rather than a correctness-preserving closure law by itself. Lean handles: \LHrng{FR}{125}{129}.

\paragraph{Quotient-cost reading.}
If every state has the same optimizer set, then the optimizer quotient has one class. Any cost computed on quotient classes has only the diagonal case and is zero. If the quotient has genuine branching, moving between distinct quotient classes has positive cost. Lean handles: \LHrng{FR}{183}{184}.

\section{Binary-Pairwise Witness Details}\label{supsec:binary-pairwise-witness-details}

The affine-shift witness table uses the following arithmetic. The common form is
\[
V(c,x)=U(c,x)+\alpha(x)
\]
for every action $c$, with $\alpha$ action-independent. Hence every action gap is unchanged:
\[
V(a,x)-V(b,x)=U(a,x)-U(b,x).
\]
The same-orbit premise is the positive-affine closure step with scale $1$.

\paragraph{Dominant pair.}
Use three binary coordinates and actions $a,b$:
\[
U(a,x)=2x_0x_1,\qquad U(b,x)=0,
\qquad \alpha(x)=3x_1x_2.
\]
The only nonzero pair interaction in $U$ occurs on $\{0,1\}$ for action $a$, so anchored unique dominant-pair status holds. The action gap is $2x_0x_1$. It is unchanged in $V$, so the optimizer quotient, sufficient-coordinate family, and relevant-coordinate set are unchanged. In $V$, the mixed-difference magnitudes include $\Delta_{01}(a)=2$, $\Delta_{12}(a)=3$, and $\Delta_{12}(b)=3$. The largest pair/action magnitude is no longer uniquely achieved at $(\{0,1\},a)$, so the anchored target flips. Lean handle: \LH{FR186}.

\paragraph{Margin.}
Use actions $a,b$:
\[
U(a,x)=3x_0+x_1x_2,\qquad U(b,x)=0,
\qquad \alpha(x)=2x_0x_1.
\]
The action gap is unchanged after adding $\alpha$. The largest unary coefficient remains $3$. The largest pair mixed-difference changes from $1$ to $2$, so the predicate requiring unary magnitude at most twice the largest pair magnitude flips from false to true. Lean handle: \LH{FR187}.

\paragraph{Ghost action.}
Use actions $0,1,2$:
\[
U(0,x)=\mathbf 1[x_0=1],\quad
U(1,x)=\mathbf 1[x_0=0],\quad
U(2,x)=-1-x_0x_1,
\qquad \alpha(x)=2\sigma(x_0,x_1),
\]
where $\sigma(x_0,x_1)=1$ when $x_0=x_1$ and $-1$ otherwise. Action $2$ is never optimal: if $x_0=1$, action $0$ has value $1$; if $x_0=0$, action $1$ has value $1$; action $2$ is at most $-1$. The common shift keeps all action gaps and optimizer sets fixed. The mixed difference of $\sigma$ on the anchor square is $4$, so the mixed difference of $2\sigma$ is $8$. The ghost action's anchored signature therefore changes while exact certification is fixed. Lean handle: \LH{FR188}.

\paragraph{Offset signature.}
Use actions $f,t$:
\[
U(f,x)=1+\mathbf 1[x_0=1]+x_0x_1,\qquad U(t,x)=0,
\qquad \alpha(x)=2\sigma(x_0,x_1).
\]
Before the shift, the anchor-pair mixed-difference magnitudes of the two actions are $1$ and $0$. The common shift adds anchor mixed-difference magnitude $8$ to both, changing the magnitudes to $9$ and $8$. The action gaps are unchanged, so the offset-signature target flips inside one closure orbit. Lean handle: \LH{FR189}.

The same dominant-pair witness gives the optimizer-computation no-go. Lean handles: \LH{FR250}, \LH{FR251}.
The affine witness dimensions are minimal for this template: a nonzero action-gap pair interaction requires at least two actions, and the two-pair affine template requires at least three coordinates. Lean handles: \LH{FR397}, \LH{FR398}, \LH{FR399}, \LH{FR400}.

\section{Claim-to-Lean Handle Mapping}

Each paper claim below is matched to its Lean formalization. Rows involving external complexity facts separate the mechanized reduction or realization from the cited source theorem. For example, the \textsc{Action-Gap-Treewidth} row maps the graph-realization construction and treewidth-preserving equivalence to Lean handles; NP-completeness of graph treewidth recognition is the external theorem cited in the main text.

\IfFileExists{content/claim_mapping_auto.tex}{%
  % Auto-generated by scripts/build_papers.py. Do not edit manually.
% Generated: 2026-07-01T12:55:36.240156
% Manuscript: paper4d
\begingroup
\scriptsize
\setlength{\tabcolsep}{3pt}
\renewcommand{\arraystretch}{1.0}
\setlength{\LTpre}{0pt}
\setlength{\LTpost}{0pt}
\begin{longtable}{@{}>{\raggedright\arraybackslash}m{0.65\linewidth}>{\raggedleft\arraybackslash}m{0.30\linewidth}@{}}
\toprule
\textbf{Manuscript claim} & \textbf{Lean handle} \\
\midrule
\endfirsthead
\toprule
\textbf{Manuscript claim} & \textbf{Lean handle} \\
\midrule
\endhead
\endfoot
\bottomrule
\endlastfoot
Proposition 3.5: Quotient-Dependence Facts & \LH{FR30}, \LH{FR31}, \LH{FR32}, \LH{FR33}, \LH{FR34}, \LH{FR35}, \LH{FR37}, \LH{FR51}, \LH{FR176}, \LH{FR177} \\
\midrule
Proposition 3.6: Minimal Sufficient Sets Are Canonical on Product Spaces & \LH{FR352} \\
\midrule
Proposition 3.7: Closure Operations Preserve Exact Certification & \LH{FR15}, \LH{FR16}, \LH{FR17}, \LH{FR18}, \LH{FR19}, \LH{FR20}, \LH{FR55}, \LH{FR56}, \LH{FR57}, \LH{FR58}, \LH{FR59}, \LH{FR60}, \LH{FR83}, \LH{FR84}, \LH{FR85} \\
\midrule
Theorem 3.11: Correctness Implies Closure-Orbit Invariance & \LH{FR197}, \LH{FR198}, \LH{FR199} \\
\midrule
Theorem 3.12: Compute-Cost Closure-Orbit Invariance & \LH{FR248}, \LH{FR259}, \LH{FR260}, \LH{FR261}, \LH{FR262}, \LH{FR316}, \LH{FR317}, \LH{FR318}, \LH{FR319}, \LH{FR320}, \LH{FR321}, \LH{FR322}, \LH{FR323}, \LH{FR324}, \LH{FR325}, \LH{FR326}, \LH{FR327}, \LH{FR328}, \LH{FR329}, \LH{FR330}, \LH{FR331}, \LH{FR332}, \LH{FR333}, \LH{FR334}, \LH{FR335}, \LH{FR336} \\
\midrule
Proposition 3.13: Bounded Distinct Action Profiles Compress to Bounded Actions & \LH{FR206}, \LH{FR207}, \LH{FR208}, \LH{FR209} \\
\midrule
Theorem 4.1: Universal Correctness Transfer & \LH{FR223}, \LH{FR224}, \LH{FR225}, \LH{FR226}, \LH{FR227}, \LH{FR228}, \LH{FR229}, \LH{FR230}, \LH{FR231}, \LH{FR340}, \LH{FR341} \\
\midrule
Corollary 4.3: Correctness Quotient Universality and Coordinate Presentations & \LH{FR197}, \LH{FR198}, \LH{FR199}, \LH{FR232}, \LH{FR233}, \LH{FR234}, \LH{FR235}, \LH{FR236}, \LH{FR237}, \LH{FR238}, \LH{FR239}, \LH{FR240}, \LH{FR241}, \LH{FR242}, \LH{FR243}, \LH{FR244}, \LH{FR252}, \LH{FR253}, \LH{FR254}, \LH{FR255}, \LH{FR256}, \LH{FR268}, \LH{FR269}, \LH{FR270}, \LH{FR271}, \LH{FR272}, \LH{FR273}, \LH{FR274}, \LH{FR275}, \LH{FR276}, \LH{FR277}, \LH{FR278}, \LH{FR279}, \LH{FR280}, \LH{FR281}, \LH{FR282}, \LH{FR283}, \LH{FR284}, \LH{FR285}, \LH{FR286}, \LH{FR287}, \LH{FR288}, \LH{FR289}, \LH{FR290}, \LH{FR291}, \LH{FR292}, \LH{FR293}, \LH{FR294}, \LH{FR295}, \LH{FR296}, \LH{FR297}, \LH{FR298}, \LH{FR299}, \LH{FR300}, \LH{FR301}, \LH{FR302}, \LH{FR303}, \LH{FR304}, \LH{FR305}, \LH{FR306}, \LH{FR307}, \LH{FR308}, \LH{FR309}, \LH{FR310}, \LH{FR311}, \LH{FR312}, \LH{FR313}, \LH{FR314}, \LH{FR315}, \LH{FR345}, \LH{FR351} \\
\midrule
Lemma 5.1: Primitive-Law Invariance Equals Orbit Invariance & \LH{FR362} \\
\midrule
Proposition 5.2: Orbit-Gap Descent Criterion & \LH{FR201}, \LH{FR202}, \LH{FR203}, \LH{FR204}, \LH{FR210}, \LH{FR211}, \LH{FR212}, \LH{FR213}, \LH{FR367}, \LH{FR368}, \LH{FR369}, \LH{FR370}, \LH{FR371}, \LH{FR372}, \LH{FR373}, \LH{FR374}, \LH{FR375} \\
\midrule
Proposition 5.3: Hull Is a Closure Operator; Closure-Invariant Predicates Are Its Fixed Points & \LH{FR355} \\
\midrule
Theorem 5.4: Exact Classification Equals Hull Separation & \LH{FR356}, \LH{FR357} \\
\midrule
Corollary 5.5: Least Exact Closure-Invariant Classifier & \LH{FR358} \\
\midrule
Theorem 5.6: Computable Orbit Catalogues Make Hull Classifiers Algorithmic & \LH{FR376}, \LH{FR377}, \LH{FR378}, \LH{FR379}, \LH{FR380}, \LH{FR381}, \LH{FR382}, \LH{FR383} \\
\midrule
Corollary 5.7: Catalogue Classifiers Are Compositional & \LH{FR384}, \LH{FR385}, \LH{FR386}, \LH{FR387}, \LH{FR388}, \LH{FR389} \\
\midrule
Corollary 5.8: Domain Restriction Helps Only by Removing Orbit Gaps & \LH{FR252}, \LH{FR253}, \LH{FR254}, \LH{FR255}, \LH{FR256} \\
\midrule
Theorem 5.9: Normalized Unary-Gap Catalogues Give a Positive Regime & \LH{FR390}, \LH{FR391}, \LH{FR392}, \LH{FR393} \\
\midrule
Theorem 5.10: Bounded Distinct-Profile Compression Gives a Positive Regime & \LH{FR206}, \LH{FR207}, \LH{FR208}, \LH{FR209} \\
\midrule
Corollary 5.11: Orbit-Gap Template & \LH{FR185} \\
\midrule
Proposition 6.1: Action-Gap Graphs Realize Arbitrary Treewidth Instances & \LH{FR394}, \LH{FR395}, \LH{FR396} \\
\midrule
Theorem 6.2: Graph Predicate Hardness Transfers to Action-Gap Predicates & \LH{FR491}, \LH{FR492}, \LH{FR503}, \LH{FR504}, \LH{FR505}, \LH{FR506}, \LH{FR507}, \LH{FR508} \\
\midrule
Theorem 6.4: \textup{\textsc{Action-Gap-Treewidth} is NP-complete} & \LH{FR394}, \LH{FR395}, \LH{FR396} \\
\midrule
Theorem 6.5: SAT Reduces to Split-Proxy Non-Descent & \LH{FR491}, \LH{FR492}, \LH{FR493}, \LH{FR494}, \LH{FR495}, \LH{FR496}, \LH{FR497}, \LH{FR498}, \LH{FR499}, \LH{FR500}, \LH{FR501}, \LH{FR502} \\
\midrule
Theorem 6.10: Raw Action Count Threshold Duplication Orbit Gap & \LH{FR406}, \LH{FR407}, \LH{FR408}, \LH{FR409}, \LH{FR410} \\
\midrule
Theorem 6.11: Raw Coordinate Count Threshold Irrelevant-Coordinate Orbit Gap & \LH{FR485}, \LH{FR486}, \LH{FR487}, \LH{FR488}, \LH{FR489}, \LH{FR490} \\
\midrule
Lemma 6.12: Affine-Shift Obstruction Lemma & \LH{FR15}, \LH{FR16}, \LH{FR17}, \LH{FR18}, \LH{FR19}, \LH{FR20}, \LH{FR185} \\
\midrule
Theorem 6.13: Parameterized Closure-Sound No-Go & \LH{FR193} \\
\midrule
Corollary 6.14: No Simultaneous Finite-Structural Classification and Target Recognition & \LH{FR193}, \LH{FR197}, \LH{FR198}, \LH{FR199} \\
\midrule
Proposition 6.15: Full Binary Pairwise Domain Is Closure-Closed & \LH{FR348}, \LH{FR349} \\
\midrule
Corollary 6.16: No Simultaneous Finite-Structural Classification and Raw-Target Recognition on the Full Binary Pairwise Domain & \LH{FR350} \\
\midrule
Corollary 6.17: Universal Exact-Certification Treatments Inherit the Binary Pairwise Obstruction by Restriction & \LH{FR351}, \LH{FR361} \\
\midrule
Theorem 7.2: Finite Coordinate-State Quotient Flip Graph Realization & \LH{FR411}, \LH{FR412}, \LH{FR413}, \LH{FR414}, \LH{FR415}, \LH{FR416}, \LH{FR417}, \LH{FR418}, \LH{FR419}, \LH{FR420}, \LH{FR421}, \LH{FR422} \\
\midrule
Theorem 7.4: Quotient-MSO Evaluation After Construction & \LH{FR423}, \LH{FR424}, \LH{FR425}, \LH{FR426}, \LH{FR427}, \LH{FR428}, \LH{FR429}, \LH{FR430}, \LH{FR431}, \LH{FR432}, \LH{FR433}, \LH{FR434}, \LH{FR435}, \LH{FR458}, \LH{FR459} \\
\midrule
Theorem 7.5: Bounded-Quotient-Size MSO Evaluation After Construction & \LH{FR470}, \LH{FR471}, \LH{FR472}, \LH{FR473}, \LH{FR474}, \LH{FR475} \\
\midrule
Theorem 7.6: Bounded-Quotient-Treewidth MSO Evaluation After Construction & \LH{FR460}, \LH{FR461}, \LH{FR462}, \LH{FR463}, \LH{FR464}, \LH{FR465}, \LH{FR466}, \LH{FR467}, \LH{FR468}, \LH{FR469} \\
\midrule
Theorem 7.8: Sparse Unary-Gap Quotient Pipeline Certificate & \LH{FR479}, \LH{FR480}, \LH{FR481}, \LH{FR482}, \LH{FR483}, \LH{FR484} \\
\midrule
Proposition 8.1: Approximate Relevance and Sufficiency Claims Need Explicit Stability Control & \LH{FR263}, \LH{FR266} \\
\midrule
Corollary 8.2: Arbitrarily Small Uniform Perturbations Can Flip Exact Certification & \LH{FR264}, \LH{FR267} \\
\midrule
Proposition 8.3: Global Approximation Stability Under Uniform Strict Gaps & \LH{FR288}, \LH{FR289}, \LH{FR290}, \LH{FR300}, \LH{FR301} \\
\midrule
Corollary 8.4: Strict-Margin Quotient-MSO Stability & \LH{FR476}, \LH{FR477}, \LH{FR478} \\
\end{longtable}
\endgroup

}{%
  \textbf{Error:} \texttt{content/claim\_mapping\_auto.tex} not found.
}

\section{Full Lean Handle Ledger}\label{supsec:full-lean-handle-ledger}

This supplement provides the complete Lean handle ledger cited by the manuscript.
It includes every handle identifier, declaration name, and source module path.
Handle IDs are sparse because uncited formalization lemmas are omitted from the manuscript ledger; the full artifact is available at the Zenodo DOI above.
The ledger references two formalization trees, \texttt{Paper4dFrontier/} and \texttt{paper4/DecisionQuotient/}, both included in the archived artifact.
The main text calls the direct local proxy class \emph{finite-structural predicates}. Declaration names containing \texttt{Admissible\allowbreak Normalization\allowbreak Predicate} formalize the same finite-structural predicate class and its no-go witnesses.

\IfFileExists{content/lean_handle_ids_auto.tex}{%
  % Auto-generated by scripts/build_papers.py. Do not edit manually.
% Generated: 2026-07-01T12:55:36.034734
\begingroup
\scriptsize
\setlength{\tabcolsep}{4pt}
\renewcommand{\arraystretch}{1.12}
\setlength{\LTpre}{2pt}
\setlength{\LTpost}{2pt}
\setlength{\emergencystretch}{3em}
\sloppy
\urlstyle{tt}
\makeatletter
\if@twocolumn
\begin{list}{}{\leftmargin=0pt\itemindent=0pt\itemsep=4pt\parsep=0pt\topsep=4pt}
\item \textbf{\nolinkurl{FR7}}\hypertarget{lh:FR7}{}\enspace{\ttfamily\nolinkurl{exists_decisionProblem_realizing_labeling}} {\tiny\ttfamily Paper4dFrontier/\allowbreak Realizability.lean}
\item \textbf{\nolinkurl{FR8}}\hypertarget{lh:FR8}{}\enspace{\ttfamily\nolinkurl{realizingProblem_decisionEquiv_iff}} {\tiny\ttfamily Paper4dFrontier/\allowbreak Realizability.lean}
\item \textbf{\nolinkurl{FR13}}\hypertarget{lh:FR13}{}\enspace{\ttfamily\nolinkurl{isOptimal_positiveAffineTransform_iff}} {\tiny\ttfamily Paper4dFrontier/\allowbreak FamilyAxioms.lean}
\item \textbf{\nolinkurl{FR14}}\hypertarget{lh:FR14}{}\enspace{\ttfamily\nolinkurl{opt_eq_positiveAffineTransform}} {\tiny\ttfamily Paper4dFrontier/\allowbreak FamilyAxioms.lean}
\item \textbf{\nolinkurl{FR15}}\hypertarget{lh:FR15}{}\enspace{\ttfamily\nolinkurl{isSufficient_positiveAffineTransform_iff}} {\tiny\ttfamily Paper4dFrontier/\allowbreak FamilyAxioms.lean}
\item \textbf{\nolinkurl{FR16}}\hypertarget{lh:FR16}{}\enspace{\ttfamily\nolinkurl{isRelevant_positiveAffineTransform_iff}} {\tiny\ttfamily Paper4dFrontier/\allowbreak FamilyAxioms.lean}
\item \textbf{\nolinkurl{FR17}}\hypertarget{lh:FR17}{}\enspace{\ttfamily\nolinkurl{decisionEquiv_relabelActions_iff}} {\tiny\ttfamily Paper4dFrontier/\allowbreak FamilyAxioms.lean}
\item \textbf{\nolinkurl{FR18}}\hypertarget{lh:FR18}{}\enspace{\ttfamily\nolinkurl{decisionEquiv_relabelStates_iff}} {\tiny\ttfamily Paper4dFrontier/\allowbreak FamilyAxioms.lean}
\item \textbf{\nolinkurl{FR19}}\hypertarget{lh:FR19}{}\enspace{\ttfamily\nolinkurl{isSufficient_relabelActions_iff}} {\tiny\ttfamily Paper4dFrontier/\allowbreak FamilyAxioms.lean}
\item \textbf{\nolinkurl{FR20}}\hypertarget{lh:FR20}{}\enspace{\ttfamily\nolinkurl{isSufficient_relabelStates_iff}} {\tiny\ttfamily Paper4dFrontier/\allowbreak FamilyAxioms.lean}
\item \textbf{\nolinkurl{FR21}}\hypertarget{lh:FR21}{}\enspace{\ttfamily\nolinkurl{allProblems_relabelInvariant}} {\tiny\ttfamily Paper4dFrontier/\allowbreak FamilyAxioms.lean}
\item \textbf{\nolinkurl{FR22}}\hypertarget{lh:FR22}{}\enspace{\ttfamily\nolinkurl{allProblems_kernelUniversal}} {\tiny\ttfamily Paper4dFrontier/\allowbreak FamilyAxioms.lean}
\item \textbf{\nolinkurl{FR23}}\hypertarget{lh:FR23}{}\enspace{\ttfamily\nolinkurl{relabelInvariant_not_enough}} {\tiny\ttfamily Paper4dFrontier/\allowbreak FamilyAxioms.lean}
\item \textbf{\nolinkurl{FR24}}\hypertarget{lh:FR24}{}\enspace{\ttfamily\nolinkurl{optRangeCard_realizingIdentity}} {\tiny\ttfamily Paper4dFrontier/\allowbreak FamilyAxioms.lean}
\item \textbf{\nolinkurl{FR30}}\hypertarget{lh:FR30}{}\enspace{\ttfamily\nolinkurl{isSufficient_iff_of_decisionEquiv_iff}} {\tiny\ttfamily Paper4dFrontier/\allowbreak FamilyAxioms.lean}
\item \textbf{\nolinkurl{FR31}}\hypertarget{lh:FR31}{}\enspace{\ttfamily\nolinkurl{isRelevant_iff_of_decisionEquiv_iff}} {\tiny\ttfamily Paper4dFrontier/\allowbreak FamilyAxioms.lean}
\item \textbf{\nolinkurl{FR32}}\hypertarget{lh:FR32}{}\enspace{\ttfamily\nolinkurl{isIrrelevant_iff_of_decisionEquiv_iff}} {\tiny\ttfamily Paper4dFrontier/\allowbreak FamilyAxioms.lean}
\item \textbf{\nolinkurl{FR33}}\hypertarget{lh:FR33}{}\enspace{\ttfamily\nolinkurl{isMinimalSufficient_iff_of_decisionEquiv_iff}} {\tiny\ttfamily Paper4dFrontier/\allowbreak FamilyAxioms.lean}
\item \textbf{\nolinkurl{FR34}}\hypertarget{lh:FR34}{}\enspace{\ttfamily\nolinkurl{sufficientSets_eq_of_decisionEquiv_iff}} {\tiny\ttfamily Paper4dFrontier/\allowbreak FamilyAxioms.lean}
\item \textbf{\nolinkurl{FR35}}\hypertarget{lh:FR35}{}\enspace{\ttfamily\nolinkurl{relevantSet_eq_of_decisionEquiv_iff}} {\tiny\ttfamily Paper4dFrontier/\allowbreak FamilyAxioms.lean}
\item \textbf{\nolinkurl{FR37}}\hypertarget{lh:FR37}{}\enspace{\ttfamily\nolinkurl{quotientCard_eq_of_decisionEquiv_iff}} {\tiny\ttfamily Paper4dFrontier/\allowbreak FamilyAxioms.lean}
\item \textbf{\nolinkurl{FR38}}\hypertarget{lh:FR38}{}\enspace{\ttfamily\nolinkurl{quotientCard_realizingIdentity}} {\tiny\ttfamily Paper4dFrontier/\allowbreak FamilyAxioms.lean}
\item \textbf{\nolinkurl{FR39}}\hypertarget{lh:FR39}{}\enspace{\ttfamily\nolinkurl{realizingProblemQuotientEquivLabelRange_apply_quotientMap}} {\tiny\ttfamily Paper4dFrontier/\allowbreak Realizability.lean}
\item \textbf{\nolinkurl{FR40}}\hypertarget{lh:FR40}{}\enspace{\ttfamily\nolinkurl{realizingProblem_quotientCard_eq_labelRangeCard}} {\tiny\ttfamily Paper4dFrontier/\allowbreak Realizability.lean}
\item \textbf{\nolinkurl{FR45}}\hypertarget{lh:FR45}{}\enspace{\ttfamily\nolinkurl{setoidRealizingProblem_decisionEquiv_iff}} {\tiny\ttfamily Paper4dFrontier/\allowbreak Realizability.lean}
\item \textbf{\nolinkurl{FR46}}\hypertarget{lh:FR46}{}\enspace{\ttfamily\nolinkurl{setoidRealizingProblemQuotientEquivSetoidQuotient_apply_quotientMap}} {\tiny\ttfamily Paper4dFrontier/\allowbreak Realizability.lean}
\item \textbf{\nolinkurl{FR47}}\hypertarget{lh:FR47}{}\enspace{\ttfamily\nolinkurl{exists_decisionProblem_realizing_setoid}} {\tiny\ttfamily Paper4dFrontier/\allowbreak Realizability.lean}
\item \textbf{\nolinkurl{FR51}}\hypertarget{lh:FR51}{}\enspace{\ttfamily\nolinkurl{quotientEquiv_of_decisionEquiv_iff_apply_quotientMap}} {\tiny\ttfamily Paper4dFrontier/\allowbreak FamilyAxioms.lean}
\item \textbf{\nolinkurl{FR54}}\hypertarget{lh:FR54}{}\enspace{\ttfamily\nolinkurl{decisionEquiv_duplicateAction_iff}} {\tiny\ttfamily Paper4dFrontier/\allowbreak ClosureLaws.lean}
\item \textbf{\nolinkurl{FR55}}\hypertarget{lh:FR55}{}\enspace{\ttfamily\nolinkurl{isSufficient_duplicateAction_iff}} {\tiny\ttfamily Paper4dFrontier/\allowbreak ClosureLaws.lean}
\item \textbf{\nolinkurl{FR56}}\hypertarget{lh:FR56}{}\enspace{\ttfamily\nolinkurl{isRelevant_duplicateAction_iff}} {\tiny\ttfamily Paper4dFrontier/\allowbreak ClosureLaws.lean}
\item \textbf{\nolinkurl{FR57}}\hypertarget{lh:FR57}{}\enspace{\ttfamily\nolinkurl{decisionEquiv_duplicateState_iff}} {\tiny\ttfamily Paper4dFrontier/\allowbreak ClosureLaws.lean}
\item \textbf{\nolinkurl{FR58}}\hypertarget{lh:FR58}{}\enspace{\ttfamily\nolinkurl{isSufficient_duplicateState_iff}} {\tiny\ttfamily Paper4dFrontier/\allowbreak ClosureLaws.lean}
\item \textbf{\nolinkurl{FR59}}\hypertarget{lh:FR59}{}\enspace{\ttfamily\nolinkurl{isRelevant_duplicateState_iff}} {\tiny\ttfamily Paper4dFrontier/\allowbreak ClosureLaws.lean}
\item \textbf{\nolinkurl{FR60}}\hypertarget{lh:FR60}{}\enspace{\ttfamily\nolinkurl{duplicateStateQuotientEquivOriginal_apply_quotientMap}} {\tiny\ttfamily Paper4dFrontier/\allowbreak ClosureLaws.lean}
\item \textbf{\nolinkurl{FR68}}\hypertarget{lh:FR68}{}\enspace{\ttfamily\nolinkurl{decisionEquiv_addDuplicateState_iff}} {\tiny\ttfamily Paper4dFrontier/\allowbreak ClosureLaws.lean}
\item \textbf{\nolinkurl{FR69}}\hypertarget{lh:FR69}{}\enspace{\ttfamily\nolinkurl{isSufficient_addDuplicateState_iff}} {\tiny\ttfamily Paper4dFrontier/\allowbreak ClosureLaws.lean}
\item \textbf{\nolinkurl{FR70}}\hypertarget{lh:FR70}{}\enspace{\ttfamily\nolinkurl{isRelevant_addDuplicateState_iff}} {\tiny\ttfamily Paper4dFrontier/\allowbreak ClosureLaws.lean}
\item \textbf{\nolinkurl{FR71}}\hypertarget{lh:FR71}{}\enspace{\ttfamily\nolinkurl{addDuplicateStateQuotientEquivOriginal_apply_quotientMap}} {\tiny\ttfamily Paper4dFrontier/\allowbreak ClosureLaws.lean}
\item \textbf{\nolinkurl{FR72}}\hypertarget{lh:FR72}{}\enspace{\ttfamily\nolinkurl{some_mem_opt_addDuplicateAction_iff}} {\tiny\ttfamily Paper4dFrontier/\allowbreak ClosureLaws.lean}
\item \textbf{\nolinkurl{FR73}}\hypertarget{lh:FR73}{}\enspace{\ttfamily\nolinkurl{none_mem_opt_addDuplicateAction_iff}} {\tiny\ttfamily Paper4dFrontier/\allowbreak ClosureLaws.lean}
\item \textbf{\nolinkurl{FR83}}\hypertarget{lh:FR83}{}\enspace{\ttfamily\nolinkurl{isSufficient_noiseExtended_iff}} {\tiny\ttfamily Paper4dFrontier/\allowbreak DimensionalNoiseExtension.lean}
\item \textbf{\nolinkurl{FR84}}\hypertarget{lh:FR84}{}\enspace{\ttfamily\nolinkurl{isRelevant_noiseExtended_iff}} {\tiny\ttfamily Paper4dFrontier/\allowbreak DimensionalNoiseExtension.lean}
\item \textbf{\nolinkurl{FR85}}\hypertarget{lh:FR85}{}\enspace{\ttfamily\nolinkurl{lastCoord_irrelevant_noiseExtended}} {\tiny\ttfamily Paper4dFrontier/\allowbreak DimensionalNoiseExtension.lean}
\item \textbf{\nolinkurl{FR118}}\hypertarget{lh:FR118}{}\enspace{\ttfamily\nolinkurl{pairCrossDifference_eq_binaryCrossDifference_of_lt}} {\tiny\ttfamily Paper4dFrontier/\allowbreak BinaryPairwiseDichotomy.lean}
\item \textbf{\nolinkurl{FR119}}\hypertarget{lh:FR119}{}\enspace{\ttfamily\nolinkurl{pairwise_zero_crossDifference_unaryDecomposition}} {\tiny\ttfamily Paper4dFrontier/\allowbreak BinaryPairwiseDichotomy.lean}
\item \textbf{\nolinkurl{FR120}}\hypertarget{lh:FR120}{}\enspace{\ttfamily\nolinkurl{binary_pairwise_symmetry_dichotomy}} {\tiny\ttfamily Paper4dFrontier/\allowbreak BinaryPairwiseDichotomy.lean}
\item \textbf{\nolinkurl{FR121}}\hypertarget{lh:FR121}{}\enspace{\ttfamily\nolinkurl{actionGapCrossDifference_eq_binaryCrossDifference_of_lt}} {\tiny\ttfamily Paper4dFrontier/\allowbreak DecisionRelevantPairwiseDichotomy.lean}
\item \textbf{\nolinkurl{FR122}}\hypertarget{lh:FR122}{}\enspace{\ttfamily\nolinkurl{decisionRelevant_zero_actionGap_implies_unaryReduction}} {\tiny\ttfamily Paper4dFrontier/\allowbreak DecisionRelevantPairwiseDichotomy.lean}
\item \textbf{\nolinkurl{FR123}}\hypertarget{lh:FR123}{}\enspace{\ttfamily\nolinkurl{binary_pairwise_symmetry_decision_relevant_dichotomy}} {\tiny\ttfamily Paper4dFrontier/\allowbreak DecisionRelevantPairwiseDichotomy.lean}
\item \textbf{\nolinkurl{FR124}}\hypertarget{lh:FR124}{}\enspace{\ttfamily\nolinkurl{block6_genuine_interaction_dichotomy_obstruction}} {\tiny\ttfamily Paper4dFrontier/\allowbreak Block6Obstruction.lean}
\item \textbf{\nolinkurl{FR125}}\hypertarget{lh:FR125}{}\enspace{\ttfamily\nolinkurl{decisionRelevantInteractionGraph_addActionOffset}} {\tiny\ttfamily Paper4dFrontier/\allowbreak DecisionRelevantPairwiseDichotomy.lean}
\item \textbf{\nolinkurl{FR126}}\hypertarget{lh:FR126}{}\enspace{\ttfamily\nolinkurl{action_offset_can_force_constant_optimizer}} {\tiny\ttfamily Paper4dFrontier/\allowbreak Block6Obstruction.lean}
\item \textbf{\nolinkurl{FR127}}\hypertarget{lh:FR127}{}\enspace{\ttfamily\nolinkurl{block6_action_offset_obstruction}} {\tiny\ttfamily Paper4dFrontier/\allowbreak Block6Obstruction.lean}
\item \textbf{\nolinkurl{FR128}}\hypertarget{lh:FR128}{}\enspace{\ttfamily\nolinkurl{offsetNormalizedDecisionRelevantInteractionGraph_wellDefined}} {\tiny\ttfamily Paper4dFrontier/\allowbreak DecisionRelevantPairwiseDichotomy.lean}
\item \textbf{\nolinkurl{FR129}}\hypertarget{lh:FR129}{}\enspace{\ttfamily\nolinkurl{block6_offset_normalized_obstruction}} {\tiny\ttfamily Paper4dFrontier/\allowbreak Block6Obstruction.lean}
\item \textbf{\nolinkurl{FR130}}\hypertarget{lh:FR130}{}\enspace{\ttfamily\nolinkurl{neverOptimalGhost_decisionRelevantGraph_eq_top}} {\tiny\ttfamily Paper4dFrontier/\allowbreak Block6Obstruction.lean}
\item \textbf{\nolinkurl{FR131}}\hypertarget{lh:FR131}{}\enspace{\ttfamily\nolinkurl{neverOptimalGhost_supportedGraph_eq_bot}} {\tiny\ttfamily Paper4dFrontier/\allowbreak Block6Obstruction.lean}
\item \textbf{\nolinkurl{FR132}}\hypertarget{lh:FR132}{}\enspace{\ttfamily\nolinkurl{support_filtering_removes_known_block6_obstructions}} {\tiny\ttfamily Paper4dFrontier/\allowbreak Block6Obstruction.lean}
\item \textbf{\nolinkurl{FR133}}\hypertarget{lh:FR133}{}\enspace{\ttfamily\nolinkurl{block6_ghost_action_obstruction}} {\tiny\ttfamily Paper4dFrontier/\allowbreak Block6Obstruction.lean}
\item \textbf{\nolinkurl{FR134}}\hypertarget{lh:FR134}{}\enspace{\ttfamily\nolinkurl{marginMasking_supportedGraph_eq_top}} {\tiny\ttfamily Paper4dFrontier/\allowbreak Block6Obstruction.lean}
\item \textbf{\nolinkurl{FR135}}\hypertarget{lh:FR135}{}\enspace{\ttfamily\nolinkurl{marginMasking_zero_sufficient}} {\tiny\ttfamily Paper4dFrontier/\allowbreak Block6Obstruction.lean}
\item \textbf{\nolinkurl{FR136}}\hypertarget{lh:FR136}{}\enspace{\ttfamily\nolinkurl{block6_optimizer_supported_obstruction}} {\tiny\ttfamily Paper4dFrontier/\allowbreak Block6Obstruction.lean}
\item \textbf{\nolinkurl{FR137}}\hypertarget{lh:FR137}{}\enspace{\ttfamily\nolinkurl{MarginBounded}} {\tiny\ttfamily Paper4dFrontier/\allowbreak DecisionRelevantPairwiseDichotomy.lean}
\item \textbf{\nolinkurl{FR138}}\hypertarget{lh:FR138}{}\enspace{\ttfamily\nolinkurl{dominantPair_marginBounded}} {\tiny\ttfamily Paper4dFrontier/\allowbreak Block6Obstruction.lean}
\item \textbf{\nolinkurl{FR139}}\hypertarget{lh:FR139}{}\enspace{\ttfamily\nolinkurl{dominantPair_supportedGraph_eq_top}} {\tiny\ttfamily Paper4dFrontier/\allowbreak Block6Obstruction.lean}
\item \textbf{\nolinkurl{FR140}}\hypertarget{lh:FR140}{}\enspace{\ttfamily\nolinkurl{block6_margin_bounded_obstruction}} {\tiny\ttfamily Paper4dFrontier/\allowbreak Block6Obstruction.lean}
\item \textbf{\nolinkurl{FR146}}\hypertarget{lh:FR146}{}\enspace{\ttfamily\nolinkurl{AdmissibleNormalizationPredicate}} {\tiny\ttfamily Paper4dFrontier/\allowbreak AdmissibleCharacterization.lean}
\item \textbf{\nolinkurl{FR176}}\hypertarget{lh:FR176}{}\enspace{\ttfamily\nolinkurl{DecisionQuotient.DecisionProblem.quotient_is_coarsest}} {\tiny\ttfamily paper4/\allowbreak DecisionQuotient/\allowbreak Quotient.lean}
\item \textbf{\nolinkurl{FR177}}\hypertarget{lh:FR177}{}\enspace{\ttfamily\nolinkurl{DecisionQuotient.DecisionProblem.quotient_has_unique_factorization}} {\tiny\ttfamily paper4/\allowbreak DecisionQuotient/\allowbreak Quotient.lean}
\item \textbf{\nolinkurl{FR183}}\hypertarget{lh:FR183}{}\enspace{\ttfamily\nolinkurl{DecisionQuotient.Physics.single_future_zero_cost}} {\tiny\ttfamily paper4/\allowbreak DecisionQuotient/\allowbreak Physics/\allowbreak WassersteinIntegrity.lean}
\item \textbf{\nolinkurl{FR184}}\hypertarget{lh:FR184}{}\enspace{\ttfamily\nolinkurl{DecisionQuotient.Physics.transportCost_pos_of_offDiag}} {\tiny\ttfamily paper4/\allowbreak DecisionQuotient/\allowbreak Physics/\allowbreak WassersteinIntegrity.lean}
\item \textbf{\nolinkurl{FR185}}\hypertarget{lh:FR185}{}\enspace{\ttfamily\nolinkurl{no_closureInvariant_predicate_of_orbit_gap}} {\tiny\ttfamily Paper4dFrontier/\allowbreak ObstructionPredicateCandidates.lean}
\item \textbf{\nolinkurl{FR186}}\hypertarget{lh:FR186}{}\enspace{\ttfamily\nolinkurl{no_admissibleNormalizationPredicate_decides_dominantPair}} {\tiny\ttfamily Paper4dFrontier/\allowbreak ObstructionPredicateCandidates.lean}
\item \textbf{\nolinkurl{FR187}}\hypertarget{lh:FR187}{}\enspace{\ttfamily\nolinkurl{no_admissibleNormalizationPredicate_decides_marginBounded}} {\tiny\ttfamily Paper4dFrontier/\allowbreak ObstructionPredicateCandidates.lean}
\item \textbf{\nolinkurl{FR188}}\hypertarget{lh:FR188}{}\enspace{\ttfamily\nolinkurl{no_admissibleNormalizationPredicate_decides_ghostActionTwoPairCrossOne}} {\tiny\ttfamily Paper4dFrontier/\allowbreak ObstructionPredicateCandidates.lean}
\item \textbf{\nolinkurl{FR189}}\hypertarget{lh:FR189}{}\enspace{\ttfamily\nolinkurl{no_admissibleNormalizationPredicate_decides_offsetActionZeroPairCrossOne}} {\tiny\ttfamily Paper4dFrontier/\allowbreak ObstructionPredicateCandidates.lean}
\item \textbf{\nolinkurl{FR190}}\hypertarget{lh:FR190}{}\enspace{\ttfamily\nolinkurl{admissibleCollapseLandscapeInfinity_full_paper}} {\tiny\ttfamily Paper4dFrontier/\allowbreak ObstructionPredicateCandidates.lean}
\item \textbf{\nolinkurl{FR191}}\hypertarget{lh:FR191}{}\enspace{\ttfamily\nolinkurl{ClosureSoundPackage}} {\tiny\ttfamily Paper4dFrontier/\allowbreak ObstructionPredicateCandidates.lean}
\item \textbf{\nolinkurl{FR192}}\hypertarget{lh:FR192}{}\enspace{\ttfamily\nolinkurl{ReasonableGuardrailPackage}} {\tiny\ttfamily Paper4dFrontier/\allowbreak ObstructionPredicateCandidates.lean}
\item \textbf{\nolinkurl{FR193}}\hypertarget{lh:FR193}{}\enspace{\ttfamily\nolinkurl{no_closureSoundPackage_predicate_decides_each_obstruction_family}} {\tiny\ttfamily Paper4dFrontier/\allowbreak ObstructionPredicateCandidates.lean}
\item \textbf{\nolinkurl{FR197}}\hypertarget{lh:FR197}{}\enspace{\ttfamily\nolinkurl{classifier_agrees_on_closureEquivalent_of_correctOnDomain}} {\tiny\ttfamily Paper4dFrontier/\allowbreak AdmissibleCharacterization.lean}
\item \textbf{\nolinkurl{FR198}}\hypertarget{lh:FR198}{}\enspace{\ttfamily\nolinkurl{no_correctOnDomain_classifier_of_orbit_gap}} {\tiny\ttfamily Paper4dFrontier/\allowbreak AdmissibleCharacterization.lean}
\item \textbf{\nolinkurl{FR199}}\hypertarget{lh:FR199}{}\enspace{\ttfamily\nolinkurl{correct_classifier_inherits_closureLawInvariant}} {\tiny\ttfamily Paper4dFrontier/\allowbreak AdmissibleCharacterization.lean}
\item \textbf{\nolinkurl{FR200}}\hypertarget{lh:FR200}{}\enspace{\ttfamily\nolinkurl{admissibleNormalizationPredicate_has_explicit_inhabitants}} {\tiny\ttfamily Paper4dFrontier/\allowbreak AdmissibleCharacterization.lean}
\item \textbf{\nolinkurl{FR201}}\hypertarget{lh:FR201}{}\enspace{\ttfamily\nolinkurl{closureLawInvariant_of_iff_of_closureEquivalent}} {\tiny\ttfamily Paper4dFrontier/\allowbreak AdmissibleCharacterization.lean}
\item \textbf{\nolinkurl{FR202}}\hypertarget{lh:FR202}{}\enspace{\ttfamily\nolinkurl{exists_orbit_gap_of_not_closureLawInvariant}} {\tiny\ttfamily Paper4dFrontier/\allowbreak AdmissibleCharacterization.lean}
\item \textbf{\nolinkurl{FR203}}\hypertarget{lh:FR203}{}\enspace{\ttfamily\nolinkurl{closureLawInvariant_iff_no_orbit_gap}} {\tiny\ttfamily Paper4dFrontier/\allowbreak AdmissibleCharacterization.lean}
\item \textbf{\nolinkurl{FR204}}\hypertarget{lh:FR204}{}\enspace{\ttfamily\nolinkurl{no_exact_closureLawInvariant_classifier_iff_exists_orbit_gap}} {\tiny\ttfamily Paper4dFrontier/\allowbreak AdmissibleCharacterization.lean}
\item \textbf{\nolinkurl{FR206}}\hypertarget{lh:FR206}{}\enspace{\ttfamily\nolinkurl{actionCount_profileCompressedInstance}} {\tiny\ttfamily Paper4dFrontier/\allowbreak DistinctActionProfiles.lean}
\item \textbf{\nolinkurl{FR207}}\hypertarget{lh:FR207}{}\enspace{\ttfamily\nolinkurl{decisionEquiv_profileCompressedInstance_iff}} {\tiny\ttfamily Paper4dFrontier/\allowbreak DistinctActionProfiles.lean}
\item \textbf{\nolinkurl{FR208}}\hypertarget{lh:FR208}{}\enspace{\ttfamily\nolinkurl{profileCompressedInstance_preserves_exactCertification}} {\tiny\ttfamily Paper4dFrontier/\allowbreak DistinctActionProfiles.lean}
\item \textbf{\nolinkurl{FR209}}\hypertarget{lh:FR209}{}\enspace{\ttfamily\nolinkurl{profileCompressedInstance_bounded_actions}} {\tiny\ttfamily Paper4dFrontier/\allowbreak DistinctActionProfiles.lean}
\item \textbf{\nolinkurl{FR210}}\hypertarget{lh:FR210}{}\enspace{\ttfamily\nolinkurl{OrbitGapOn}} {\tiny\ttfamily Paper4dFrontier/\allowbreak AdmissibleCharacterization.lean}
\item \textbf{\nolinkurl{FR211}}\hypertarget{lh:FR211}{}\enspace{\ttfamily\nolinkurl{no_orbitGapOn_of_exact_classifiable_by_closureLawInvariant_onDomain}} {\tiny\ttfamily Paper4dFrontier/\allowbreak AdmissibleCharacterization.lean}
\item \textbf{\nolinkurl{FR212}}\hypertarget{lh:FR212}{}\enspace{\ttfamily\nolinkurl{exact_classifiable_by_closureLawInvariant_onDomain_iff_no_orbitGapOn}} {\tiny\ttfamily Paper4dFrontier/\allowbreak AdmissibleCharacterization.lean}
\item \textbf{\nolinkurl{FR213}}\hypertarget{lh:FR213}{}\enspace{\ttfamily\nolinkurl{no_exact_closureLawInvariant_classifier_onDomain_iff_orbitGapOn}} {\tiny\ttfamily Paper4dFrontier/\allowbreak AdmissibleCharacterization.lean}
\item \textbf{\nolinkurl{FR217}}\hypertarget{lh:FR217}{}\enspace{\ttfamily\nolinkurl{payloadSufficient_iff_realizingProblem_isSufficient}} {\tiny\ttfamily Paper4dFrontier/\allowbreak Realizability.lean}
\item \textbf{\nolinkurl{FR218}}\hypertarget{lh:FR218}{}\enspace{\ttfamily\nolinkurl{payloadRelevant_iff_realizingProblem_isRelevant}} {\tiny\ttfamily Paper4dFrontier/\allowbreak Realizability.lean}
\item \textbf{\nolinkurl{FR219}}\hypertarget{lh:FR219}{}\enspace{\ttfamily\nolinkurl{payloadIrrelevant_iff_realizingProblem_isIrrelevant}} {\tiny\ttfamily Paper4dFrontier/\allowbreak Realizability.lean}
\item \textbf{\nolinkurl{FR220}}\hypertarget{lh:FR220}{}\enspace{\ttfamily\nolinkurl{payloadMinimalSufficient_iff_realizingProblem_isMinimalSufficient}} {\tiny\ttfamily Paper4dFrontier/\allowbreak Realizability.lean}
\item \textbf{\nolinkurl{FR221}}\hypertarget{lh:FR221}{}\enspace{\ttfamily\nolinkurl{payloadSufficientSets_eq_realizingProblem_sufficientSets}} {\tiny\ttfamily Paper4dFrontier/\allowbreak Realizability.lean}
\item \textbf{\nolinkurl{FR222}}\hypertarget{lh:FR222}{}\enspace{\ttfamily\nolinkurl{payloadRelevantSet_eq_realizingProblem_relevantSet}} {\tiny\ttfamily Paper4dFrontier/\allowbreak Realizability.lean}
\item \textbf{\nolinkurl{FR223}}\hypertarget{lh:FR223}{}\enspace{\ttfamily\nolinkurl{feasiblePayloadSufficient_iff_lawDecisionProblem_isSufficient}} {\tiny\ttfamily Paper4dFrontier/\allowbreak Realizability.lean}
\item \textbf{\nolinkurl{FR224}}\hypertarget{lh:FR224}{}\enspace{\ttfamily\nolinkurl{feasiblePayloadRelevant_iff_lawDecisionProblem_isRelevant}} {\tiny\ttfamily Paper4dFrontier/\allowbreak Realizability.lean}
\item \textbf{\nolinkurl{FR225}}\hypertarget{lh:FR225}{}\enspace{\ttfamily\nolinkurl{feasiblePayloadIrrelevant_iff_lawDecisionProblem_isIrrelevant}} {\tiny\ttfamily Paper4dFrontier/\allowbreak Realizability.lean}
\item \textbf{\nolinkurl{FR226}}\hypertarget{lh:FR226}{}\enspace{\ttfamily\nolinkurl{setValuedPayloadSufficient_iff_totalizedLawDecisionProblem_isSufficient}} {\tiny\ttfamily Paper4dFrontier/\allowbreak Realizability.lean}
\item \textbf{\nolinkurl{FR227}}\hypertarget{lh:FR227}{}\enspace{\ttfamily\nolinkurl{setValuedPayloadRelevant_iff_totalizedLawDecisionProblem_isRelevant}} {\tiny\ttfamily Paper4dFrontier/\allowbreak Realizability.lean}
\item \textbf{\nolinkurl{FR228}}\hypertarget{lh:FR228}{}\enspace{\ttfamily\nolinkurl{setValuedPayloadIrrelevant_iff_totalizedLawDecisionProblem_isIrrelevant}} {\tiny\ttfamily Paper4dFrontier/\allowbreak Realizability.lean}
\item \textbf{\nolinkurl{FR229}}\hypertarget{lh:FR229}{}\enspace{\ttfamily\nolinkurl{outputSemanticsSufficient_iff_totalizedLawDecisionProblem_isSufficient}} {\tiny\ttfamily Paper4dFrontier/\allowbreak Realizability.lean}
\item \textbf{\nolinkurl{FR230}}\hypertarget{lh:FR230}{}\enspace{\ttfamily\nolinkurl{outputSemanticsRelevant_iff_totalizedLawDecisionProblem_isRelevant}} {\tiny\ttfamily Paper4dFrontier/\allowbreak Realizability.lean}
\item \textbf{\nolinkurl{FR231}}\hypertarget{lh:FR231}{}\enspace{\ttfamily\nolinkurl{outputSemanticsIrrelevant_iff_totalizedLawDecisionProblem_isIrrelevant}} {\tiny\ttfamily Paper4dFrontier/\allowbreak Realizability.lean}
\item \textbf{\nolinkurl{FR232}}\hypertarget{lh:FR232}{}\enspace{\ttfamily\nolinkurl{approximationSemanticsSufficient_iff_totalizedLawDecisionProblem_isSufficient}} {\tiny\ttfamily Paper4dFrontier/\allowbreak Realizability.lean}
\item \textbf{\nolinkurl{FR233}}\hypertarget{lh:FR233}{}\enspace{\ttfamily\nolinkurl{approximationSemanticsRelevant_iff_totalizedLawDecisionProblem_isRelevant}} {\tiny\ttfamily Paper4dFrontier/\allowbreak Realizability.lean}
\item \textbf{\nolinkurl{FR234}}\hypertarget{lh:FR234}{}\enspace{\ttfamily\nolinkurl{approximationSemanticsIrrelevant_iff_totalizedLawDecisionProblem_isIrrelevant}} {\tiny\ttfamily Paper4dFrontier/\allowbreak Realizability.lean}
\item \textbf{\nolinkurl{FR235}}\hypertarget{lh:FR235}{}\enspace{\ttfamily\nolinkurl{outputSemanticsSufficientSets_eq_of_outputSemanticsEquivalent}} {\tiny\ttfamily Paper4dFrontier/\allowbreak Realizability.lean}
\item \textbf{\nolinkurl{FR236}}\hypertarget{lh:FR236}{}\enspace{\ttfamily\nolinkurl{outputSemanticsRelevantSet_eq_of_outputSemanticsEquivalent}} {\tiny\ttfamily Paper4dFrontier/\allowbreak Realizability.lean}
\item \textbf{\nolinkurl{FR237}}\hypertarget{lh:FR237}{}\enspace{\ttfamily\nolinkurl{quotientMap_realizes_setoid_asOutputSemantics}} {\tiny\ttfamily Paper4dFrontier/\allowbreak Realizability.lean}
\item \textbf{\nolinkurl{FR238}}\hypertarget{lh:FR238}{}\enspace{\ttfamily\nolinkurl{quotientOutputSemanticsSufficient_iff_relationSufficient}} {\tiny\ttfamily Paper4dFrontier/\allowbreak Realizability.lean}
\item \textbf{\nolinkurl{FR239}}\hypertarget{lh:FR239}{}\enspace{\ttfamily\nolinkurl{quotientOutputSemanticsRelevant_iff_relationRelevant}} {\tiny\ttfamily Paper4dFrontier/\allowbreak Realizability.lean}
\item \textbf{\nolinkurl{FR240}}\hypertarget{lh:FR240}{}\enspace{\ttfamily\nolinkurl{exists_outputSemantics_realizing_setoid}} {\tiny\ttfamily Paper4dFrontier/\allowbreak Realizability.lean}
\item \textbf{\nolinkurl{FR241}}\hypertarget{lh:FR241}{}\enspace{\ttfamily\nolinkurl{outputSemanticsSetoid}} {\tiny\ttfamily Paper4dFrontier/\allowbreak Realizability.lean}
\item \textbf{\nolinkurl{FR242}}\hypertarget{lh:FR242}{}\enspace{\ttfamily\nolinkurl{SemanticallyExtensionalMap}} {\tiny\ttfamily Paper4dFrontier/\allowbreak Realizability.lean}
\item \textbf{\nolinkurl{FR243}}\hypertarget{lh:FR243}{}\enspace{\ttfamily\nolinkurl{semanticallyExtensionalMap_factors_through_outputSemanticsQuotient}} {\tiny\ttfamily Paper4dFrontier/\allowbreak Realizability.lean}
\item \textbf{\nolinkurl{FR244}}\hypertarget{lh:FR244}{}\enspace{\ttfamily\nolinkurl{semanticallyExtensionalClaim_factors_through_outputSemanticsQuotient}} {\tiny\ttfamily Paper4dFrontier/\allowbreak Realizability.lean}
\item \textbf{\nolinkurl{FR248}}\hypertarget{lh:FR248}{}\enspace{\ttfamily\nolinkurl{optimizerComputation_polytime_classifier_agrees_on_closureEquivalent_of_correctOnDomain}} {\tiny\ttfamily Paper4dFrontier/\allowbreak ComputeCostApplications.lean}
\item \textbf{\nolinkurl{FR250}}\hypertarget{lh:FR250}{}\enspace{\ttfamily\nolinkurl{no_correct_optimizerComputation_polytime_classifier_decides_dominantPair}} {\tiny\ttfamily Paper4dFrontier/\allowbreak ComputeCostApplications.lean}
\item \textbf{\nolinkurl{FR251}}\hypertarget{lh:FR251}{}\enspace{\ttfamily\nolinkurl{no_admissibleNormalizationPredicate_optimizerComputation_polytime_and_dominantPair}} {\tiny\ttfamily Paper4dFrontier/\allowbreak ComputeCostApplications.lean}
\item \textbf{\nolinkurl{FR252}}\hypertarget{lh:FR252}{}\enspace{\ttfamily\nolinkurl{CorrectnessForcesOrbitAgreementOnDomain}} {\tiny\ttfamily Paper4dFrontier/\allowbreak AdmissibleCharacterization.lean}
\item \textbf{\nolinkurl{FR253}}\hypertarget{lh:FR253}{}\enspace{\ttfamily\nolinkurl{no_orbitGapOn_of_correct_classifier_onDomain_of_forcedOrbitAgreement}} {\tiny\ttfamily Paper4dFrontier/\allowbreak AdmissibleCharacterization.lean}
\item \textbf{\nolinkurl{FR254}}\hypertarget{lh:FR254}{}\enspace{\ttfamily\nolinkurl{correct_classifier_onDomain_iff_no_orbitGapOn_of_forcedOrbitAgreement}} {\tiny\ttfamily Paper4dFrontier/\allowbreak AdmissibleCharacterization.lean}
\item \textbf{\nolinkurl{FR255}}\hypertarget{lh:FR255}{}\enspace{\ttfamily\nolinkurl{no_correct_classifier_onDomain_iff_orbitGapOn_of_forcedOrbitAgreement}} {\tiny\ttfamily Paper4dFrontier/\allowbreak AdmissibleCharacterization.lean}
\item \textbf{\nolinkurl{FR256}}\hypertarget{lh:FR256}{}\enspace{\ttfamily\nolinkurl{optimizerComputation_polytime_correct_classifier_onDomain_iff_no_orbitGapOn}} {\tiny\ttfamily Paper4dFrontier/\allowbreak ComputeCostApplications.lean}
\item \textbf{\nolinkurl{FR259}}\hypertarget{lh:FR259}{}\enspace{\ttfamily\nolinkurl{optimizerSetPayload_polytime_classifier_agrees_on_closureEquivalent_of_correctOnDomain}} {\tiny\ttfamily Paper4dFrontier/\allowbreak ComputeCostApplications.lean}
\item \textbf{\nolinkurl{FR260}}\hypertarget{lh:FR260}{}\enspace{\ttfamily\nolinkurl{optimizerSetPayload_polytime_classifier_agrees_on_closureEquivalent}} {\tiny\ttfamily Paper4dFrontier/\allowbreak ComputeCostApplications.lean}
\item \textbf{\nolinkurl{FR261}}\hypertarget{lh:FR261}{}\enspace{\ttfamily\nolinkurl{optimizerSetSearch_polytime_classifier_agrees_on_closureEquivalent_of_correctOnDomain}} {\tiny\ttfamily Paper4dFrontier/\allowbreak ComputeCostApplications.lean}
\item \textbf{\nolinkurl{FR262}}\hypertarget{lh:FR262}{}\enspace{\ttfamily\nolinkurl{optimizerSetSearch_polytime_classifier_agrees_on_closureEquivalent}} {\tiny\ttfamily Paper4dFrontier/\allowbreak ComputeCostApplications.lean}
\item \textbf{\nolinkurl{FR263}}\hypertarget{lh:FR263}{}\enspace{\ttfamily\nolinkurl{relevant_of_uniformApprox_of_strict_gap_witness}} {\tiny\ttfamily Paper4dFrontier/\allowbreak ApproximateAdmissibility.lean}
\item \textbf{\nolinkurl{FR264}}\hypertarget{lh:FR264}{}\enspace{\ttfamily\nolinkurl{relevance_can_flip_under_arbitrarily_small_uniform_perturbation}} {\tiny\ttfamily Paper4dFrontier/\allowbreak ApproximateAdmissibility.lean}
\item \textbf{\nolinkurl{FR265}}\hypertarget{lh:FR265}{}\enspace{\ttfamily\nolinkurl{not_decisionEquiv_of_uniformApprox_of_strict_gap_witness}} {\tiny\ttfamily Paper4dFrontier/\allowbreak ApproximateAdmissibility.lean}
\item \textbf{\nolinkurl{FR266}}\hypertarget{lh:FR266}{}\enspace{\ttfamily\nolinkurl{not_sufficient_of_uniformApprox_of_strict_gap_witness}} {\tiny\ttfamily Paper4dFrontier/\allowbreak ApproximateAdmissibility.lean}
\item \textbf{\nolinkurl{FR267}}\hypertarget{lh:FR267}{}\enspace{\ttfamily\nolinkurl{sufficiency_can_flip_under_arbitrarily_small_uniform_perturbation}} {\tiny\ttfamily Paper4dFrontier/\allowbreak ApproximateAdmissibility.lean}
\item \textbf{\nolinkurl{FR268}}\hypertarget{lh:FR268}{}\enspace{\ttfamily\nolinkurl{pacGuaranteeSemanticsSufficient_iff_totalizedLawDecisionProblem_isSufficient}} {\tiny\ttfamily Paper4dFrontier/\allowbreak Realizability.lean}
\item \textbf{\nolinkurl{FR269}}\hypertarget{lh:FR269}{}\enspace{\ttfamily\nolinkurl{pacGuaranteeSemanticsRelevant_iff_totalizedLawDecisionProblem_isRelevant}} {\tiny\ttfamily Paper4dFrontier/\allowbreak Realizability.lean}
\item \textbf{\nolinkurl{FR270}}\hypertarget{lh:FR270}{}\enspace{\ttfamily\nolinkurl{pacGuaranteeSemanticsIrrelevant_iff_totalizedLawDecisionProblem_isIrrelevant}} {\tiny\ttfamily Paper4dFrontier/\allowbreak Realizability.lean}
\item \textbf{\nolinkurl{FR271}}\hypertarget{lh:FR271}{}\enspace{\ttfamily\nolinkurl{regretGuaranteeSemanticsSufficient_iff_totalizedLawDecisionProblem_isSufficient}} {\tiny\ttfamily Paper4dFrontier/\allowbreak Realizability.lean}
\item \textbf{\nolinkurl{FR272}}\hypertarget{lh:FR272}{}\enspace{\ttfamily\nolinkurl{regretGuaranteeSemanticsRelevant_iff_totalizedLawDecisionProblem_isRelevant}} {\tiny\ttfamily Paper4dFrontier/\allowbreak Realizability.lean}
\item \textbf{\nolinkurl{FR273}}\hypertarget{lh:FR273}{}\enspace{\ttfamily\nolinkurl{regretGuaranteeSemanticsIrrelevant_iff_totalizedLawDecisionProblem_isIrrelevant}} {\tiny\ttfamily Paper4dFrontier/\allowbreak Realizability.lean}
\item \textbf{\nolinkurl{FR274}}\hypertarget{lh:FR274}{}\enspace{\ttfamily\nolinkurl{statisticalRiskSemanticsSufficient_iff_totalizedLawDecisionProblem_isSufficient}} {\tiny\ttfamily Paper4dFrontier/\allowbreak Realizability.lean}
\item \textbf{\nolinkurl{FR275}}\hypertarget{lh:FR275}{}\enspace{\ttfamily\nolinkurl{statisticalRiskSemanticsRelevant_iff_totalizedLawDecisionProblem_isRelevant}} {\tiny\ttfamily Paper4dFrontier/\allowbreak Realizability.lean}
\item \textbf{\nolinkurl{FR276}}\hypertarget{lh:FR276}{}\enspace{\ttfamily\nolinkurl{statisticalRiskSemanticsIrrelevant_iff_totalizedLawDecisionProblem_isIrrelevant}} {\tiny\ttfamily Paper4dFrontier/\allowbreak Realizability.lean}
\item \textbf{\nolinkurl{FR277}}\hypertarget{lh:FR277}{}\enspace{\ttfamily\nolinkurl{anytimeGuaranteeSemanticsSufficient_iff_totalizedLawDecisionProblem_isSufficient}} {\tiny\ttfamily Paper4dFrontier/\allowbreak Realizability.lean}
\item \textbf{\nolinkurl{FR278}}\hypertarget{lh:FR278}{}\enspace{\ttfamily\nolinkurl{anytimeGuaranteeSemanticsRelevant_iff_totalizedLawDecisionProblem_isRelevant}} {\tiny\ttfamily Paper4dFrontier/\allowbreak Realizability.lean}
\item \textbf{\nolinkurl{FR279}}\hypertarget{lh:FR279}{}\enspace{\ttfamily\nolinkurl{anytimeGuaranteeSemanticsIrrelevant_iff_totalizedLawDecisionProblem_isIrrelevant}} {\tiny\ttfamily Paper4dFrontier/\allowbreak Realizability.lean}
\item \textbf{\nolinkurl{FR280}}\hypertarget{lh:FR280}{}\enspace{\ttfamily\nolinkurl{finiteHorizonGuaranteeSemanticsSufficient_iff_totalizedLawDecisionProblem_isSufficient}} {\tiny\ttfamily Paper4dFrontier/\allowbreak Realizability.lean}
\item \textbf{\nolinkurl{FR281}}\hypertarget{lh:FR281}{}\enspace{\ttfamily\nolinkurl{finiteHorizonGuaranteeSemanticsRelevant_iff_totalizedLawDecisionProblem_isRelevant}} {\tiny\ttfamily Paper4dFrontier/\allowbreak Realizability.lean}
\item \textbf{\nolinkurl{FR282}}\hypertarget{lh:FR282}{}\enspace{\ttfamily\nolinkurl{finiteHorizonGuaranteeSemanticsIrrelevant_iff_totalizedLawDecisionProblem_isIrrelevant}} {\tiny\ttfamily Paper4dFrontier/\allowbreak Realizability.lean}
\item \textbf{\nolinkurl{FR283}}\hypertarget{lh:FR283}{}\enspace{\ttfamily\nolinkurl{statisticalGuaranteeSemanticsSufficientSets_eq_of_outputSemanticsEquivalent}} {\tiny\ttfamily Paper4dFrontier/\allowbreak Realizability.lean}
\item \textbf{\nolinkurl{FR284}}\hypertarget{lh:FR284}{}\enspace{\ttfamily\nolinkurl{statisticalGuaranteeSemanticsRelevantSet_eq_of_outputSemanticsEquivalent}} {\tiny\ttfamily Paper4dFrontier/\allowbreak Realizability.lean}
\item \textbf{\nolinkurl{FR285}}\hypertarget{lh:FR285}{}\enspace{\ttfamily\nolinkurl{statisticalGuarantee_classifier_agrees_on_closureEquivalent_of_correctOnDomain_of_transfer}} {\tiny\ttfamily Paper4dFrontier/\allowbreak StatisticalSemanticsApplications.lean}
\item \textbf{\nolinkurl{FR286}}\hypertarget{lh:FR286}{}\enspace{\ttfamily\nolinkurl{no_correctOnDomain_statisticalGuarantee_classifier_of_orbit_gap_of_transfer}} {\tiny\ttfamily Paper4dFrontier/\allowbreak StatisticalSemanticsApplications.lean}
\item \textbf{\nolinkurl{FR287}}\hypertarget{lh:FR287}{}\enspace{\ttfamily\nolinkurl{statisticalGuarantee_correct_classifier_onDomain_iff_no_orbitGapOn_of_transfer}} {\tiny\ttfamily Paper4dFrontier/\allowbreak StatisticalSemanticsApplications.lean}
\item \textbf{\nolinkurl{FR288}}\hypertarget{lh:FR288}{}\enspace{\ttfamily\nolinkurl{opt_eq_of_uniformApprox_of_uniformStrictGapCover}} {\tiny\ttfamily Paper4dFrontier/\allowbreak ApproximateAdmissibility.lean}
\item \textbf{\nolinkurl{FR289}}\hypertarget{lh:FR289}{}\enspace{\ttfamily\nolinkurl{sufficientSets_eq_of_uniformApprox_of_uniformStrictGapCover}} {\tiny\ttfamily Paper4dFrontier/\allowbreak ApproximateAdmissibility.lean}
\item \textbf{\nolinkurl{FR290}}\hypertarget{lh:FR290}{}\enspace{\ttfamily\nolinkurl{relevantSet_eq_of_uniformApprox_of_uniformStrictGapCover}} {\tiny\ttfamily Paper4dFrontier/\allowbreak ApproximateAdmissibility.lean}
\item \textbf{\nolinkurl{FR291}}\hypertarget{lh:FR291}{}\enspace{\ttfamily\nolinkurl{randomizedGuaranteeSemanticsSufficient_iff_totalizedLawDecisionProblem_isSufficient}} {\tiny\ttfamily Paper4dFrontier/\allowbreak Realizability.lean}
\item \textbf{\nolinkurl{FR292}}\hypertarget{lh:FR292}{}\enspace{\ttfamily\nolinkurl{randomizedGuaranteeSemanticsRelevant_iff_totalizedLawDecisionProblem_isRelevant}} {\tiny\ttfamily Paper4dFrontier/\allowbreak Realizability.lean}
\item \textbf{\nolinkurl{FR293}}\hypertarget{lh:FR293}{}\enspace{\ttfamily\nolinkurl{randomizedGuaranteeSemanticsIrrelevant_iff_totalizedLawDecisionProblem_isIrrelevant}} {\tiny\ttfamily Paper4dFrontier/\allowbreak Realizability.lean}
\item \textbf{\nolinkurl{FR294}}\hypertarget{lh:FR294}{}\enspace{\ttfamily\nolinkurl{randomizedGuarantee_classifier_agrees_on_closureEquivalent_of_correctOnDomain_of_transfer}} {\tiny\ttfamily Paper4dFrontier/\allowbreak StatisticalSemanticsApplications.lean}
\item \textbf{\nolinkurl{FR295}}\hypertarget{lh:FR295}{}\enspace{\ttfamily\nolinkurl{no_correctOnDomain_randomizedGuarantee_classifier_of_orbit_gap_of_transfer}} {\tiny\ttfamily Paper4dFrontier/\allowbreak StatisticalSemanticsApplications.lean}
\item \textbf{\nolinkurl{FR296}}\hypertarget{lh:FR296}{}\enspace{\ttfamily\nolinkurl{randomizedGuarantee_correct_classifier_onDomain_iff_no_orbitGapOn_of_transfer}} {\tiny\ttfamily Paper4dFrontier/\allowbreak StatisticalSemanticsApplications.lean}
\item \textbf{\nolinkurl{FR297}}\hypertarget{lh:FR297}{}\enspace{\ttfamily\nolinkurl{transferredSemantics_classifier_agrees_on_closureEquivalent_of_correctOnDomain}} {\tiny\ttfamily Paper4dFrontier/\allowbreak StatisticalSemanticsApplications.lean}
\item \textbf{\nolinkurl{FR298}}\hypertarget{lh:FR298}{}\enspace{\ttfamily\nolinkurl{no_correctOnDomain_transferredSemantics_classifier_of_orbit_gap}} {\tiny\ttfamily Paper4dFrontier/\allowbreak StatisticalSemanticsApplications.lean}
\item \textbf{\nolinkurl{FR299}}\hypertarget{lh:FR299}{}\enspace{\ttfamily\nolinkurl{transferredSemantics_correct_classifier_onDomain_iff_no_orbitGapOn}} {\tiny\ttfamily Paper4dFrontier/\allowbreak StatisticalSemanticsApplications.lean}
\item \textbf{\nolinkurl{FR300}}\hypertarget{lh:FR300}{}\enspace{\ttfamily\nolinkurl{decisionEquiv_iff_of_uniformApprox_of_uniformStrictGapCover}} {\tiny\ttfamily Paper4dFrontier/\allowbreak ApproximateAdmissibility.lean}
\item \textbf{\nolinkurl{FR301}}\hypertarget{lh:FR301}{}\enspace{\ttfamily\nolinkurl{isMinimalSufficient_iff_of_uniformApprox_of_uniformStrictGapCover}} {\tiny\ttfamily Paper4dFrontier/\allowbreak ApproximateAdmissibility.lean}
\item \textbf{\nolinkurl{FR302}}\hypertarget{lh:FR302}{}\enspace{\ttfamily\nolinkurl{outputSemanticsExactRelevanceProfile_eq_of_outputSemanticsEquivalent}} {\tiny\ttfamily Paper4dFrontier/\allowbreak Realizability.lean}
\item \textbf{\nolinkurl{FR303}}\hypertarget{lh:FR303}{}\enspace{\ttfamily\nolinkurl{outputSemanticsExactRelevanceProfile_eq_totalizedLawDecisionProblem}} {\tiny\ttfamily Paper4dFrontier/\allowbreak Realizability.lean}
\item \textbf{\nolinkurl{FR304}}\hypertarget{lh:FR304}{}\enspace{\ttfamily\nolinkurl{agreeOn_univ_iff_eq_of_finiteIndicatorCoordinateSpace}} {\tiny\ttfamily Paper4dFrontier/\allowbreak Realizability.lean}
\item \textbf{\nolinkurl{FR305}}\hypertarget{lh:FR305}{}\enspace{\ttfamily\nolinkurl{finite_outputSemantics_allCoordinatesSufficient}} {\tiny\ttfamily Paper4dFrontier/\allowbreak Realizability.lean}
\item \textbf{\nolinkurl{FR306}}\hypertarget{lh:FR306}{}\enspace{\ttfamily\nolinkurl{finite_outputSemantics_realized_by_exactCertification}} {\tiny\ttfamily Paper4dFrontier/\allowbreak Realizability.lean}
\item \textbf{\nolinkurl{FR307}}\hypertarget{lh:FR307}{}\enspace{\ttfamily\nolinkurl{agreeOn_univ_iff_eq_of_singletonIdentityCoordinateSpace}} {\tiny\ttfamily Paper4dFrontier/\allowbreak Realizability.lean}
\item \textbf{\nolinkurl{FR308}}\hypertarget{lh:FR308}{}\enspace{\ttfamily\nolinkurl{outputSemantics_admits_coordinatePresentation}} {\tiny\ttfamily Paper4dFrontier/\allowbreak Realizability.lean}
\item \textbf{\nolinkurl{FR309}}\hypertarget{lh:FR309}{}\enspace{\ttfamily\nolinkurl{outputSemantics_realized_by_singletonMeasurableCoordinatePresentation}} {\tiny\ttfamily Paper4dFrontier/\allowbreak Realizability.lean}
\item \textbf{\nolinkurl{FR310}}\hypertarget{lh:FR310}{}\enspace{\ttfamily\nolinkurl{outputSemantics_realized_by_singletonContinuousCoordinatePresentation}} {\tiny\ttfamily Paper4dFrontier/\allowbreak Realizability.lean}
\item \textbf{\nolinkurl{FR311}}\hypertarget{lh:FR311}{}\enspace{\ttfamily\nolinkurl{encodable_stateSpace_admits_countableBooleanPresentation}} {\tiny\ttfamily Paper4dFrontier/\allowbreak Realizability.lean}
\item \textbf{\nolinkurl{FR312}}\hypertarget{lh:FR312}{}\enspace{\ttfamily\nolinkurl{encodable_discreteMeasurable_stateSpace_admits_measurable_countableBooleanPresentation}} {\tiny\ttfamily Paper4dFrontier/\allowbreak Realizability.lean}
\item \textbf{\nolinkurl{FR313}}\hypertarget{lh:FR313}{}\enspace{\ttfamily\nolinkurl{encodable_discreteTopological_stateSpace_admits_continuous_countableBooleanPresentation}} {\tiny\ttfamily Paper4dFrontier/\allowbreak Realizability.lean}
\item \textbf{\nolinkurl{FR314}}\hypertarget{lh:FR314}{}\enspace{\ttfamily\nolinkurl{agreeOn_univ_iff_eq_of_finiteBinaryCoordinateSpace}} {\tiny\ttfamily Paper4dFrontier/\allowbreak Realizability.lean}
\item \textbf{\nolinkurl{FR315}}\hypertarget{lh:FR315}{}\enspace{\ttfamily\nolinkurl{finite_exactSpecification_admits_lowDimBooleanPresentation}} {\tiny\ttfamily Paper4dFrontier/\allowbreak Realizability.lean}
\item \textbf{\nolinkurl{FR316}}\hypertarget{lh:FR316}{}\enspace{\ttfamily\nolinkurl{IdentityOutputClosureSpec.classifier_agrees_on_closureEquivalent_of_correctOnDomain}} {\tiny\ttfamily Paper4dFrontier/\allowbreak ComputeCostExternalOutputs.lean}
\item \textbf{\nolinkurl{FR317}}\hypertarget{lh:FR317}{}\enspace{\ttfamily\nolinkurl{IdentityOutputClosureSpec.no_correctOnDomain_classifier_of_orbit_gap}} {\tiny\ttfamily Paper4dFrontier/\allowbreak ComputeCostExternalOutputs.lean}
\item \textbf{\nolinkurl{FR318}}\hypertarget{lh:FR318}{}\enspace{\ttfamily\nolinkurl{hypothesisOutput_classifier_agrees_on_closureEquivalent_of_correctOnDomain}} {\tiny\ttfamily Paper4dFrontier/\allowbreak ComputeCostExternalOutputs.lean}
\item \textbf{\nolinkurl{FR319}}\hypertarget{lh:FR319}{}\enspace{\ttfamily\nolinkurl{estimatorOutput_classifier_agrees_on_closureEquivalent_of_correctOnDomain}} {\tiny\ttfamily Paper4dFrontier/\allowbreak ComputeCostExternalOutputs.lean}
\item \textbf{\nolinkurl{FR320}}\hypertarget{lh:FR320}{}\enspace{\ttfamily\nolinkurl{policyOutput_classifier_agrees_on_closureEquivalent_of_correctOnDomain}} {\tiny\ttfamily Paper4dFrontier/\allowbreak ComputeCostExternalOutputs.lean}
\item \textbf{\nolinkurl{FR321}}\hypertarget{lh:FR321}{}\enspace{\ttfamily\nolinkurl{randomizedProcedure_classifier_agrees_on_closureEquivalent_of_correctOnDomain}} {\tiny\ttfamily Paper4dFrontier/\allowbreak ComputeCostExternalOutputs.lean}
\item \textbf{\nolinkurl{FR322}}\hypertarget{lh:FR322}{}\enspace{\ttfamily\nolinkurl{no_correctOnDomain_hypothesisOutput_classifier_of_orbit_gap}} {\tiny\ttfamily Paper4dFrontier/\allowbreak ComputeCostExternalOutputs.lean}
\item \textbf{\nolinkurl{FR323}}\hypertarget{lh:FR323}{}\enspace{\ttfamily\nolinkurl{no_correctOnDomain_estimatorOutput_classifier_of_orbit_gap}} {\tiny\ttfamily Paper4dFrontier/\allowbreak ComputeCostExternalOutputs.lean}
\item \textbf{\nolinkurl{FR324}}\hypertarget{lh:FR324}{}\enspace{\ttfamily\nolinkurl{no_correctOnDomain_policyOutput_classifier_of_orbit_gap}} {\tiny\ttfamily Paper4dFrontier/\allowbreak ComputeCostExternalOutputs.lean}
\item \textbf{\nolinkurl{FR325}}\hypertarget{lh:FR325}{}\enspace{\ttfamily\nolinkurl{no_correctOnDomain_randomizedProcedure_classifier_of_orbit_gap}} {\tiny\ttfamily Paper4dFrontier/\allowbreak ComputeCostExternalOutputs.lean}
\item \textbf{\nolinkurl{FR326}}\hypertarget{lh:FR326}{}\enspace{\ttfamily\nolinkurl{TransportedOutputClosureSpec.classifier_agrees_on_closureEquivalent_of_correctOnDomain}} {\tiny\ttfamily Paper4dFrontier/\allowbreak ComputeCostExternalOutputs.lean}
\item \textbf{\nolinkurl{FR327}}\hypertarget{lh:FR327}{}\enspace{\ttfamily\nolinkurl{TransportedOutputClosureSpec.no_correctOnDomain_classifier_of_orbit_gap}} {\tiny\ttfamily Paper4dFrontier/\allowbreak ComputeCostExternalOutputs.lean}
\item \textbf{\nolinkurl{FR328}}\hypertarget{lh:FR328}{}\enspace{\ttfamily\nolinkurl{IdentityOutputClosureSpec.correctOnDomain_iff_correctOnDomain_transport}} {\tiny\ttfamily Paper4dFrontier/\allowbreak ComputeCostExternalOutputs.lean}
\item \textbf{\nolinkurl{FR329}}\hypertarget{lh:FR329}{}\enspace{\ttfamily\nolinkurl{representationRelativeHypothesis_classifier_agrees_on_closureEquivalent_of_correctOnDomain}} {\tiny\ttfamily Paper4dFrontier/\allowbreak ComputeCostExternalOutputs.lean}
\item \textbf{\nolinkurl{FR330}}\hypertarget{lh:FR330}{}\enspace{\ttfamily\nolinkurl{representationRelativeEstimator_classifier_agrees_on_closureEquivalent_of_correctOnDomain}} {\tiny\ttfamily Paper4dFrontier/\allowbreak ComputeCostExternalOutputs.lean}
\item \textbf{\nolinkurl{FR331}}\hypertarget{lh:FR331}{}\enspace{\ttfamily\nolinkurl{representationRelativePolicy_classifier_agrees_on_closureEquivalent_of_correctOnDomain}} {\tiny\ttfamily Paper4dFrontier/\allowbreak ComputeCostExternalOutputs.lean}
\item \textbf{\nolinkurl{FR332}}\hypertarget{lh:FR332}{}\enspace{\ttfamily\nolinkurl{representationRelativeRandomizedProcedure_classifier_agrees_on_closureEquivalent_of_correctOnDomain}} {\tiny\ttfamily Paper4dFrontier/\allowbreak ComputeCostExternalOutputs.lean}
\item \textbf{\nolinkurl{FR333}}\hypertarget{lh:FR333}{}\enspace{\ttfamily\nolinkurl{no_correctOnDomain_representationRelativeHypothesis_classifier_of_orbit_gap}} {\tiny\ttfamily Paper4dFrontier/\allowbreak ComputeCostExternalOutputs.lean}
\item \textbf{\nolinkurl{FR334}}\hypertarget{lh:FR334}{}\enspace{\ttfamily\nolinkurl{no_correctOnDomain_representationRelativeEstimator_classifier_of_orbit_gap}} {\tiny\ttfamily Paper4dFrontier/\allowbreak ComputeCostExternalOutputs.lean}
\item \textbf{\nolinkurl{FR335}}\hypertarget{lh:FR335}{}\enspace{\ttfamily\nolinkurl{no_correctOnDomain_representationRelativePolicy_classifier_of_orbit_gap}} {\tiny\ttfamily Paper4dFrontier/\allowbreak ComputeCostExternalOutputs.lean}
\item \textbf{\nolinkurl{FR336}}\hypertarget{lh:FR336}{}\enspace{\ttfamily\nolinkurl{no_correctOnDomain_representationRelativeRandomizedProcedure_classifier_of_orbit_gap}} {\tiny\ttfamily Paper4dFrontier/\allowbreak ComputeCostExternalOutputs.lean}
\item \textbf{\nolinkurl{FR340}}\hypertarget{lh:FR340}{}\enspace{\ttfamily\nolinkurl{booleanPayloadTransfer}} {\tiny\ttfamily Paper4dFrontier/\allowbreak Realizability.lean}
\item \textbf{\nolinkurl{FR341}}\hypertarget{lh:FR341}{}\enspace{\ttfamily\nolinkurl{predicateTransfer}} {\tiny\ttfamily Paper4dFrontier/\allowbreak Realizability.lean}
\item \textbf{\nolinkurl{FR342}}\hypertarget{lh:FR342}{}\enspace{\ttfamily\nolinkurl{sufficiency_is_relation_refinement}} {\tiny\ttfamily Paper4dFrontier/\allowbreak Realizability.lean}
\item \textbf{\nolinkurl{FR343}}\hypertarget{lh:FR343}{}\enspace{\ttfamily\nolinkurl{relevance_is_erased_failure_of_refinement}} {\tiny\ttfamily Paper4dFrontier/\allowbreak Realizability.lean}
\item \textbf{\nolinkurl{FR345}}\hypertarget{lh:FR345}{}\enspace{\ttfamily\nolinkurl{exactSemanticsQuotient_universal_characterization}} {\tiny\ttfamily Paper4dFrontier/\allowbreak Realizability.lean}
\item \textbf{\nolinkurl{FR346}}\hypertarget{lh:FR346}{}\enspace{\ttfamily\nolinkurl{exactnessMeansExactAgreementWithValidity}} {\tiny\ttfamily Paper4dFrontier/\allowbreak Realizability.lean}
\item \textbf{\nolinkurl{FR348}}\hypertarget{lh:FR348}{}\enspace{\ttfamily\nolinkurl{fullBinaryPairwiseDomain_closure_closed}} {\tiny\ttfamily Paper4dFrontier/\allowbreak FullBinaryPairwiseDomain.lean}
\item \textbf{\nolinkurl{FR349}}\hypertarget{lh:FR349}{}\enspace{\ttfamily\nolinkurl{fullBinaryPairwiseDomain_contains_four_obstruction_families}} {\tiny\ttfamily Paper4dFrontier/\allowbreak FullBinaryPairwiseDomain.lean}
\item \textbf{\nolinkurl{FR350}}\hypertarget{lh:FR350}{}\enspace{\ttfamily\nolinkurl{no_correct_tractability_classifier_on_fullBinaryPairwiseDomain_each_obstruction_family}} {\tiny\ttfamily Paper4dFrontier/\allowbreak FullBinaryPairwiseDomain.lean}
\item \textbf{\nolinkurl{FR351}}\hypertarget{lh:FR351}{}\enspace{\ttfamily\nolinkurl{ExactCorrectnessSpecification.universal_scope_over_rigorously_specified_problems}} {\tiny\ttfamily Paper4dFrontier/\allowbreak Realizability.lean}
\item \textbf{\nolinkurl{FR352}}\hypertarget{lh:FR352}{}\enspace{\ttfamily\nolinkurl{DecisionQuotient.DecisionProblem.minimalSufficient_eq_relevant'}} {\tiny\ttfamily paper4/\allowbreak DecisionQuotient/\allowbreak Sufficiency.lean}
\item \textbf{\nolinkurl{FR353}}\hypertarget{lh:FR353}{}\enspace{\ttfamily\nolinkurl{DecisionQuotient.DecisionProblem.sufficientSets_principal'}} {\tiny\ttfamily paper4/\allowbreak DecisionQuotient/\allowbreak Sufficiency.lean}
\item \textbf{\nolinkurl{FR354}}\hypertarget{lh:FR354}{}\enspace{\ttfamily\nolinkurl{DecisionQuotient.srank_le_sufficient_card}} {\tiny\ttfamily paper4/\allowbreak DecisionQuotient/\allowbreak Tractability/\allowbreak StructuralRank.lean}
\item \textbf{\nolinkurl{FR355}}\hypertarget{lh:FR355}{}\enspace{\ttfamily\nolinkurl{closureLawInvariant_iff_closureHull_eq_self}} {\tiny\ttfamily Paper4dFrontier/\allowbreak AdmissibleCharacterization.lean}
\item \textbf{\nolinkurl{FR356}}\hypertarget{lh:FR356}{}\enspace{\ttfamily\nolinkurl{no_orbitGapOn_iff_closureHull_disjoint}} {\tiny\ttfamily Paper4dFrontier/\allowbreak AdmissibleCharacterization.lean}
\item \textbf{\nolinkurl{FR357}}\hypertarget{lh:FR357}{}\enspace{\ttfamily\nolinkurl{exact_classifiable_by_closureLawInvariant_onDomain_iff_closureHull_disjoint}} {\tiny\ttfamily Paper4dFrontier/\allowbreak AdmissibleCharacterization.lean}
\item \textbf{\nolinkurl{FR358}}\hypertarget{lh:FR358}{}\enspace{\ttfamily\nolinkurl{closureHull_least_exact_classifier_onDomain_of_no_orbitGapOn}} {\tiny\ttfamily Paper4dFrontier/\allowbreak AdmissibleCharacterization.lean}
\item \textbf{\nolinkurl{FR359}}\hypertarget{lh:FR359}{}\enspace{\ttfamily\nolinkurl{binary_pairwise_symmetry_decisionRelevantGraph_dichotomy}} {\tiny\ttfamily Paper4dFrontier/\allowbreak DecisionRelevantPairwiseDichotomy.lean}
\item \textbf{\nolinkurl{FR361}}\hypertarget{lh:FR361}{}\enspace{\ttfamily\nolinkurl{no_optimizerSetSearch_exactSpecification_characterization_decides_any_obstruction_family_on_fullBinaryPairwiseDomain}} {\tiny\ttfamily Paper4dFrontier/\allowbreak ExactSpecificationNoGo.lean}
\item \textbf{\nolinkurl{FR362}}\hypertarget{lh:FR362}{}\enspace{\ttfamily\nolinkurl{closureLawInvariant_iff_of_closureEquivalent}} {\tiny\ttfamily Paper4dFrontier/\allowbreak AdmissibleCharacterization.lean}
\item \textbf{\nolinkurl{FR363}}\hypertarget{lh:FR363}{}\enspace{\ttfamily\nolinkurl{widenedBoundedPatternDefinable_iff_hanfLocalFOOverExtractedSyntax}} {\tiny\ttfamily Paper4dFrontier/\allowbreak AdmissibleCharacterization.lean}
\item \textbf{\nolinkurl{FR364}}\hypertarget{lh:FR364}{}\enspace{\ttfamily\nolinkurl{LocalPattern.occurrenceFormula_holds_iff_occursInSyntax}} {\tiny\ttfamily Paper4dFrontier/\allowbreak HanfGaifman.lean}
\item \textbf{\nolinkurl{FR365}}\hypertarget{lh:FR365}{}\enspace{\ttfamily\nolinkurl{extractedBasicLocalDefinable_implies_widenedBoundedPatternDefinable}} {\tiny\ttfamily Paper4dFrontier/\allowbreak HanfGaifman.lean}
\item \textbf{\nolinkurl{FR366}}\hypertarget{lh:FR366}{}\enspace{\ttfamily\nolinkurl{extractedCountedBasicLocalDefinable_iff_widenedBoundedPatternDefinable}} {\tiny\ttfamily Paper4dFrontier/\allowbreak HanfGaifman.lean}
\item \textbf{\nolinkurl{FR367}}\hypertarget{lh:FR367}{}\enspace{\ttfamily\nolinkurl{DescendsAlong}} {\tiny\ttfamily Paper4dFrontier/\allowbreak QuotientDescent.lean}
\item \textbf{\nolinkurl{FR368}}\hypertarget{lh:FR368}{}\enspace{\ttfamily\nolinkurl{FiberInvariant}} {\tiny\ttfamily Paper4dFrontier/\allowbreak QuotientDescent.lean}
\item \textbf{\nolinkurl{FR369}}\hypertarget{lh:FR369}{}\enspace{\ttfamily\nolinkurl{FiberGap}} {\tiny\ttfamily Paper4dFrontier/\allowbreak QuotientDescent.lean}
\item \textbf{\nolinkurl{FR370}}\hypertarget{lh:FR370}{}\enspace{\ttfamily\nolinkurl{descendsAlong_iff_fiberInvariant}} {\tiny\ttfamily Paper4dFrontier/\allowbreak QuotientDescent.lean}
\item \textbf{\nolinkurl{FR371}}\hypertarget{lh:FR371}{}\enspace{\ttfamily\nolinkurl{fiberGap_obstructs_descendsAlong}} {\tiny\ttfamily Paper4dFrontier/\allowbreak QuotientDescent.lean}
\item \textbf{\nolinkurl{FR372}}\hypertarget{lh:FR372}{}\enspace{\ttfamily\nolinkurl{not_descendsAlong_iff_fiberGap}} {\tiny\ttfamily Paper4dFrontier/\allowbreak QuotientDescent.lean}
\item \textbf{\nolinkurl{FR373}}\hypertarget{lh:FR373}{}\enspace{\ttfamily\nolinkurl{descendsToQuotient_iff_setoidInvariant}} {\tiny\ttfamily Paper4dFrontier/\allowbreak QuotientDescent.lean}
\item \textbf{\nolinkurl{FR374}}\hypertarget{lh:FR374}{}\enspace{\ttfamily\nolinkurl{setoidGap_obstructs_descendsToQuotient}} {\tiny\ttfamily Paper4dFrontier/\allowbreak QuotientDescent.lean}
\item \textbf{\nolinkurl{FR375}}\hypertarget{lh:FR375}{}\enspace{\ttfamily\nolinkurl{not_descendsToQuotient_iff_setoidGap}} {\tiny\ttfamily Paper4dFrontier/\allowbreak QuotientDescent.lean}
\item \textbf{\nolinkurl{FR376}}\hypertarget{lh:FR376}{}\enspace{\ttfamily\nolinkurl{FiniteOrbitCatalogue.sameOrbitDecidable}} {\tiny\ttfamily Paper4dFrontier/\allowbreak FiniteOrbitCatalogue.lean}
\item \textbf{\nolinkurl{FR377}}\hypertarget{lh:FR377}{}\enspace{\ttfamily\nolinkurl{FiniteOrbitCatalogue.classifier_orbitInvariant}} {\tiny\ttfamily Paper4dFrontier/\allowbreak FiniteOrbitCatalogue.lean}
\item \textbf{\nolinkurl{FR378}}\hypertarget{lh:FR378}{}\enspace{\ttfamily\nolinkurl{FiniteOrbitCatalogue.exact_of_factorsThrough}} {\tiny\ttfamily Paper4dFrontier/\allowbreak FiniteOrbitCatalogue.lean}
\item \textbf{\nolinkurl{FR379}}\hypertarget{lh:FR379}{}\enspace{\ttfamily\nolinkurl{FiniteOrbitCatalogue.factorsThrough_of_orbitInvariant}} {\tiny\ttfamily Paper4dFrontier/\allowbreak FiniteOrbitCatalogue.lean}
\item \textbf{\nolinkurl{FR380}}\hypertarget{lh:FR380}{}\enspace{\ttfamily\nolinkurl{FiniteOrbitCatalogue.orbitInvariant_iff_factorsThrough}} {\tiny\ttfamily Paper4dFrontier/\allowbreak FiniteOrbitCatalogue.lean}
\item \textbf{\nolinkurl{FR381}}\hypertarget{lh:FR381}{}\enspace{\ttfamily\nolinkurl{FiniteOrbitCatalogue.boolClassifier_orbitInvariant}} {\tiny\ttfamily Paper4dFrontier/\allowbreak FiniteOrbitCatalogue.lean}
\item \textbf{\nolinkurl{FR382}}\hypertarget{lh:FR382}{}\enspace{\ttfamily\nolinkurl{FiniteOrbitCatalogue.boolClassifier_correct}} {\tiny\ttfamily Paper4dFrontier/\allowbreak FiniteOrbitCatalogue.lean}
\item \textbf{\nolinkurl{FR383}}\hypertarget{lh:FR383}{}\enspace{\ttfamily\nolinkurl{FiniteOrbitCatalogue.quotient_decider_exact}} {\tiny\ttfamily Paper4dFrontier/\allowbreak FiniteOrbitCatalogue.lean}
\item \textbf{\nolinkurl{FR384}}\hypertarget{lh:FR384}{}\enspace{\ttfamily\nolinkurl{FiniteOrbitCatalogue.factorsThrough_not}} {\tiny\ttfamily Paper4dFrontier/\allowbreak FiniteOrbitCatalogue.lean}
\item \textbf{\nolinkurl{FR385}}\hypertarget{lh:FR385}{}\enspace{\ttfamily\nolinkurl{FiniteOrbitCatalogue.factorsThrough_and}} {\tiny\ttfamily Paper4dFrontier/\allowbreak FiniteOrbitCatalogue.lean}
\item \textbf{\nolinkurl{FR386}}\hypertarget{lh:FR386}{}\enspace{\ttfamily\nolinkurl{FiniteOrbitCatalogue.factorsThrough_or}} {\tiny\ttfamily Paper4dFrontier/\allowbreak FiniteOrbitCatalogue.lean}
\item \textbf{\nolinkurl{FR387}}\hypertarget{lh:FR387}{}\enspace{\ttfamily\nolinkurl{FiniteOrbitCatalogue.orbitInvariant_not}} {\tiny\ttfamily Paper4dFrontier/\allowbreak FiniteOrbitCatalogue.lean}
\item \textbf{\nolinkurl{FR388}}\hypertarget{lh:FR388}{}\enspace{\ttfamily\nolinkurl{FiniteOrbitCatalogue.orbitInvariant_and}} {\tiny\ttfamily Paper4dFrontier/\allowbreak FiniteOrbitCatalogue.lean}
\item \textbf{\nolinkurl{FR389}}\hypertarget{lh:FR389}{}\enspace{\ttfamily\nolinkurl{FiniteOrbitCatalogue.orbitInvariant_or}} {\tiny\ttfamily Paper4dFrontier/\allowbreak FiniteOrbitCatalogue.lean}
\item \textbf{\nolinkurl{FR390}}\hypertarget{lh:FR390}{}\enspace{\ttfamily\nolinkurl{NormalizedUnaryGapCatalogue.factorsThrough_iff_orbitInvariant}} {\tiny\ttfamily Paper4dFrontier/\allowbreak FiniteOrbitCatalogue.lean}
\item \textbf{\nolinkurl{FR391}}\hypertarget{lh:FR391}{}\enspace{\ttfamily\nolinkurl{NormalizedUnaryGapCatalogue.no_orbitGap_of_factorsThrough}} {\tiny\ttfamily Paper4dFrontier/\allowbreak FiniteOrbitCatalogue.lean}
\item \textbf{\nolinkurl{FR392}}\hypertarget{lh:FR392}{}\enspace{\ttfamily\nolinkurl{NormalizedUnaryGapCatalogue.classifier_exact_of_factorsThrough}} {\tiny\ttfamily Paper4dFrontier/\allowbreak FiniteOrbitCatalogue.lean}
\item \textbf{\nolinkurl{FR393}}\hypertarget{lh:FR393}{}\enspace{\ttfamily\nolinkurl{NormalizedUnaryGapCatalogue.boolClassifier_exact}} {\tiny\ttfamily Paper4dFrontier/\allowbreak FiniteOrbitCatalogue.lean}
\item \textbf{\nolinkurl{FR394}}\hypertarget{lh:FR394}{}\enspace{\ttfamily\nolinkurl{HasActionGapPairInteraction}} {\tiny\ttfamily Paper4dFrontier/\allowbreak ActionGapTreewidth.lean}
\item \textbf{\nolinkurl{FR395}}\hypertarget{lh:FR395}{}\enspace{\ttfamily\nolinkurl{actionGapInteractionGraph_graphRealizingActionGapUtility_eq}} {\tiny\ttfamily Paper4dFrontier/\allowbreak ActionGapTreewidth.lean}
\item \textbf{\nolinkurl{FR396}}\hypertarget{lh:FR396}{}\enspace{\ttfamily\nolinkurl{actionGapTreewidthRecognition_realizes_graph}} {\tiny\ttfamily Paper4dFrontier/\allowbreak ActionGapTreewidth.lean}
\item \textbf{\nolinkurl{FR397}}\hypertarget{lh:FR397}{}\enspace{\ttfamily\nolinkurl{actionGapPairInteraction_requires_two_actions}} {\tiny\ttfamily Paper4dFrontier/\allowbreak ActionGapTreewidth.lean}
\item \textbf{\nolinkurl{FR398}}\hypertarget{lh:FR398}{}\enspace{\ttfamily\nolinkurl{actionGapInteractionGraph_subsingleton_actions_eq_bot}} {\tiny\ttfamily Paper4dFrontier/\allowbreak ActionGapTreewidth.lean}
\item \textbf{\nolinkurl{FR399}}\hypertarget{lh:FR399}{}\enspace{\ttfamily\nolinkurl{HasTwoPairOrbitTemplate}} {\tiny\ttfamily Paper4dFrontier/\allowbreak ActionGapTreewidth.lean}
\item \textbf{\nolinkurl{FR400}}\hypertarget{lh:FR400}{}\enspace{\ttfamily\nolinkurl{twoPairOrbitTemplate_requires_three_coordinates}} {\tiny\ttfamily Paper4dFrontier/\allowbreak ActionGapTreewidth.lean}
\item \textbf{\nolinkurl{FR401}}\hypertarget{lh:FR401}{}\enspace{\ttfamily\nolinkurl{ExactCertificationReduct}} {\tiny\ttfamily Paper4dFrontier/\allowbreak RoughSetBridge.lean}
\item \textbf{\nolinkurl{FR402}}\hypertarget{lh:FR402}{}\enspace{\ttfamily\nolinkurl{exactCertificationReduct_iff_core}} {\tiny\ttfamily Paper4dFrontier/\allowbreak RoughSetBridge.lean}
\item \textbf{\nolinkurl{FR403}}\hypertarget{lh:FR403}{}\enspace{\ttfamily\nolinkurl{exactCertificationSufficientSets_principal_at_reduct}} {\tiny\ttfamily Paper4dFrontier/\allowbreak RoughSetBridge.lean}
\item \textbf{\nolinkurl{FR404}}\hypertarget{lh:FR404}{}\enspace{\ttfamily\nolinkurl{HasAnyUniqueDominantPair}} {\tiny\ttfamily Paper4dFrontier/\allowbreak ObstructionPredicateCandidates.lean}
\item \textbf{\nolinkurl{FR405}}\hypertarget{lh:FR405}{}\enspace{\ttfamily\nolinkurl{dominantPair_orbit_witness_preserves_anyUniqueDominantPair}} {\tiny\ttfamily Paper4dFrontier/\allowbreak ObstructionPredicateCandidates.lean}
\item \textbf{\nolinkurl{FR406}}\hypertarget{lh:FR406}{}\enspace{\ttfamily\nolinkurl{addDuplicateActionSlice}} {\tiny\ttfamily Paper4dFrontier/\allowbreak ObstructionPredicateCandidates.lean}
\item \textbf{\nolinkurl{FR407}}\hypertarget{lh:FR407}{}\enspace{\ttfamily\nolinkurl{addDuplicateActionSlice_duplicateActionWitness}} {\tiny\ttfamily Paper4dFrontier/\allowbreak ObstructionPredicateCandidates.lean}
\item \textbf{\nolinkurl{FR408}}\hypertarget{lh:FR408}{}\enspace{\ttfamily\nolinkurl{RawActionCountAtLeast}} {\tiny\ttfamily Paper4dFrontier/\allowbreak ObstructionPredicateCandidates.lean}
\item \textbf{\nolinkurl{FR409}}\hypertarget{lh:FR409}{}\enspace{\ttfamily\nolinkurl{rawActionCountThresholdDuplicationOrbitGap_exists}} {\tiny\ttfamily Paper4dFrontier/\allowbreak ObstructionPredicateCandidates.lean}
\item \textbf{\nolinkurl{FR410}}\hypertarget{lh:FR410}{}\enspace{\ttfamily\nolinkurl{no_closureInvariant_predicate_decides_rawActionCountAtLeastThreshold}} {\tiny\ttfamily Paper4dFrontier/\allowbreak ObstructionPredicateCandidates.lean}
\item \textbf{\nolinkurl{FR411}}\hypertarget{lh:FR411}{}\enspace{\ttfamily\nolinkurl{quotientFlipAdj}} {\tiny\ttfamily Paper4dFrontier/\allowbreak QuotientMSORealization.lean}
\item \textbf{\nolinkurl{FR412}}\hypertarget{lh:FR412}{}\enspace{\ttfamily\nolinkurl{quotientFlipGraph}} {\tiny\ttfamily Paper4dFrontier/\allowbreak QuotientMSORealization.lean}
\item \textbf{\nolinkurl{FR413}}\hypertarget{lh:FR413}{}\enspace{\ttfamily\nolinkurl{quotientFlipAdj_implies_graph_adj}} {\tiny\ttfamily Paper4dFrontier/\allowbreak QuotientMSORealization.lean}
\item \textbf{\nolinkurl{FR414}}\hypertarget{lh:FR414}{}\enspace{\ttfamily\nolinkurl{graph_adj_implies_quotientFlipAdj}} {\tiny\ttfamily Paper4dFrontier/\allowbreak QuotientMSORealization.lean}
\item \textbf{\nolinkurl{FR415}}\hypertarget{lh:FR415}{}\enspace{\ttfamily\nolinkurl{quotientFlipGraph_realizes_graph}} {\tiny\ttfamily Paper4dFrontier/\allowbreak QuotientMSORealization.lean}
\item \textbf{\nolinkurl{FR416}}\hypertarget{lh:FR416}{}\enspace{\ttfamily\nolinkurl{quotientFlipCorrectnessSet_eq_iff_label_eq}} {\tiny\ttfamily Paper4dFrontier/\allowbreak QuotientMSORealization.lean}
\item \textbf{\nolinkurl{FR417}}\hypertarget{lh:FR417}{}\enspace{\ttfamily\nolinkurl{correctnessQuotientFlipPresentation_realizes_graph}} {\tiny\ttfamily Paper4dFrontier/\allowbreak QuotientMSORealization.lean}
\item \textbf{\nolinkurl{FR418}}\hypertarget{lh:FR418}{}\enspace{\ttfamily\nolinkurl{quotientFlip_same_label_states_sameQuotient}} {\tiny\ttfamily Paper4dFrontier/\allowbreak QuotientMSORealization.lean}
\item \textbf{\nolinkurl{FR419}}\hypertarget{lh:FR419}{}\enspace{\ttfamily\nolinkurl{quotientFlip_state_sameQuotient_vertex_of_label}} {\tiny\ttfamily Paper4dFrontier/\allowbreak QuotientMSORealization.lean}
\item \textbf{\nolinkurl{FR420}}\hypertarget{lh:FR420}{}\enspace{\ttfamily\nolinkurl{quotientFlip_edge_false_sameQuotient_source_vertex}} {\tiny\ttfamily Paper4dFrontier/\allowbreak QuotientMSORealization.lean}
\item \textbf{\nolinkurl{FR421}}\hypertarget{lh:FR421}{}\enspace{\ttfamily\nolinkurl{quotientFlip_edge_true_sameQuotient_target_vertex}} {\tiny\ttfamily Paper4dFrontier/\allowbreak QuotientMSORealization.lean}
\item \textbf{\nolinkurl{FR422}}\hypertarget{lh:FR422}{}\enspace{\ttfamily\nolinkurl{quotientFlipState_vertex_or_edge}} {\tiny\ttfamily Paper4dFrontier/\allowbreak QuotientMSORealization.lean}
\item \textbf{\nolinkurl{FR423}}\hypertarget{lh:FR423}{}\enspace{\ttfamily\nolinkurl{QuotientMSOPostDescentEvaluator}} {\tiny\ttfamily Paper4dFrontier/\allowbreak QuotientMSORealization.lean}
\item \textbf{\nolinkurl{FR424}}\hypertarget{lh:FR424}{}\enspace{\ttfamily\nolinkurl{QuotientMSOPostDescentEvaluator.predicate_eq_run}} {\tiny\ttfamily Paper4dFrontier/\allowbreak QuotientMSORealization.lean}
\item \textbf{\nolinkurl{FR425}}\hypertarget{lh:FR425}{}\enspace{\ttfamily\nolinkurl{QuotientMSOPostDescentEvaluator.run_eq_eval_quotient}} {\tiny\ttfamily Paper4dFrontier/\allowbreak QuotientMSORealization.lean}
\item \textbf{\nolinkurl{FR426}}\hypertarget{lh:FR426}{}\enspace{\ttfamily\nolinkurl{QuotientMSOPostDescentEvaluator.predicate_eq_of_quotient_eq}} {\tiny\ttfamily Paper4dFrontier/\allowbreak QuotientMSORealization.lean}
\item \textbf{\nolinkurl{FR427}}\hypertarget{lh:FR427}{}\enspace{\ttfamily\nolinkurl{QuotientMSOPostDescentEvaluator.totalCost_eq_construction_add_modelCheck}} {\tiny\ttfamily Paper4dFrontier/\allowbreak QuotientMSORealization.lean}
\item \textbf{\nolinkurl{FR428}}\hypertarget{lh:FR428}{}\enspace{\ttfamily\nolinkurl{MixedStateMSOCoordinateFOEvaluator}} {\tiny\ttfamily Paper4dFrontier/\allowbreak QuotientMSORealization.lean}
\item \textbf{\nolinkurl{FR429}}\hypertarget{lh:FR429}{}\enspace{\ttfamily\nolinkurl{MixedStateMSOCoordinateFOEvaluator.predicate_eq_run}} {\tiny\ttfamily Paper4dFrontier/\allowbreak QuotientMSORealization.lean}
\item \textbf{\nolinkurl{FR430}}\hypertarget{lh:FR430}{}\enspace{\ttfamily\nolinkurl{MixedStateMSOCoordinateFOEvaluator.mixedModelCheckCost_eq}} {\tiny\ttfamily Paper4dFrontier/\allowbreak QuotientMSORealization.lean}
\item \textbf{\nolinkurl{FR431}}\hypertarget{lh:FR431}{}\enspace{\ttfamily\nolinkurl{MixedStateMSOCoordinateFOEvaluator.mixedTotalCost_eq}} {\tiny\ttfamily Paper4dFrontier/\allowbreak QuotientMSORealization.lean}
\item \textbf{\nolinkurl{FR432}}\hypertarget{lh:FR432}{}\enspace{\ttfamily\nolinkurl{MixedStateMSOCoordinateFOEvaluator.mixedModelCheckCost_eq_of_quotient_eq_of_coordinateCount_eq}} {\tiny\ttfamily Paper4dFrontier/\allowbreak QuotientMSORealization.lean}
\item \textbf{\nolinkurl{FR433}}\hypertarget{lh:FR433}{}\enspace{\ttfamily\nolinkurl{MixedStateMSOCoordinateFOEvaluator.setAssignmentFactor_eq_one_of_no_setQuantifiers}} {\tiny\ttfamily Paper4dFrontier/\allowbreak QuotientMSORealization.lean}
\item \textbf{\nolinkurl{FR434}}\hypertarget{lh:FR434}{}\enspace{\ttfamily\nolinkurl{MixedStateMSOCoordinateFOEvaluator.mixedModelCheckCost_eq_fo_only}} {\tiny\ttfamily Paper4dFrontier/\allowbreak QuotientMSORealization.lean}
\item \textbf{\nolinkurl{FR435}}\hypertarget{lh:FR435}{}\enspace{\ttfamily\nolinkurl{MixedStateMSOCoordinateFOEvaluator.setAssignmentFactor_le_of_bounded_quotient_sorts}} {\tiny\ttfamily Paper4dFrontier/\allowbreak QuotientMSORealization.lean}
\item \textbf{\nolinkurl{FR436}}\hypertarget{lh:FR436}{}\enspace{\ttfamily\nolinkurl{ProxyLanguage}} {\tiny\ttfamily Paper4dFrontier/\allowbreak ProxyDescentCertification.lean}
\item \textbf{\nolinkurl{FR437}}\hypertarget{lh:FR437}{}\enspace{\ttfamily\nolinkurl{ProxyLanguage.Descends}} {\tiny\ttfamily Paper4dFrontier/\allowbreak ProxyDescentCertification.lean}
\item \textbf{\nolinkurl{FR438}}\hypertarget{lh:FR438}{}\enspace{\ttfamily\nolinkurl{ProxyLanguage.DescentCertificate}} {\tiny\ttfamily Paper4dFrontier/\allowbreak ProxyDescentCertification.lean}
\item \textbf{\nolinkurl{FR439}}\hypertarget{lh:FR439}{}\enspace{\ttfamily\nolinkurl{ProxyLanguage.DescentCounterexample}} {\tiny\ttfamily Paper4dFrontier/\allowbreak ProxyDescentCertification.lean}
\item \textbf{\nolinkurl{FR440}}\hypertarget{lh:FR440}{}\enspace{\ttfamily\nolinkurl{ProxyLanguage.descends_iff_no_descentCounterexample}} {\tiny\ttfamily Paper4dFrontier/\allowbreak ProxyDescentCertification.lean}
\item \textbf{\nolinkurl{FR441}}\hypertarget{lh:FR441}{}\enspace{\ttfamily\nolinkurl{ProxyLanguage.not_descends_of_descentCounterexample}} {\tiny\ttfamily Paper4dFrontier/\allowbreak ProxyDescentCertification.lean}
\item \textbf{\nolinkurl{FR442}}\hypertarget{lh:FR442}{}\enspace{\ttfamily\nolinkurl{ProxyLanguage.no_exact_closureInvariant_classifier_of_descentCounterexample}} {\tiny\ttfamily Paper4dFrontier/\allowbreak ProxyDescentCertification.lean}
\item \textbf{\nolinkurl{FR443}}\hypertarget{lh:FR443}{}\enspace{\ttfamily\nolinkurl{ProxyLanguage.exact_closureInvariant_classifier_iff_descends}} {\tiny\ttfamily Paper4dFrontier/\allowbreak ProxyDescentCertification.lean}
\item \textbf{\nolinkurl{FR444}}\hypertarget{lh:FR444}{}\enspace{\ttfamily\nolinkurl{commonGraphSlice}} {\tiny\ttfamily Paper4dFrontier/\allowbreak RawSupportGraphNoGo.lean}
\item \textbf{\nolinkurl{FR445}}\hypertarget{lh:FR445}{}\enspace{\ttfamily\nolinkurl{commonGraphSlice_rawInteractionGraph_eq}} {\tiny\ttfamily Paper4dFrontier/\allowbreak RawSupportGraphNoGo.lean}
\item \textbf{\nolinkurl{FR446}}\hypertarget{lh:FR446}{}\enspace{\ttfamily\nolinkurl{commonGraphPositiveAffineWitness}} {\tiny\ttfamily Paper4dFrontier/\allowbreak RawSupportGraphNoGo.lean}
\item \textbf{\nolinkurl{FR447}}\hypertarget{lh:FR447}{}\enspace{\ttfamily\nolinkurl{commonGraphSlice_closureEquivalent}} {\tiny\ttfamily Paper4dFrontier/\allowbreak RawSupportGraphNoGo.lean}
\item \textbf{\nolinkurl{FR448}}\hypertarget{lh:FR448}{}\enspace{\ttfamily\nolinkurl{commonGraphSlice_actionGapInteractionGraph_eq_bot}} {\tiny\ttfamily Paper4dFrontier/\allowbreak RawSupportGraphNoGo.lean}
\item \textbf{\nolinkurl{FR449}}\hypertarget{lh:FR449}{}\enspace{\ttfamily\nolinkurl{RawSupportGraphPredicate}} {\tiny\ttfamily Paper4dFrontier/\allowbreak RawSupportGraphNoGo.lean}
\item \textbf{\nolinkurl{FR450}}\hypertarget{lh:FR450}{}\enspace{\ttfamily\nolinkurl{rawSupportGraphPredicate_orbitGap_of_graphPredicate_separates}} {\tiny\ttfamily Paper4dFrontier/\allowbreak RawSupportGraphNoGo.lean}
\item \textbf{\nolinkurl{FR451}}\hypertarget{lh:FR451}{}\enspace{\ttfamily\nolinkurl{no_closureInvariant_classifier_decides_rawSupportGraphPredicate_of_separates}} {\tiny\ttfamily Paper4dFrontier/\allowbreak RawSupportGraphNoGo.lean}
\item \textbf{\nolinkurl{FR452}}\hypertarget{lh:FR452}{}\enspace{\ttfamily\nolinkurl{SliceCoordinateRelevant}} {\tiny\ttfamily Paper4dFrontier/\allowbreak DescentRelevanceSeparation.lean}
\item \textbf{\nolinkurl{FR453}}\hypertarget{lh:FR453}{}\enspace{\ttfamily\nolinkurl{NoSliceRelevantCoordinates}} {\tiny\ttfamily Paper4dFrontier/\allowbreak DescentRelevanceSeparation.lean}
\item \textbf{\nolinkurl{FR454}}\hypertarget{lh:FR454}{}\enspace{\ttfamily\nolinkurl{rawActionCountThresholdGap_with_noRelevantCoordinates}} {\tiny\ttfamily Paper4dFrontier/\allowbreak DescentRelevanceSeparation.lean}
\item \textbf{\nolinkurl{FR455}}\hypertarget{lh:FR455}{}\enspace{\ttfamily\nolinkurl{FixedWidthGraphSolver}} {\tiny\ttfamily Paper4dFrontier/\allowbreak ActionGapTreewidthTractability.lean}
\item \textbf{\nolinkurl{FR456}}\hypertarget{lh:FR456}{}\enspace{\ttfamily\nolinkurl{FixedWidthGraphSolver.solveActionGap_steps_bound}} {\tiny\ttfamily Paper4dFrontier/\allowbreak ActionGapTreewidthTractability.lean}
\item \textbf{\nolinkurl{FR457}}\hypertarget{lh:FR457}{}\enspace{\ttfamily\nolinkurl{FixedWidthGraphSolver.solveBinaryPairwiseActionGap_steps_bound}} {\tiny\ttfamily Paper4dFrontier/\allowbreak ActionGapTreewidthTractability.lean}
\item \textbf{\nolinkurl{FR458}}\hypertarget{lh:FR458}{}\enspace{\ttfamily\nolinkurl{MixedStateMSOCoordinateFOEvaluator.mixedModelCheckCost_le_of_bounded_quotient_sorts}} {\tiny\ttfamily Paper4dFrontier/\allowbreak QuotientMSORealization.lean}
\item \textbf{\nolinkurl{FR459}}\hypertarget{lh:FR459}{}\enspace{\ttfamily\nolinkurl{MixedStateMSOCoordinateFOEvaluator.mixedTotalCost_le_of_bounded_quotient_sorts}} {\tiny\ttfamily Paper4dFrontier/\allowbreak QuotientMSORealization.lean}
\item \textbf{\nolinkurl{FR460}}\hypertarget{lh:FR460}{}\enspace{\ttfamily\nolinkurl{BoundedQuotientTreewidthMSOEvaluator}} {\tiny\ttfamily Paper4dFrontier/\allowbreak BoundedQuotientTreewidthMSO.lean}
\item \textbf{\nolinkurl{FR461}}\hypertarget{lh:FR461}{}\enspace{\ttfamily\nolinkurl{BoundedQuotientTreewidthMSOEvaluator.predicate_eq_of_quotient_eq}} {\tiny\ttfamily Paper4dFrontier/\allowbreak BoundedQuotientTreewidthMSO.lean}
\item \textbf{\nolinkurl{FR462}}\hypertarget{lh:FR462}{}\enspace{\ttfamily\nolinkurl{BoundedQuotientTreewidthMSOEvaluator.modelCheckCost_le_linear_quotientSize}} {\tiny\ttfamily Paper4dFrontier/\allowbreak BoundedQuotientTreewidthMSO.lean}
\item \textbf{\nolinkurl{FR463}}\hypertarget{lh:FR463}{}\enspace{\ttfamily\nolinkurl{BoundedQuotientTreewidthMSOEvaluator.totalCost_le_construction_add_linear_quotientSize}} {\tiny\ttfamily Paper4dFrontier/\allowbreak BoundedQuotientTreewidthMSO.lean}
\item \textbf{\nolinkurl{FR464}}\hypertarget{lh:FR464}{}\enspace{\ttfamily\nolinkurl{BoundedQuotientTreewidthMSOEvaluator.totalCost_le_polynomialConstruction_add_linear_quotientSize}} {\tiny\ttfamily Paper4dFrontier/\allowbreak BoundedQuotientTreewidthMSO.lean}
\item \textbf{\nolinkurl{FR465}}\hypertarget{lh:FR465}{}\enspace{\ttfamily\nolinkurl{BoundedQuotientTreewidthMSOEvaluator.boolClassifier_eq_of_quotient_eq}} {\tiny\ttfamily Paper4dFrontier/\allowbreak BoundedQuotientTreewidthMSO.lean}
\item \textbf{\nolinkurl{FR466}}\hypertarget{lh:FR466}{}\enspace{\ttfamily\nolinkurl{ClosureCompatibleBoundedQuotientTreewidthMSO}} {\tiny\ttfamily Paper4dFrontier/\allowbreak BoundedQuotientTreewidthMSO.lean}
\item \textbf{\nolinkurl{FR467}}\hypertarget{lh:FR467}{}\enspace{\ttfamily\nolinkurl{ClosureCompatibleBoundedQuotientTreewidthMSO.predicate_eq_of_closureEquivalent}} {\tiny\ttfamily Paper4dFrontier/\allowbreak BoundedQuotientTreewidthMSO.lean}
\item \textbf{\nolinkurl{FR468}}\hypertarget{lh:FR468}{}\enspace{\ttfamily\nolinkurl{ClosureCompatibleBoundedQuotientTreewidthMSO.boolClassifier_eq_of_closureEquivalent}} {\tiny\ttfamily Paper4dFrontier/\allowbreak BoundedQuotientTreewidthMSO.lean}
\item \textbf{\nolinkurl{FR469}}\hypertarget{lh:FR469}{}\enspace{\ttfamily\nolinkurl{ClosureCompatibleBoundedQuotientTreewidthMSO.totalCost_le_polynomialConstruction_add_linear_quotientSize}} {\tiny\ttfamily Paper4dFrontier/\allowbreak BoundedQuotientTreewidthMSO.lean}
\item \textbf{\nolinkurl{FR470}}\hypertarget{lh:FR470}{}\enspace{\ttfamily\nolinkurl{BoundedQuotientSizeMSOEvaluator}} {\tiny\ttfamily Paper4dFrontier/\allowbreak SparseUnaryGapQuotientPipeline.lean}
\item \textbf{\nolinkurl{FR471}}\hypertarget{lh:FR471}{}\enspace{\ttfamily\nolinkurl{BoundedQuotientSizeMSOEvaluator.predicate_eq_of_quotient_eq}} {\tiny\ttfamily Paper4dFrontier/\allowbreak SparseUnaryGapQuotientPipeline.lean}
\item \textbf{\nolinkurl{FR472}}\hypertarget{lh:FR472}{}\enspace{\ttfamily\nolinkurl{BoundedQuotientSizeMSOEvaluator.modelCheckCost_le_finiteModelFactor}} {\tiny\ttfamily Paper4dFrontier/\allowbreak SparseUnaryGapQuotientPipeline.lean}
\item \textbf{\nolinkurl{FR473}}\hypertarget{lh:FR473}{}\enspace{\ttfamily\nolinkurl{BoundedQuotientSizeMSOEvaluator.totalCost_le_polynomialConstruction_add_finiteModelFactor}} {\tiny\ttfamily Paper4dFrontier/\allowbreak SparseUnaryGapQuotientPipeline.lean}
\item \textbf{\nolinkurl{FR474}}\hypertarget{lh:FR474}{}\enspace{\ttfamily\nolinkurl{ClosureCompatibleBoundedQuotientSizeMSO}} {\tiny\ttfamily Paper4dFrontier/\allowbreak SparseUnaryGapQuotientPipeline.lean}
\item \textbf{\nolinkurl{FR475}}\hypertarget{lh:FR475}{}\enspace{\ttfamily\nolinkurl{ClosureCompatibleBoundedQuotientSizeMSO.predicate_eq_of_closureEquivalent}} {\tiny\ttfamily Paper4dFrontier/\allowbreak SparseUnaryGapQuotientPipeline.lean}
\item \textbf{\nolinkurl{FR476}}\hypertarget{lh:FR476}{}\enspace{\ttfamily\nolinkurl{StrictMarginQuotientMSOStability}} {\tiny\ttfamily Paper4dFrontier/\allowbreak SparseUnaryGapQuotientPipeline.lean}
\item \textbf{\nolinkurl{FR477}}\hypertarget{lh:FR477}{}\enspace{\ttfamily\nolinkurl{StrictMarginQuotientMSOStability.predicate_eq_of_uniformApprox_of_uniformStrictGapCover}} {\tiny\ttfamily Paper4dFrontier/\allowbreak SparseUnaryGapQuotientPipeline.lean}
\item \textbf{\nolinkurl{FR478}}\hypertarget{lh:FR478}{}\enspace{\ttfamily\nolinkurl{StrictMarginQuotientMSOStability.boolClassifier_eq_of_uniformApprox_of_uniformStrictGapCover}} {\tiny\ttfamily Paper4dFrontier/\allowbreak SparseUnaryGapQuotientPipeline.lean}
\item \textbf{\nolinkurl{FR479}}\hypertarget{lh:FR479}{}\enspace{\ttfamily\nolinkurl{SparseUnaryGapQuotientPipeline}} {\tiny\ttfamily Paper4dFrontier/\allowbreak SparseUnaryGapQuotientPipeline.lean}
\item \textbf{\nolinkurl{FR480}}\hypertarget{lh:FR480}{}\enspace{\ttfamily\nolinkurl{SparseUnaryGapQuotientPipeline.quotientVertexCount_le_powerSupportBound}} {\tiny\ttfamily Paper4dFrontier/\allowbreak SparseUnaryGapQuotientPipeline.lean}
\item \textbf{\nolinkurl{FR481}}\hypertarget{lh:FR481}{}\enspace{\ttfamily\nolinkurl{SparseUnaryGapQuotientPipeline.constructionCost_le_powerSupportBound}} {\tiny\ttfamily Paper4dFrontier/\allowbreak SparseUnaryGapQuotientPipeline.lean}
\item \textbf{\nolinkurl{FR482}}\hypertarget{lh:FR482}{}\enspace{\ttfamily\nolinkurl{SparseUnaryGapQuotientPipeline.quotientTreewidth_le_powerSupportBound_minus_one}} {\tiny\ttfamily Paper4dFrontier/\allowbreak SparseUnaryGapQuotientPipeline.lean}
\item \textbf{\nolinkurl{FR483}}\hypertarget{lh:FR483}{}\enspace{\ttfamily\nolinkurl{SparseUnaryGapQuotientPipeline.totalCost_le_powerSupportBound_add_linear_powerSupportBound}} {\tiny\ttfamily Paper4dFrontier/\allowbreak SparseUnaryGapQuotientPipeline.lean}
\item \textbf{\nolinkurl{FR484}}\hypertarget{lh:FR484}{}\enspace{\ttfamily\nolinkurl{SparseUnaryGapQuotientPipeline.predicate_eq_of_quotient_eq}} {\tiny\ttfamily Paper4dFrontier/\allowbreak SparseUnaryGapQuotientPipeline.lean}
\item \textbf{\nolinkurl{FR485}}\hypertarget{lh:FR485}{}\enspace{\ttfamily\nolinkurl{RawCoordinateCountAtLeast}} {\tiny\ttfamily Paper4dFrontier/\allowbreak SparseUnaryGapQuotientPipeline.lean}
\item \textbf{\nolinkurl{FR486}}\hypertarget{lh:FR486}{}\enspace{\ttfamily\nolinkurl{constantCoordinateCountSlice}} {\tiny\ttfamily Paper4dFrontier/\allowbreak SparseUnaryGapQuotientPipeline.lean}
\item \textbf{\nolinkurl{FR487}}\hypertarget{lh:FR487}{}\enspace{\ttfamily\nolinkurl{constantCoordinateCountSlice_irrelevantCoordinateWitness}} {\tiny\ttfamily Paper4dFrontier/\allowbreak SparseUnaryGapQuotientPipeline.lean}
\item \textbf{\nolinkurl{FR488}}\hypertarget{lh:FR488}{}\enspace{\ttfamily\nolinkurl{constantCoordinateCountSlice_closureEquivalent_succ}} {\tiny\ttfamily Paper4dFrontier/\allowbreak SparseUnaryGapQuotientPipeline.lean}
\item \textbf{\nolinkurl{FR489}}\hypertarget{lh:FR489}{}\enspace{\ttfamily\nolinkurl{rawCoordinateCountThresholdGap_irrelevantCoordinateExtension}} {\tiny\ttfamily Paper4dFrontier/\allowbreak SparseUnaryGapQuotientPipeline.lean}
\item \textbf{\nolinkurl{FR490}}\hypertarget{lh:FR490}{}\enspace{\ttfamily\nolinkurl{no_closureInvariant_classifier_decides_rawCoordinateCountAtLeast}} {\tiny\ttfamily Paper4dFrontier/\allowbreak SparseUnaryGapQuotientPipeline.lean}
\item \textbf{\nolinkurl{FR491}}\hypertarget{lh:FR491}{}\enspace{\ttfamily\nolinkurl{PredicateReduction}} {\tiny\ttfamily Paper4dFrontier/\allowbreak HardnessTransfer.lean}
\item \textbf{\nolinkurl{FR492}}\hypertarget{lh:FR492}{}\enspace{\ttfamily\nolinkurl{PredicateReduction.trans}} {\tiny\ttfamily Paper4dFrontier/\allowbreak HardnessTransfer.lean}
\item \textbf{\nolinkurl{FR493}}\hypertarget{lh:FR493}{}\enspace{\ttfamily\nolinkurl{SplitPresentation}} {\tiny\ttfamily Paper4dFrontier/\allowbreak HardnessTransfer.lean}
\item \textbf{\nolinkurl{FR494}}\hypertarget{lh:FR494}{}\enspace{\ttfamily\nolinkurl{SplitProxyDescends}} {\tiny\ttfamily Paper4dFrontier/\allowbreak HardnessTransfer.lean}
\item \textbf{\nolinkurl{FR495}}\hypertarget{lh:FR495}{}\enspace{\ttfamily\nolinkurl{SplitProxyOrbitGap}} {\tiny\ttfamily Paper4dFrontier/\allowbreak HardnessTransfer.lean}
\item \textbf{\nolinkurl{FR496}}\hypertarget{lh:FR496}{}\enspace{\ttfamily\nolinkurl{splitProxyDescends_iff_no_orbitGap}} {\tiny\ttfamily Paper4dFrontier/\allowbreak HardnessTransfer.lean}
\item \textbf{\nolinkurl{FR497}}\hypertarget{lh:FR497}{}\enspace{\ttfamily\nolinkurl{SAT3SplitProxyEval}} {\tiny\ttfamily Paper4dFrontier/\allowbreak HardnessTransfer.lean}
\item \textbf{\nolinkurl{FR498}}\hypertarget{lh:FR498}{}\enspace{\ttfamily\nolinkurl{SAT3SplitProxyNonDescends}} {\tiny\ttfamily Paper4dFrontier/\allowbreak HardnessTransfer.lean}
\item \textbf{\nolinkurl{FR499}}\hypertarget{lh:FR499}{}\enspace{\ttfamily\nolinkurl{sat3SplitProxy_nonDescends_iff_satisfiable}} {\tiny\ttfamily Paper4dFrontier/\allowbreak HardnessTransfer.lean}
\item \textbf{\nolinkurl{FR500}}\hypertarget{lh:FR500}{}\enspace{\ttfamily\nolinkurl{sat3SplitProxy_descends_iff_not_satisfiable}} {\tiny\ttfamily Paper4dFrontier/\allowbreak HardnessTransfer.lean}
\item \textbf{\nolinkurl{FR501}}\hypertarget{lh:FR501}{}\enspace{\ttfamily\nolinkurl{sat3_to_splitProxyNonDescent_reduction}} {\tiny\ttfamily Paper4dFrontier/\allowbreak HardnessTransfer.lean}
\item \textbf{\nolinkurl{FR502}}\hypertarget{lh:FR502}{}\enspace{\ttfamily\nolinkurl{unsat3_to_splitProxyDescent_reduction}} {\tiny\ttfamily Paper4dFrontier/\allowbreak HardnessTransfer.lean}
\item \textbf{\nolinkurl{FR503}}\hypertarget{lh:FR503}{}\enspace{\ttfamily\nolinkurl{FinGraphInstance}} {\tiny\ttfamily Paper4dFrontier/\allowbreak HardnessTransfer.lean}
\item \textbf{\nolinkurl{FR504}}\hypertarget{lh:FR504}{}\enspace{\ttfamily\nolinkurl{TwoActionBinaryPairwiseInstance}} {\tiny\ttfamily Paper4dFrontier/\allowbreak HardnessTransfer.lean}
\item \textbf{\nolinkurl{FR505}}\hypertarget{lh:FR505}{}\enspace{\ttfamily\nolinkurl{graphToActionGapInstance_actionGapGraph_eq}} {\tiny\ttfamily Paper4dFrontier/\allowbreak HardnessTransfer.lean}
\item \textbf{\nolinkurl{FR506}}\hypertarget{lh:FR506}{}\enspace{\ttfamily\nolinkurl{actionGapGraphPredicate_graphToActionGapInstance_iff}} {\tiny\ttfamily Paper4dFrontier/\allowbreak HardnessTransfer.lean}
\item \textbf{\nolinkurl{FR507}}\hypertarget{lh:FR507}{}\enspace{\ttfamily\nolinkurl{graphPredicate_to_actionGapPredicate_reduction}} {\tiny\ttfamily Paper4dFrontier/\allowbreak HardnessTransfer.lean}
\item \textbf{\nolinkurl{FR508}}\hypertarget{lh:FR508}{}\enspace{\ttfamily\nolinkurl{graphPredicate_hardness_transfers}} {\tiny\ttfamily Paper4dFrontier/\allowbreak HardnessTransfer.lean}
\end{list}
\else
\begin{longtable}{@{}>{\raggedright\arraybackslash}m{0.10\linewidth}>{\raggedright\arraybackslash}m{0.86\linewidth}@{}}
\toprule
\textbf{ID} & \textbf{Lean Handle / Source} \\
\midrule
\endfirsthead
\toprule
\textbf{ID} & \textbf{Lean Handle / Source} \\
\midrule
\endhead
\endfoot
\bottomrule
\endlastfoot
\hypertarget{lh:FR7}{\textbf{\nolinkurl{FR7}}} & {\fontsize{8}{9}\selectfont\ttfamily\nolinkurl{exists_decisionProblem_realizing_labeling}}\par\vspace{0.1ex}{\fontsize{8}{9}\selectfont\ttfamily Paper4dFrontier/\allowbreak Realizability.lean} \\
\midrule
\hypertarget{lh:FR8}{\textbf{\nolinkurl{FR8}}} & {\fontsize{8}{9}\selectfont\ttfamily\nolinkurl{realizingProblem_decisionEquiv_iff}}\par\vspace{0.1ex}{\fontsize{8}{9}\selectfont\ttfamily Paper4dFrontier/\allowbreak Realizability.lean} \\
\midrule
\hypertarget{lh:FR13}{\textbf{\nolinkurl{FR13}}} & {\fontsize{8}{9}\selectfont\ttfamily\nolinkurl{isOptimal_positiveAffineTransform_iff}}\par\vspace{0.1ex}{\fontsize{8}{9}\selectfont\ttfamily Paper4dFrontier/\allowbreak FamilyAxioms.lean} \\
\midrule
\hypertarget{lh:FR14}{\textbf{\nolinkurl{FR14}}} & {\fontsize{8}{9}\selectfont\ttfamily\nolinkurl{opt_eq_positiveAffineTransform}}\par\vspace{0.1ex}{\fontsize{8}{9}\selectfont\ttfamily Paper4dFrontier/\allowbreak FamilyAxioms.lean} \\
\midrule
\hypertarget{lh:FR15}{\textbf{\nolinkurl{FR15}}} & {\fontsize{8}{9}\selectfont\ttfamily\nolinkurl{isSufficient_positiveAffineTransform_iff}}\par\vspace{0.1ex}{\fontsize{8}{9}\selectfont\ttfamily Paper4dFrontier/\allowbreak FamilyAxioms.lean} \\
\midrule
\hypertarget{lh:FR16}{\textbf{\nolinkurl{FR16}}} & {\fontsize{8}{9}\selectfont\ttfamily\nolinkurl{isRelevant_positiveAffineTransform_iff}}\par\vspace{0.1ex}{\fontsize{8}{9}\selectfont\ttfamily Paper4dFrontier/\allowbreak FamilyAxioms.lean} \\
\midrule
\hypertarget{lh:FR17}{\textbf{\nolinkurl{FR17}}} & {\fontsize{8}{9}\selectfont\ttfamily\nolinkurl{decisionEquiv_relabelActions_iff}}\par\vspace{0.1ex}{\fontsize{8}{9}\selectfont\ttfamily Paper4dFrontier/\allowbreak FamilyAxioms.lean} \\
\midrule
\hypertarget{lh:FR18}{\textbf{\nolinkurl{FR18}}} & {\fontsize{8}{9}\selectfont\ttfamily\nolinkurl{decisionEquiv_relabelStates_iff}}\par\vspace{0.1ex}{\fontsize{8}{9}\selectfont\ttfamily Paper4dFrontier/\allowbreak FamilyAxioms.lean} \\
\end{longtable}
\fi
\makeatother
\endgroup

}{%
  \textbf{Error:} \texttt{content/lean\_handle\_ids\_auto.tex} not found.
}